\documentclass[aps,prx,english,superscriptaddress,twocolumn,tightenlines,showpacs,longbibliography,floatfix,amsfonts]{revtex4-2}
\usepackage[T1]{fontenc}
\usepackage[utf8]{inputenc}
\usepackage[normalem]{ulem}
\usepackage{textcomp}
\usepackage{graphicx}
\usepackage{subfigure}
\usepackage[dvipsnames,svgnames]{xcolor}
\usepackage{mathrsfs,mathtools,xfrac}
\usepackage{amsmath,amssymb,amsthm,amscd,bm,bbm}
\usepackage{txfonts}
\usepackage{afterpackage,multirow,comment}
\usepackage{needspace}
\usepackage{booktabs,minibox}
\usepackage[colorlinks,bookmarks=true,bookmarksnumbered=true,bookmarksdepth=2,citecolor=blue,linkcolor=blue,urlcolor=blue,filecolor=blue]{hyperref}
\hypersetup{
  pdftitle={Superselection Rigidity of Topological Calderbank--Shor--Steane Codes: From Koszul-Complex Models to Nonsplit Extensions},
  pdfauthor={Hao Song}
}
\usepackage{tikz-cd}

\newtheorem{theorem}{Theorem}
\newtheorem{corollary}{Corollary}
\newtheorem{proposition}{Proposition}
\newtheorem{lemma}{Lemma}

\newtheorem{fact}{Fact}

\theoremstyle{remark}
\newtheorem{remark}{Remark}

\theoremstyle{definition}
\newtheorem{definition}{Definition}
\newtheorem{example}{Example}
\newtheorem{construction}{Construction}

\newcommand{\im}{\operatorname{im}}

\begin{document}

	\title{Superselection Rigidity of Topological Calderbank--Shor--Steane Codes:\texorpdfstring{\\}{ }From Koszul-Complex Models to Nonsplit Extensions}
	\author{Hao Song}
	\email{songhao@itp.ac.cn}
	\affiliation{Institute of Theoretical Physics, Chinese Academy of Sciences, Beijing 100190, China}
	\date{\today}

	\begin{abstract}

		Can topological quantum codes be classified up to finite-depth local circuits purely by their excitation structure? For translation-invariant Calderbank--Shor--Steane (CSS) codes in arbitrary spatial dimensions, we introduce \textit{excitation superselection profiles} $\mathscr E_\sigma:=\tau_{\geq1}\mathbf R\!\operatorname{Hom}_R(\overline{\operatorname{coker}\,\varphi_\sigma},R)$ to capture excitation structure, where $\varphi_\sigma$ is the stabilizer map for the $\sigma=X,Z$ Pauli sector, $R$ is the translation group ring, and the overline reverses translations. Their cohomology layers $E_\sigma^\ell:=H^\ell(\mathscr E_\sigma)\cong_R\operatorname{Ext}_R^\ell(\overline{\operatorname{coker}\,\varphi_\sigma},R)$ encode excitation types, fusion rules, and translation action, while $\mathscr E_\sigma$ retains cross-degree gluing. Pointlike, looplike, and higher-dimensional excitations are described by $E_\sigma^\ell$ at degrees $\ell=1,2,\ldots$, respectively. Our matrix-level Schanuel-lemma method proves \textit{single-degree superselection rigidity} for topological codes: if $\mathscr E_\sigma$ is supported in a single degree $\ell$, then $E_\sigma^\ell$ and $\ell$ fix the stable translation-invariant Clifford class under a finite-resolution hypothesis that is automatic for prime qudits. When $E_\sigma^\ell$ is also finite, prime-qudit CSS codes can be decoupled into toric-code stacks by Clifford circuits after coarse graining, with the translation kernel giving the exact minimal threshold. In parallel, as a central motivation and testing ground for our general framework, we develop a unified topological theory of \textit{Koszul-complex stabilizer models} (a family encompassing the toric-code hierarchy, bivariate-bicycle codes, and Haah's cubic code): for regular models, we compute $E_\sigma^\ell$, prove $\mathscr E_\sigma$ is single-degree, and express finite-torus logical modules as $\operatorname{Tor}$ groups; in the full-length case, we obtain toric-code-stack normal forms after coarse graining. We also introduce two constructions extending this framework: one realizes any admissible single-degree profile, while the other builds multi-degree extensions, including nonsplit examples. The latter yields eight 3D qubit toric-code extensions with identical layers in both Pauli sectors and trivial translation action, yet distinct size-dependent ground-state degeneracies. The seven nonsplit models require blocking to decouple. Thus, cross-degree gluing carries detectable translation-symmetry-enriched topological information invisible to individual superselection layers.
	\end{abstract}

	\maketitle
	\begingroup
	\small
	\hypersetup{linkcolor=black}
	\makeatletter
	\setlength{\parskip}{0pt}
	\@booleantrue\raggedcolumn@sw
	\raggedbottom
	\let\l@subsubsection\@gobbletwo
	\def\l@f@section{\addpenalty{\@secpenalty}\addvspace{0.4em}}
	\makeatother
	\tableofcontents
	\clearpage
	\endgroup

	\section{Introduction}
	Translation-invariant stabilizer codes sit at the vibrant intersection of fault-tolerant quantum error correction—where homological constructions power high-rate quantum low-density parity-check (qLDPC) codes—and condensed-matter physics, where they provide exact models of topological and fracton phases of matter.
	A central question bridging both fields is whether a topological quantum code can be uniquely determined, up to finite-depth local circuits and product-state ancillas, purely by its macroscopic excitation structure.
	We address this question for
	Calderbank--Shor--Steane (CSS) codes~\cite{CalderbankShor1996,Steane1996}
	on lattices of arbitrary spatial dimension.

	Locality and translation symmetry give this question a natural algebraic form.
	Local Pauli operators and stabilizers~\cite{Gottesman1997} on \(\mathbb Z_N\) qudits
	are represented by modules and matrices over the Laurent ring
	\(R=\mathbb Z_N[x_1^{\pm1},\ldots,x_D^{\pm1}]\), with monomials encoding
	lattice translations~\cite{Haah2013}. This viewpoint connects the
	classification of topological stabilizer phases with the construction and
	analysis of translation-invariant qLDPC codes. Homological
	constructions have driven major advances in qLDPC
	parameters~\cite{TillichZemor2014,BreuckmannEberhardt2021,HastingsHaahODonnell2021,
	PanteleevKalachev2022AlmostLinear,PanteleevKalachev2022}, while
	Laurent-polynomial methods
	allow topological and finite-size code data to be extracted directly from
	stabilizer presentations~\cite{Liang2024Extracting,Chen2025PRL,
	Liang2025TwistedTori}.

	Lattice translations can enrich a topological order by permuting its
	superselection sectors~\cite{Kitaev2006Anyons,Barkeshli2019}. Such an action can make the torus
	ground-state degeneracy depend on system size even when translation symmetry
	remains unbroken~\cite{Wen2003QuantumOrders,Watanabe2023GSD}. The question
	is then how much of this translation-symmetry-enriched topological (SET)
	order is fixed by excitation data. Ruba and Yang introduced
	a hierarchy of charge modules expressed through the \(\operatorname{Ext}\)
	functor and showed that their module structures retain the action of lattice
	translations~\cite{Yang2023}. They left open the broader question of how far
	these invariants determine a stabilizer code up to symplectic transformations
	(corresponding to Clifford quantum cellular automata), coarse graining, and
	stabilization. Related work provides an algorithm for extracting anyon types,
	string operators, fusion rules, topological spins, and braiding statistics
	directly from two-dimensional generalized Pauli stabilizer
	presentations~\cite{Liang2024Extracting}. For bivariate-bicycle codes, earlier
	work introduced the \textit{mobility sublattice}---the subgroup of translations
	preserving every anyon type---together with its coordinate spacings, termed
	\textit{anyon periods}, and related them to topological frustration and the
	system-size dependence of the ground-state
	degeneracy~\cite{Chen2025PRL}. This SET framework was subsequently extended
	to \(\mathbb Z_N\) stabilizer codes~\cite{He2026PRB}.

	At the circuit level, two-dimensional prime-qudit
		topological stabilizer codes are also known to reduce to toric-code stacks
		after suitable coarse graining~\cite{Bombin2012Universal,
		Bombin2014Structure,Haah2016Algebraic,Haah2021Classification}.
	Color codes on closed manifolds likewise admit explicit local-unitary
	reductions to toric-code stacks in arbitrary
	dimensions~\cite{Kubica2015Unfolding}. In three
		dimensions, complementary string--membrane diagnostics motivate the
		conjecture that stabilizer models with only fully mobile particles can be
		disentangled into toric-code stacks~\cite{Dua2019Sorting}. For
		composite-dimensional qudits, however, Pauli stabilizer models can realize
	more general Abelian topological orders, including twisted quantum
		doubles~\cite{Ellison2022TwistedQuantumDoubles}.
	Circuit-level equivalence also underlies entanglement
	renormalization~\cite{AguadoVidal2008Entanglement}, which relates
	codes at different length scales.

	Recent \(L\)-theoretic approaches address related classification
	problems for stabilizer systems. Ruba and Yang establish a bulk--boundary
	correspondence and, in two dimensions, characterize stabilizer states by
	their Abelian anyon models with gappable
	boundaries~\cite{RubaYang2025Witt}. Geiko and Shuklin classify invertible
	trivial-charge codes using relative \(L\)-theory, modulo condensation and
	stabilization~\cite{GeikoShuklin2025Invertible}, while Yang and Yu formulate
		an \(L\)-theoretic classification of fully mobile prime-qudit codes up to
		algebraic gapped interfaces and finite coarse
	graining~\cite{YangYu2026Classification}. In contrast to these Witt- or
	surgery-based classifications, we ask when excitation data determine a code
	up to stable equivalence under finite-depth local circuits, with the retained
	translation subgroup specified.

	Koszul-complex stabilizer models encompass several familiar families of
	translation-invariant CSS codes. Their diversity motivates the need for a
	unified topological classification, while the models themselves provide a
	concrete testing ground.
	Given a sequence \(\mathbf f=(f_1,\ldots,f_s)\in R^s\),
	the Koszul construction yields \(s+1\) stabilizer models
	indexed by \(j=0,\ldots,s\) (Sec.~\ref{sec:koszul-complex-codes}).
	The degree-\(j\) model labels qudits within a unit cell by the basis
	elements of \(K_j(\mathbf f)\) and uses the adjacent Koszul differentials
	as its stabilizer maps. For regular \(\mathbf f\), exactness of the
	complex and its dual gives the topological condition at every
	intermediate degree \(0<j<s\).
	The full-length sequence \(f_i=x_i-1\) recovers the toric
	code~\cite{Kitaev2003Anyons} and its higher-dimensional
	generalizations~\cite{dennis2002topological}, while the
	two-polynomial check maps of bivariate-bicycle (BB)
	codes~\cite{Bravyi2024BB} have the length-two Koszul form.

	Koszul complexes have recently found several applications in quantum error
	correction. Our earlier work used a duality between two Koszul-homology
	computations to relate anyons and logical operators in BB
	codes~\cite{Chen2025PRL}. Related algebro-geometric work realizes
	open-boundary tile codes as derived sections of a shifted Koszul complex and
	uses this description to study their logical operators and derived
	automorphisms~\cite{Breuckmann2025TileCodes}. Mian, Gwilliam, and Krastanov
	develop multivariate multicycle codes from Koszul complexes over finite
	binary quotient rings, which encode periodic boundary conditions closing
	each lattice direction into a cycle~\cite{Mian2026Multivariate}.
	Up to transpose and \(X/Z\) conventions, their middle-degree check maps
	agree with those of the corresponding periodic qubit specialization of our
	models. They use neighboring
	differentials as metachecks in their study of complete single-shot decoding.
	More broadly, such finite-periodic Koszul constructions encompass established
	qubit families, including BB codes~\cite{Bravyi2024BB}, multivariate-bicycle
	codes~\cite{Voss2025MultivariateBicycle}, trivariate-tricycle
	codes~\cite{Jacob2025Tricycle,Menon2026MagicTricycles}, generalized-bicycle
	codes~\cite{KovalevPryadko2013}, and abelian two-block group-algebra
	codes~\cite{LinPryadko2024}.

	\subsection*{Overview of our contributions}
	We develop a general homological framework for excitation structure
	alongside a unified topological theory of regular Koszul-complex stabilizer
	models. Two further constructions extend the Koszul setting to
	single-degree realization and multi-degree extensions. Our five main
	contributions are summarized below.

	\textbf{(1) Excitation superselection profiles.}
	We introduce the
	\textit{excitation superselection profile} to package excitation data for
	local-circuit classification. It is the positive-degree derived object
	(Definition~\ref{def:superselection-profile})
	\begin{equation*}
		\mathscr E_\sigma\coloneqq
		\tau_{\geq1}\mathbf R\!\operatorname{Hom}_R
		\bigl(\overline{\operatorname{coker}\varphi_\sigma},R\bigr),
	\end{equation*}
	where \(\sigma\in\{X,Z\}\), \(\varphi_\sigma\) is the corresponding
	stabilizer map, and the bar denotes spatial inversion.
	For \(\ell\geq1\), its cohomology modules
	\begin{equation*}
		E_\sigma^\ell=H^\ell(\mathscr E_\sigma)
		\cong_R\operatorname{Ext}_R^\ell
		(\overline{\operatorname{coker}\varphi_\sigma},R)
	\end{equation*}
	form superselection layers encoding excitation fusion rules and
	translation action~\cite{Yang2023}. The cohomological degree tracks the underlying
	hierarchy of stabilizer generators, local relations among them, and
	successive higher relations. The layer at degree \(\ell=1\) describes
	pointlike excitations; for finite profiles, higher-degree layers can be
	interpreted geometrically in terms of looplike and higher-dimensional
	excitations.
	Beyond these individual layers, the
	full profile retains \textit{cross-degree gluing}---the extension data
	linking superselection layers at different cohomological degrees.
	A profile is single-degree
	when at most one positive-degree layer is nonzero, including the
	vanishing profile.
	Table~\ref{tab:terminology} summarizes the algebra--physics dictionary.
	For general translation-invariant CSS codes,
	Theorem~\ref{thm:superselection} gives a support-sensitive stabilizer-map
	derivation of the \(\operatorname{Ext}\) description of superselection
	sectors.

	\textbf{(2) Single-degree rigidity and minimal coarse graining.}
	Our matrix-level Schanuel-lemma method converts stable module isomorphism
	into stable Clifford equivalence. Theorem~\ref{thm:topological-rigidity}
	shows that a single-degree profile in either Pauli sector fixes a
	topological CSS code's stable translation-invariant Clifford class through
	the layer's \(R\)-module structure and degree, provided the
	corresponding reduced Pauli module admits a finite-length free resolution.
	This hypothesis is automatic for prime qudits
	(Corollary~\ref{cor:clifford-equivalence-isolated}); the layer need not be
	finite. Here superselection rigidity means uniqueness of the circuit
	class, not of a stabilizer presentation or an implementing circuit.
	The hypothesis is substantive: Appendix~\ref{app:pm-obstruction} gives a
	one-dimensional \(\mathbb Z_{p^2}\) CSS chain with vanishing \(X\)-profile
	that is FDLU-trivial but has no stable finite-depth Clifford disentangler,
	even without translation symmetry.

	For topological prime-qudit CSS codes with a single-degree profile and
	finite nonzero layer \(E\),
	Corollary~\ref{thm:minimal-coarse-graining} fixes the exact minimal coarse
	graining to a toric-code stack: a finite-index translation subgroup can be
	retained precisely when it acts trivially on \(E\). The translation-action
	kernel is therefore the maximal retained subgroup. The equivalence uses a
	finite-depth local Clifford circuit and product-state ancillas, and the
	criterion applies in every spatial dimension and admissible degree.
	In the overlapping 2D qubit setting, it agrees with
	Wang \textit{et al.}, who additionally give a polynomial-time decoupling
	algorithm and bounds on operator spreading~\cite{Wang2026Decoupling}.
	
	\textbf{(3) Unified topological theory of regular Koszul models.}
	For regular Koszul-complex stabilizer models over the infinite-lattice
	Laurent ring, we develop a unified topological theory connecting excitation
	structure, finite-size logical information, and stable Clifford classification.
	We explicitly compute the excitation profiles and
	finite-size logical modules from a common quotient module
	\(S\coloneqq R/(\mathbf f)\), through \(\operatorname{Ext}\) and
	\(\operatorname{Tor}\), respectively.
	For a degree-\(j\) model built from a regular sequence of length \(s\),
	with \(1\leq j<s\), all positive-degree superselection layers vanish except
	\begin{equation*}
		\operatorname{Ext}_R^j
		\bigl(\overline{\operatorname{coker}\varphi_X},R\bigr)
		\cong_R S,
		\qquad
		\operatorname{Ext}_R^{s-j}
		\bigl(\overline{\operatorname{coker}\varphi_Z},R\bigr)
		\cong_R\overline S.
	\end{equation*}
	The Abelian-group structure of \(S\) gives excitation fusion, while its
	\(R\)-module structure records the lattice-translation action. On the
	\(\mathbf L\)-periodic torus, the same module determines the finite-size
	\(Z\)-logical space through
	\begin{equation*}
		\mathcal H_j^Z(\mathbf L)
		\cong_R
		\operatorname{Tor}_j^R(S,R_{\mathbf L}).
	\end{equation*}
	This reinterprets the earlier BB-code duality~\cite{Chen2025PRL} as the
	natural symmetry of \(\operatorname{Tor}\) and extends the familiar
	anyon--logical correspondence of the two-dimensional toric code to all
	degrees and spatial dimensions in this regular family. It also yields
	duality, Euler, and arithmetic finite-size relations.

	For full-length regular models, the translation action on \(S\) fixes the
	exact minimal coarse graining to a finite toric-code stack
	(Corollary~\ref{cor:koszul-multistack-rigidity}). This also holds for cyclic
	prime-power qudits because \(S\) is finite free over \(\mathbb Z_{p^m}\);
	Chinese-remainder decomposition extends the result to cyclic
	\(\mathbb Z_N\) qudits componentwise.

	\textbf{(4) Single-degree profile realizability.}
	Our first construction realizes arbitrary admissible single-degree profiles
	beyond regular Koszul models.
	Lemma~\ref{lem:topological-grade-bound}
	recasts the charge-module dimension bound of Ruba and
	Yang~\cite[Prop.~25]{Yang2023} as a necessary grade condition for every
	superselection layer in either Pauli sector.
	Theorem~\ref{thm:single-degree-realizability} proves the converse for
	single-degree profiles: a finitely generated module \(E\) is realizable in degree
	\(j\geq1\) precisely when \(\operatorname{grade}_R(E)>j\).
	Construction~\ref{constr:single-degree-realization} gives an explicit
	realization from a free resolution, extending regular Koszul models beyond
	regular-sequence quotients. Proposition~\ref{prop:companion-sector-layers}
	determines all companion \(Z\)-layers from a prescribed single-degree
	\(X\)-sector profile; these may occupy several degrees.

	\textbf{(5) Nonsplit extensions beyond degreewise classification.}
	Our second construction couples degree-one and degree-two Koszul models
	to build multi-degree extensions, including nonsplit examples
	(Construction~\ref{constr:degree-12-koszul-extension}).
	In the three-dimensional qubit toric-code
	specialization, the eight models have identical layers \(E_\sigma^\ell\) for
	\(\sigma=X,Z\) and \(\ell=1,2\), with trivial translation action, but
	distinct ground-state-degeneracy functions. Unlike familiar
	translation-induced anyon-permutation mechanisms, this size dependence
	arises despite trivial translation action on every individual layer.
	The models are therefore inequivalent under
	finite-depth local unitaries with the original translation symmetry fixed.
	Uniform \(2\times2\times2\) blocking nevertheless
	reduces all eight to the same decoupled toric-code stack. Thus forgetting
	the cross-degree gluing retained in \(\mathscr E_\sigma\) discards physically
	detectable translation SET information: the individual layers
	\(E_\sigma^\ell\) alone do not classify translation SET order for codes with
	multi-degree profiles.

	\smallskip
	The remainder of the paper develops these results as follows. We
	first review the polynomial CSS formalism
	(Sec.~\ref{sec:polynomial-formalism}) and introduce Koszul-complex
	models (Sec.~\ref{sec:koszul-complex-codes}). We then formulate the
	superselection profiles and compute them for regular Koszul models
	(Secs.~\ref{sec:ext-excitation-framework} and~\ref{sec:regular-koszul-excitations}),
	whose finite-size logical modules we express as \(\operatorname{Tor}\)
	groups (Sec.~\ref{sec:finite-size-logical}). We establish single-degree rigidity
	and its consequences for coarse graining, and extend the Koszul construction
	through single-degree realizability (Sec.~\ref{sec:structure_module}),
	before turning to nonsplit extensions and their physical implications
	(Sec.~\ref{sec:non-splitting-extensions}).
	We conclude with a discussion and outlook (Sec.~\ref{sec:postnikov-outlook}).
	The appendices provide supporting proofs and supplementary constructions.

\section{Translation-invariant CSS models and the polynomial formalism}
\label{sec:polynomial-formalism}

We begin by formulating translation-invariant CSS stabilizer models
using the Laurent polynomial formalism~\cite{Haah2013}. Consider a $\mathbb{Z}_N$-qudit CSS
model on the $D$-dimensional lattice $\mathbb{Z}^D$. Suppose that each unit
cell contains $q$ physical qudits, along with $a$ $X$-type and $b$ $Z$-type
stabilizer generators (also known as check operators).

Recall that for a $\mathbb{Z}_N$ qudit, the generalized Pauli operators $X$ and $Z$ act on the computational basis states $|z \rangle$ (for $z \in \mathbb{Z}_N =\{0, 1, \dots, N-1\}$) as
\begin{equation}
	X| z \rangle = | (z + 1) \bmod N \rangle, \qquad Z|z \rangle = \eta^z |z \rangle,
\end{equation}
where $\eta = \exp(2\pi \mathrm{i}/ N)$ is a primitive $N$-th root of unity. These operators satisfy $X^N = Z^N = I$ and $ZX = \eta XZ$.

	Throughout, \(\mathbb Z_N\) denotes \(\mathbb Z/N\mathbb Z\); in particular,
	\(\mathbb Z_p\) is a field when \(p\) is prime. We call \(p\)-level qudits
	with \(p\) prime \textit{prime qudits}. For \(m>1\),
	\(\mathbb Z_{p^m}\) should not be confused with the finite field
	\(\mathbb F_{p^m}\)~\footnote{Finite-field (Galois) qudits over
		\(k=\mathbb F_{p^m}\) define a different Pauli
		system~\cite{Ketkar2006FiniteFields}. After choosing an
		\(\mathbb F_p\)-basis, this system is on-site equivalent to a system of \(m\)
		\(\mathbb F_p\)-qudits carrying the induced \(k\)-linear structure. The Koszul,
	\(\operatorname{Ext}\), \(\operatorname{Tor}\), and rigidity constructions
	have a parallel formulation over \(k[\mathbb Z^D]\) within the
	\(k\)-linear CSS circuit category. In this category, finite-length free
	resolutions are automatic, and the nilpotent-coefficient obstruction of
	Appendix~\ref{app:pm-obstruction} is absent. We do not pursue this
	formulation here.}. Prime-power qudits always mean cyclic
	\(\mathbb Z_{p^m}\) qudits.

We algebraically encode the translation symmetry and $ \mathbb Z_N $ nature of qudits using the group ring
\begin{equation}
	R\coloneqq \mathbb Z_N[\mathbb Z^D],
\end{equation}
which can be identified with the Laurent polynomial ring
\begin{equation}
	R = \mathbb{Z}_N[x_1^{\pm1},x_2^{\pm1},\dots,x_D^{\pm1}].
\end{equation}
Here \(x_j\) corresponds to a unit translation along the \(j\)-th coordinate direction.
Generally, for  $\mathbf{v}=(v_1,\dots,v_D)\in\mathbb{Z}^D$, let
\begin{equation}
	x^{\mathbf v}
	\coloneqq
	x_1^{v_1}x_2^{v_2}\cdots x_D^{v_D},
\end{equation}
so that multiplication by \(x^{\mathbf v}\) implements translation by \(\mathbf v\).

The ring $ R $ is naturally equipped with the spatial-inversion
involution
\begin{equation}
	\overline{x^{\mathbf v}}=x^{-\mathbf v},
	\qquad
	\overline{c}=c
	\quad
	(c\in\mathbb Z_N),
\end{equation}
which extends $\mathbb Z_N $-linearly.

We write \(M\cong_A N\) for an \(A\)-linear isomorphism; for chain
complexes, the same notation denotes an \(A\)-linear chain isomorphism.

\subsection{Finite-support Pauli modules}
\label{subsec:Pauli}
Fix a reference unit cell containing $ q $ physical qudits. Recall that a free $R$-module is the analogue of a vector space, but with scalars drawn from the ring $R$. Because $R$ consists of Laurent polynomials with only finitely many nonzero terms, this algebraic structure inherently restricts us to operators with finite spatial extent. The  \textit{finitely supported} Pauli-$X$ and Pauli-$Z$ operators are thus elegantly described as two copies of a free $R$-module of rank $q$:
\begin{equation}
	P_X
	\coloneqq
	\bigoplus_{i=1}^q R e_i^X
	\cong_R R^q,
	\qquad
	P_Z
	\coloneqq
	\bigoplus_{i=1}^q R e_i^Z
	\cong_R R^q.
\end{equation}
Here, the canonical bases $\{e_i^X\}$ and $\{e_i^Z\}$ are formal labels for the single-qudit Pauli-$X$ and Pauli-$Z$ operators, respectively, on the $i$-th qudit within the reference cell. Elements of these modules naturally encode any finitely supported Pauli operator. For instance, $e_i^X + (c - x^{\mathbf v}) e_j^X\in P_X$ denotes $X_{i,\mathbf 0}\, X_{j,\mathbf 0}^{\,c}\, X_{j,\mathbf v}^{\,-1}$, where $c \in \mathbb{Z}_N$ and $X_{i,\mathbf v}$ is the Pauli-$X$ operator acting on the $i$-th qudit in cell $\mathbf v$. The analogous interpretation applies to $P_Z$.

The Pauli commutation relations give a non-degenerate $\mathbb{Z}_N$-bilinear pairing $\langle -, - \rangle_{\mathbb{Z}_N} \colon P_Z \times P_X \to \mathbb{Z}_N$, which lifts uniquely to a perfect sesquilinear pairing
\begin{equation}
	\langle -, - \rangle_R \colon P_Z \times P_X \longrightarrow R,
\end{equation}
antilinear in the first argument, $\langle r \,\xi_Z, \xi_X \rangle_R = \overline{r}\,\langle \xi_Z, \xi_X \rangle_R$, and linear in the second, $\langle \xi_Z, r\, \xi_X \rangle_R = r\,\langle \xi_Z, \xi_X \rangle_R$, such that $\operatorname{tr} \circ \langle -, - \rangle_R = \langle -, - \rangle_{\mathbb{Z}_N}$, where $\operatorname{tr}\colon R\to\mathbb{Z}_N$ extracts the constant term of Laurent polynomials. Explicitly, on the qudit-indexed bases $\{e_i^X\}$ and $\{e_i^Z\}$, we have $\langle e_i^Z, e_j^X \rangle_R = \delta_{i,j}$.

For any $R$-module $M$, we define the \textit{dagger dual}
\begin{equation}
	M^\dagger \coloneqq \operatorname{Hom}_R(\overline{M}, R),
\end{equation}
where $\overline{M}$ denotes the same abelian group with the conjugated
$R$-action $r \ast m = \overline{r}\, m$. Elements of $M^\dagger$ are thus
antilinear maps $\lambda: M \to R$, $\lambda(r\, m) = \overline{r}\, \lambda(m)$. The module structure on $M^\dagger$ is specified by the $R$-action on maps $(r \cdot \lambda)(m) = r\, \lambda(m)$.

The $R$-module isomorphisms induced by  $\langle -, - \rangle_R$, namely
\begin{equation}
	\begin{aligned}
		P_X &\xrightarrow{\sim} P_Z^\dagger,\quad \xi_X\mapsto \langle -, \xi_X \rangle_R, \quad \text{and} \\
		P_Z &\xrightarrow{\sim} P_X^\dagger,\quad \xi_Z\mapsto \overline{\langle \xi_Z, - \rangle_R},
	\end{aligned}
\end{equation}
allow us to freely identify
\begin{equation}
	P_Z = P_X^\dagger \quad \text{and} \quad P_X = P_Z^\dagger,
\end{equation}
under which $\{e_j\}$ and $\{e^j\}$ become mutually dual bases. We refer to this as the \textit{finite-support physical duality}.

\subsection{Stabilizer maps and chain complexes}

	To specify a translation-invariant CSS code, we introduce a pair of free $R$-modules to parameterize the abstract $X$- and $Z$-type stabilizers. Analogous to the Pauli modules, these are elegantly expressed as direct sums
	\begin{equation}
		G_X \coloneqq \bigoplus_{i=1}^a R \gamma_i^X \cong_R R^a, \qquad G_Z \coloneqq \bigoplus_{i=1}^b R \gamma_i^Z \cong_R R^b
	\end{equation}
	over canonical bases $\{\gamma_i^X\}_{i=1}^a$ and $\{\gamma_i^Z\}_{i=1}^b$, whose elements are formal labels for the chosen $X$- and $Z$-type stabilizer generators, respectively.

	Due to translation invariance, \textit{stabilizer generating matrices}~\cite{Gottesman1997} are promoted to \textit{$R$-linear stabilizer maps}:
	\begin{equation}
		\varphi_X: G_X \to P_X, \qquad \varphi_Z: G_Z \to P_Z.
	\end{equation}
	Over canonical bases,
	these maps are explicitly represented by polynomial matrices $\varphi_X \in M_{q \times a}(R)$ and $\varphi_Z \in M_{q \times b}(R)$, whose columns directly encode the Pauli configurations for a set of stabilizer generators.

To extract error syndromes, we must invert this perspective by mapping physical Pauli configurations back to the stabilizer spaces. In standard coding theory, this is achieved by \textit{parity-check matrices}; correspondingly, translation invariance promotes these to \textit{\(R\)-linear check maps} (often referred to as \textit{excitation maps} in condensed matter physics). Algebraically, they are obtained by taking the dual of stabilizer maps. While the full \(\mathbb{Z}_N\)-linear dual is mathematically available, the inherent locality of the code makes it physically natural to restrict our attention to the finite-support physical (dagger) duals:
\begin{equation}
	\varphi_X^\dagger: P_Z \to G_X^\dagger, \qquad \varphi_Z^\dagger: P_X \to G_Z^\dagger.
\end{equation}
Here, the canonical identifications \(P_X = P_Z^\dagger\) and \(P_Z = P_X^\dagger\) have been applied. Physically, these excitation maps evaluate the violated \(X\)- and \(Z\)-type stabilizers given a Pauli-\(Z\) or Pauli-\(X\) error configuration, respectively. In matrix form, constructing this dual simply corresponds to the conjugate transpose
\begin{equation}
	\varphi^\dagger \coloneqq \overline{\varphi}^{\mathsf{T}}, \label{eq:dual_map}
\end{equation}
where the spatial inversion acts entry-wise on $\overline{\varphi}$.
	For a literal representation in the \(\mathbb{Z}_N\) setting, a minus sign is needed before \(\varphi_X^\dagger\).
For brevity, we omit this sign because it is irrelevant to determining the allowed excitation patterns; it matters only when specifying which operator creates a given excitation configuration.

Any pair of polynomial matrices \(\varphi_X\) and \(\varphi_Z\) defines a valid CSS code if and only if they satisfy the \textit{CSS commutation condition}, which physically requires all \(X\)- and \(Z\)-type stabilizer generators to commute with each other.
Algebraically, this condition can be expressed as
\begin{equation}
	\varphi_X^\dagger\varphi_Z=0,
	\qquad
	\varphi_Z^\dagger\varphi_X=0.
\end{equation}
Since $(\varphi_Z^\dagger\varphi_X)^\dagger = \varphi_X^\dagger\varphi_Z$, either of these equations actually implies the other. Equivalently, this enforces that the sequences
\begin{equation}
	G_Z\xrightarrow{\;\varphi_Z\;}P_Z
	\xrightarrow{\;\varphi_X^\dagger\;}G_X^\dagger,
	\qquad
	G_X\xrightarrow{\;\varphi_X\;}P_X
	\xrightarrow{\;\varphi_Z^\dagger\;}G_Z^\dagger
	\label{eq:CSS-two-complexes}
\end{equation}
form well-defined chain complexes.
Here \(G_\sigma^\dagger\cong_R G_\sigma\), for \(\sigma=X,Z\), by
identifying the dagger-dual basis with the chosen generator basis.

A CSS code is termed \textit{topological} if the two chain complexes in Eq.~\eqref{eq:CSS-two-complexes} are exact at the physical Pauli modules $P_X$ and $P_Z$:
\begin{equation}
	\operatorname{im}\varphi_X = \ker\varphi_Z^\dagger, \qquad \operatorname{im}\varphi_Z = \ker\varphi_X^\dagger.
	\label{eq:topological-exactness-CSS}
\end{equation}
Because the ring $R$ consists of Laurent polynomials with finite support, Eq.~\eqref{eq:topological-exactness-CSS} constitutes an exactness statement strictly for spatially localized operators. Physically, it guarantees that every finitely supported Pauli operator commuting with all stabilizers is trivially a product of local stabilizers. Equivalently, the infinite system admits no nontrivial finite-support logical operators.

A striking consequence of this topological condition (i.e., Eq.~\eqref{eq:topological-exactness-CSS}) is that the two stabilizer sectors are tightly interlocked: fixing one uniquely determines the other as a submodule. If $\varphi_X$ is fixed, exactness forces $\operatorname{im}\varphi_Z = \ker\varphi_X^\dagger$, meaning $\varphi_Z$ acts merely as a choice of generators for the submodule $\ker\varphi_X^\dagger \subseteq P_Z$. Conversely, $\varphi_Z$ dictates the $\varphi_X$ sector via $\operatorname{im}\varphi_X = \ker\varphi_Z^\dagger$. In this sense, a topological CSS code is entirely specified by any one of $\varphi_X$ or $\varphi_Z$, provided the companion sector is engineered to enforce exactness.

\subsection{Periodic boundary conditions via tensor product}
\label{subsec:periodic-boundary-conditions}

To transition from the infinite translation-invariant model to a finite periodic code with dimensions
\begin{equation}
	\mathbf{L} = (L_1, \dots, L_D) \in \mathbb{Z}_{>0}^D,
\end{equation}
we impose periodic boundary conditions via the quotient ring
\begin{equation}
	R_{\mathbf L} \coloneqq R/(\mathbf{b}),\qquad \mathbf{b}\coloneqq (x_1^{L_1}-1,\dots,x_D^{L_D}-1)\in R^D.
	\label{eq:RL-vector}
\end{equation}
This algebraically identifies spatial translations by $L_\alpha$ lattice sites in the $ x_\alpha $-direction with the identity.

For the periodic system, all the modules and maps are obtained by a base ring change along $R \mapsto R_{\mathbf L}$. Specifically,
\begin{equation}
	\begin{aligned}
		P_X &\mapsto P_X\otimes_R R_{\mathbf L} \cong_{R_{\mathbf L}} R_{\mathbf L}^{\,q}, &\quad P_Z &\mapsto P_Z\otimes_R R_{\mathbf L} \cong_{R_{\mathbf L}} R_{\mathbf L}^{\,q}, \\
		G_X &\mapsto G_X\otimes_R R_{\mathbf L} \cong_{R_{\mathbf L}} R_{\mathbf L}^{\,a}, &\quad G_Z &\mapsto G_Z\otimes_R R_{\mathbf L} \cong_{R_{\mathbf L}} R_{\mathbf L}^{\,b}.
	\end{aligned}
	\label{eq:periodic-modules}
\end{equation}
The stabilizer maps descend to $R_{\mathbf L}$-linear maps
\begin{equation}
	\varphi_{X,Z} \mapsto \varphi_{X,Z} \otimes_R \operatorname{id}_{R_{\mathbf L}},
	\label{eq:periodic-stabilizer-maps-general}
\end{equation}
which in practice amounts to reducing the entries of $\varphi_X$ and $\varphi_Z$ modulo the boundary-condition ideal generated by $\mathbf{b}$.

Since the boundary-condition ideal is invariant under spatial inversion and the modules are finitely generated and free, the dagger dual commutes with the base ring change:
\begin{equation}
	\left(\varphi\otimes_R \operatorname{id}_{R_{\mathbf L}}\right)^\dagger = \varphi^\dagger\otimes_R \operatorname{id}_{R_{\mathbf L}}.
	\label{eq:periodic-dagger-compatibility}
\end{equation}

Applying this base ring change to the infinite-lattice models yields the finite periodic CSS complexes
\begin{equation}
	\begin{split}
		G_Z \otimes_R R_{\mathbf{L}} \xrightarrow{\;\varphi_Z \otimes_R \operatorname{id}_{R_{\mathbf{L}}}\;} &P_Z \otimes_R R_{\mathbf{L}} \xrightarrow{\;\varphi_X^\dagger \otimes_R \operatorname{id}_{R_{\mathbf{L}}}\;} G_X^\dagger \otimes_R R_{\mathbf{L}}, \\
		G_X \otimes_R R_{\mathbf{L}} \xrightarrow{\;\varphi_X \otimes_R \operatorname{id}_{R_{\mathbf{L}}}\;} &P_X \otimes_R R_{\mathbf{L}} \xrightarrow{\;\varphi_Z^\dagger \otimes_R \operatorname{id}_{R_{\mathbf{L}}}\;} G_Z^\dagger \otimes_R R_{\mathbf{L}}.
	\end{split}
	\label{eq:periodic-CSS-two-complexes}
\end{equation}
The CSS commutation condition is preserved by this quotient, ensuring Eq.~\eqref{eq:periodic-CSS-two-complexes} defines a valid finite CSS code. Since $R_{\mathbf L}$ is a free $\mathbb Z_N$-module of rank
\begin{equation}
	|\mathbf L| \coloneqq L_1L_2\cdots L_D,
\end{equation}
the resulting code encodes $n = q|\mathbf L|$ physical $\mathbb Z_N$-qudits.

The finite-size logical operators are given by the homology of these periodic complexes. Explicitly, the modules of $Z$-type and $X$-type logical operators are
\begin{equation}
	\!\!\mathcal{H}_Z(\mathbf{L}) \coloneqq \frac{\ker(\varphi_X^\dagger \otimes_R \operatorname{id}_{R_{\mathbf{L}}})}{\operatorname{im}(\varphi_Z \otimes_R \operatorname{id}_{R_{\mathbf{L}}})},
	\;
	\mathcal{H}_X(\mathbf{L}) \coloneqq \frac{\ker(\varphi_Z^\dagger \otimes_R \operatorname{id}_{R_{\mathbf{L}}})}{\operatorname{im}(\varphi_X \otimes_R \operatorname{id}_{R_{\mathbf{L}}})}.
	\label{eq:periodic-logical-modules}
\end{equation}
Crucially, even if the infinite-lattice model is topologically exact, tensoring with $R_{\mathbf L}$ generally fails to preserve exactness because $R_{\mathbf L}$ is not a flat $R$-module.

\begin{remark}[Models versus codes]
	Throughout this work, we distinguish between stabilizer \textit{models} and stabilizer \textit{codes}. A \textit{stabilizer model}---such as the toric-code model or a Koszul-complex stabilizer model---refers to the translation-invariant local stabilizer interaction pattern, or equivalently, the underlying local Hamiltonian data. A model may be analyzed on the infinite lattice or under specified boundary conditions. Conversely, a \textit{code} refers to a finite periodic realization of a stabilizer model. In this paper, all finite codes are taken with PBCs to strictly preserve translation symmetry.
\end{remark}

\section{Koszul-complex stabilizer models}
\label{sec:koszul-complex-codes}

Koszul-complex stabilizer models form an important class of
translation-invariant CSS codes and motivate the superselection and rigidity
framework developed below. We construct them by specializing the
chain-complex formalism of Sec.~\ref{sec:polynomial-formalism} to Koszul
complexes, then analyze regularity and canonical examples.
Throughout this section, we restrict our attention to infinite systems, deferring the discussion of periodic boundary conditions to Sec.~\ref{sec:finite-size-logical}.

\subsection{General construction of Koszul-complex models}
\label{subsec:general-koszul-models}

We begin by recalling the definition of a Koszul complex.

\begin{definition}[Koszul complex]
	Given a sequence of $s$ elements in the commutative ring $R$,
	\begin{equation}
		\mathbf{f} = (f_1, f_2, \dots, f_s) \in R^s,
	\end{equation}
	the corresponding Koszul complex $K_\bullet(\mathbf{f})$ is a chain complex
	\begin{equation}
		\!\!0 \to K_s(\mathbf{f}) \xrightarrow{\partial_{s}} K_{s-1}(\mathbf{f}) \xrightarrow{\partial_{s-1}} \cdots \xrightarrow{\partial_{2}} K_1(\mathbf{f}) \xrightarrow{\partial_{1}} K_0(\mathbf{f}) \to 0,
	\end{equation}
	with chain modules defined by the exterior algebra over $R$:
	\begin{equation}
		K_j(\mathbf{f}) \coloneqq \bigwedge^{j} E \cong_R R^{\binom{s}{j}},
	\end{equation}
	where $E \coloneqq R^s$. The differentials are the $R$-linear maps defined on the standard basis $\{e_1, e_2, \dots, e_s\}$ of $E$ by
	\begin{equation}
		\partial_j(e_{a_1} \wedge \cdots \wedge e_{a_j}) = \sum_{i=1}^{j} (-1)^{i-1} f_{a_i} \; e_{a_1} \wedge \cdots \wedge \widehat{e_{a_i}} \wedge \cdots \wedge e_{a_j},
	\end{equation}
	where the wide hat notation $\widehat{e_{a_i}}$ indicates that the specific element $e_{a_i}$ is omitted from the wedge product.
\end{definition}

For explicit matrix representations of the Koszul differentials for sequence lengths $s=2,3,4$, see Sec.~\ref{sec:tc_models} (Eqs.~\eqref{eq:KC2}, \eqref{eq:KC3}, and \eqref{eq:KC4}, respectively).

\begin{remark}
	As an interesting structural property, the Koszul complex of a sequence is isomorphic to the tensor product of the complexes of the individual elements. That is,
	\begin{equation}
		K_\bullet(f_1, f_2, \dots, f_s) \cong_R K_\bullet(f_1) \otimes_R K_\bullet(f_2) \otimes_R \cdots \otimes_R K_\bullet(f_s).
	\end{equation}
	For a detailed treatment of Koszul complexes, see textbooks such as Eisenbud~\cite[Chap.~17]{eisenbud1995commutative} or Weibel~\cite[Sec.~4.5]{weibel1994introduction}.
\end{remark}

By construction, these differentials satisfy $\partial_j  \partial_{j+1} = 0$, ensuring that $K_\bullet(\mathbf{f})$ is indeed a valid chain complex. This fundamental property---that the image of one differential is contained within the kernel of the next---can then be used to construct CSS stabilizer models.

\begin{definition}[Koszul-complex stabilizer model]
	\label{def:koszul-complex-model}
	Let $D$ denote the spatial dimension and $\mathbf{f} \in (\mathbb{Z}_N[\mathbb{Z}^D])^s$ be a sequence of length $s$. The \textit{degree-$j$ Koszul-complex stabilizer model}, denoted by $\mathrm{KC}_j^{(D,s)}(\mathbf{f}; \mathbb{Z}_N)$, is the CSS code defined by the adjacent Koszul-complex differentials:
	\begin{equation}
		\begin{aligned}
			\varphi_X = \partial_j^\dagger &\colon K_{j-1}^\dagger(\mathbf{f}) \to K_j^\dagger(\mathbf{f}), \\
			\varphi_Z = \partial_{j+1} &\colon K_{j+1}(\mathbf{f}) \to K_j(\mathbf{f}).
		\end{aligned}
		\label{eq:general-koszul-stabilizer-maps}
	\end{equation}
	In this model, the physical $\mathbb{Z}_N$ qudits within a reference unit cell are labeled by the basis elements of $K_j(\mathbf{f})$.
\end{definition}

At the endpoints, one stabilizer map vanishes: the \(j=0\) model has
only \(Z\)-type stabilizers, while the \(j=s\) model has only
\(X\)-type stabilizers.

The CSS commutation condition is automatic: $\varphi_Z^\dagger \varphi_X = 0$ follows directly from $\partial_j \partial_{j+1} = 0$. Thus, any sequence $\mathbf{f}$ yields commuting stabilizer models, though not necessarily topological ones.

\subsection{Regular Koszul-complex models}
\label{sec:regular-koszul-complex-models}

For the stabilizer models $\operatorname{KC}_j^{(D,s)}(\mathbf{f};\mathbb{Z}_N)$ with $0 < j < s$ to be topological, the condition expressed in Eq.~\eqref{eq:topological-exactness-CSS} requires the simultaneous exactness of the Koszul complex $K_\bullet(\mathbf{f})$ and its dual $K_\bullet^\dagger(\mathbf{f})$ at degree $j$. That is, the sequences
\begin{equation}
	\begin{split}
		&K_{j+1}(\mathbf{f})\xrightarrow{\;\partial_{j+1}\;}K_{j}(\mathbf{f})
		\xrightarrow{\;\partial_j\;}K_{j-1}(\mathbf{f}), \\
		&K_{j-1}^\dagger(\mathbf{f})\xrightarrow{\;\partial_j^\dagger\;}K_{j}^\dagger(\mathbf{f})
		\xrightarrow{\;\partial_{j+1}^\dagger\;}K_{j+1}^\dagger(\mathbf{f})
	\end{split}
	\label{eq:KC-topological-complexes}
\end{equation}
must both be exact in the middle. In homological terms, this corresponds to the vanishing of the $j$-th homology and the $j$-th cohomology:
\begin{equation}
	H_j(K_\bullet(\mathbf{f})) \coloneqq \frac{\ker \partial_j}{\operatorname{im}\partial_{j+1}} = 0
	\quad \text{and} \quad
	H^j(K_\bullet^\dagger(\mathbf{f})) \coloneqq \frac{\ker\partial_{j+1}^\dagger}{\operatorname{im} \partial_j^\dagger} = 0.
	\label{eq:KC-homology-cohomology}
\end{equation}
On the infinite lattice, this is precisely the statement that there are no local operators capable of distinguishing ground states.

The concept of a regular sequence provides a remarkably elegant way to guarantee this topological condition.

\begin{definition}[Regular sequence]
	\label{def:regular-sequence}
	A sequence $\mathbf{f}=(f_1,\dots,f_{s})$ in a commutative ring $R$ is \textit{regular} if $f_1$ is not a zero divisor in $R$ (meaning $f_1 g = 0$ implies $g = 0$ for any $g \in R$) and, for all $2\le j\le s$, the image of $f_j$ remains a non-zero divisor within the quotient ring $R/(f_1,\dots,f_{j-1})$.
\end{definition}

A fundamental theorem of commutative algebra states that if $\mathbf{f}$ is a regular sequence, its associated Koszul complex $K_\bullet(\mathbf{f})$ is exact at all nonzero degrees~\cite[Chap.~17]{eisenbud1995commutative}. Consequently,
\begin{equation}
	H_j(K_\bullet) \cong_R
	\begin{cases}
		R/(\mathbf{f}), & j = 0, \\
		0, & j \neq 0.
	\end{cases} \label{eq:koszul-homology}
\end{equation}
Because the only nonvanishing homology, located at degree zero, will be shown to govern the physics of Koszul stabilizer models, we assign it a dedicated name.

\begin{definition}[Koszul structure module]
	\label{def:koszul-structure-module}
	We define
	\begin{equation}
		S \coloneqq R/(\mathbf{f}) \label{eq:structure}
	\end{equation}
	as the \textit{Koszul structure module} associated with the regular sequence $\mathbf{f}$.
\end{definition}

Because the Koszul complex consists of free modules and is exact at all $j > 0$, it naturally provides a finite-length free resolution of $S$. That is, the augmented sequence
\begin{equation}
	0 \longrightarrow K_s(\mathbf{f}) \longrightarrow \cdots \longrightarrow K_1(\mathbf{f}) \xrightarrow{\;\partial_1\;} K_0(\mathbf{f}) \longrightarrow S \longrightarrow 0
\end{equation}
is exact. This free resolution is the algebraic origin of our subsequent $\operatorname{Ext}$ and $\operatorname{Tor}$ formalisms for studying the superselection sectors and ground state degeneracies of Koszul stabilizer models.

Furthermore, by the self-duality $K_\bullet^\dagger(\mathbf{f}) \cong_R K_{s-\bullet}(\overline{\mathbf{f}})$ of the Koszul complex and the fact that $\overline{\mathbf{f}} \coloneqq (\overline{f_1}, \dots, \overline{f_s})$ is also regular, the cohomology is similarly concentrated at the top degree:
\begin{equation}
	H^j(K_\bullet^\dagger(\mathbf{f})) \cong_R
	\begin{cases}
		\overline{S}, & j = s, \\
		0, & j < s,
	\end{cases} \label{eq:koszul-cohomology}
\end{equation}
where $\overline{S} = R/(\overline{\mathbf{f}})$.

The simultaneous vanishing of both $H_j$ and $H^j$ for intermediate degrees $0 < j < s$ guarantees that the associated model is strictly topological.

\begin{definition}[Regular Koszul stabilizer model]
	\label{def:regular_KC}
	We refer to $\operatorname{KC}_j^{(D,s)}(\mathbf{f};\mathbb{Z}_N)$ as the \textit{degree-$j$ regular Koszul stabilizer model} associated with $\mathbf{f}$ when $\mathbf{f}$ is a regular sequence.
\end{definition}

Famous topological codes that can be viewed as regular Koszul-complex models include:
\begin{itemize}
	\item Toric codes ($D \geq 2$) are Koszul-complex models generated by
	\(\mathbf{f}=(x_1-1,x_2-1,\dots,x_D-1)\)~\cite{Kitaev2003Anyons}.
	\item The color code ($D = 2$) on the honeycomb lattice is a
	Koszul-complex model generated by
	\(\mathbf{f}=(1+x_1+x_2,1+\overline{x_1}+\overline{x_2})\)
	~\cite{BombinMartinDelgado2006}.
	\item Haah's cubic code ($D = 3$) is a Koszul-complex model generated by
	\(\mathbf{f}=(1+x_1+x_2+x_3,\,1+x_1x_2+x_2x_3+x_3x_1)\)
	~\cite{Haah2011PRA}.
\end{itemize}
In the toric- and color-code examples, the defining sequences are regular
and full-length, $s=D$, and their structure modules are finite. Together,
these properties yield conventional topological order with superselection
sectors that become mobile after finite coarse graining. In Haah's cubic code,
$s<D$, and the model exhibits exotic fracton excitations. This contrast
motivates a formal distinction based on the length of the regular sequence
relative to the spatial dimension $D$, singling out the full-length cases as
a more tractable subclass.

\begin{definition}[Full-length regular sequence and model]
	\label{def:full-length-regular-sequence}
	A regular sequence $\mathbf{f} = (f_1,\dots,f_s)$ in $R$ is called \textit{full-length} if its length exactly matches the spatial dimension of the physical lattice, $s = D$. The corresponding models are referred to as \textit{full-length regular Koszul-complex models} and are denoted by
	\begin{equation}
		\mathrm{KC}_j^{(D)}(\mathbf{f};\mathbb{Z}_N) \coloneqq \mathrm{KC}_j^{(D,D)}(\mathbf{f};\mathbb{Z}_N),
	\end{equation}
	where we suppress the redundant index $s=D$.
\end{definition}

We will demonstrate in subsequent sections that a full-length sequence ensures the finiteness of the Koszul structure module $S=R/(\mathbf{f})$, making a complete topological classification possible. Before doing so, let us look at several examples explicitly to gain familiarity with the formalism.

\subsection{Canonical examples: toric-code models}
\label{sec:tc_models}

We now illustrate the Koszul-complex construction using its simplest nontrivial examples: the toric-code models in various dimensions. By choosing the regular sequence $f_i = x_i-1$, these well-known topological models immediately emerge.

\begin{definition}[Toric-code specialization]
	\label{def:toric-code-model}
	For spatial dimension $D\geq 2$ and intermediate degrees $1\le j\le D-1$, the \textit{$D$-dimensional toric-code model} $\mathrm{TC}_j^{(D)}(\mathbb{Z}_N)$ with qudits assigned to $j$-cells is given exactly by the regular Koszul-complex construction:
	\begin{equation}
		\mathrm{TC}_j^{(D)}(\mathbb{Z}_N) = \mathrm{KC}_j^{(D)}(x_1-1,\dots,x_D-1;\mathbb{Z}_N).
	\end{equation}
\end{definition}

Historically, the ``toric'' nomenclature arises from defining the model on a closed $D$-dimensional torus $T^D$. Our focus in this section, however, is on formulating the bulk interaction: we study these models on the infinite lattice $\mathbb{Z}^D$, deferring the explicit construction of finite periodic systems to Sec.~\ref{subsec:periodic-boundary-conditions}. Toric codes serve as the archetypal full-length ($s = D$) regular Koszul-complex models, naturally exhibiting nontrivial topological order alongside a direct gauge-theoretic interpretation.

\subsubsection{Two dimensions ($D=2$)}

\begin{figure}[t]
	\centering
	\includegraphics[width=1\columnwidth]{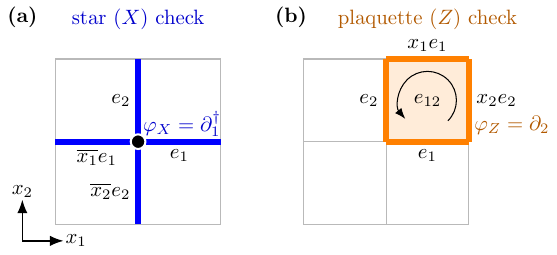}
	\caption{Toric-code stabilizers derived from the Koszul-complex construction $\mathrm{KC}_1^{(2)}(\mathbf{f};\mathbb{Z}_N)$ on the infinite square lattice with $ \mathbb{Z}_N $ qudits assigned to edges. (a) The star ($X$-type) operator, $\varphi_X = \partial_1^\dagger$, acts on four edges sharing a central vertex. (b) The plaquette ($Z$-type) operator, $\varphi_Z = \partial_2$, acts on the four boundary edges of a plaquette $e_{12}$.}
	\label{fig:toric2D}
\end{figure}

The simplest nontrivial Koszul complex is generated by a length-two sequence $\mathbf{f}=(f_1, f_2)$ in $R$. Choosing the default degree-1 basis $\{e_1, e_2\}$ and degree-2 basis $\{e_{12} \coloneqq e_1 \wedge e_2\}$, we can write the complex in the form
\begin{equation}
	0 \longrightarrow R \xrightarrow{\;\partial_2\;} R^2 \xrightarrow{\;\partial_1\;} R \longrightarrow 0,
	\label{eq:2d-toric-koszul}
\end{equation}
with boundary maps explicitly given by the matrices
\begin{equation}
	\partial_1=
	\begin{pmatrix} f_1 & f_2 \end{pmatrix}, \qquad \partial_2= \begin{pmatrix} -f_2\\ f_1
	\end{pmatrix}. \label{eq:KC2}
\end{equation}

The Koszul-complex construction $\mathrm{KC}_1^{(2)}(\mathbf{f};\mathbb{Z}_N)$ exactly reproduces the toric code on the square lattice by specializing to spatial dimension $D=2$ and setting
\begin{equation}
	f_1=x_1-1, \qquad f_2=x_2-1.
\end{equation}
The physical qudit module is $K_1(\mathbf{f})\cong_R R^2$, representing qudits assigned to the orthogonal edges $e_1$ and $e_2$. The stabilizer mappings $\varphi_X=\partial_1^\dagger$ and $\varphi_Z=\partial_2$ take the explicit matrix forms
\begin{equation}
	\varphi_X = \begin{pmatrix} \overline{x_1}-1 \\ \overline{x_2}-1 \end{pmatrix}, \qquad \varphi_Z = \begin{pmatrix} -(x_2-1) \\ x_1-1 \end{pmatrix},
\end{equation}
which precisely encode the standard star and plaquette check operators of the toric-code model, as depicted in Fig.~\ref{fig:toric2D}.

\subsubsection{Three dimensions ($D=3$)}

Moving to a length-three sequence $\mathbf{f}=(f_1, f_2, f_3)$, the general Koszul complex expands to
\begin{equation}
	0 \longrightarrow K_3 \xrightarrow{\;\partial_3\;} K_2 \xrightarrow{\;\partial_2\;} K_1 \xrightarrow{\;\partial_1\;} K_0 \longrightarrow 0,
	\label{eq:3d-toric-koszul}
\end{equation}
where the free modules are given by the exterior algebra of $R^3$. Employing the wedge shorthand $e_{ab}\coloneqq e_a\wedge e_b$ and $e_{abc}\coloneqq e_a\wedge e_b\wedge e_c$, we decompose the modules as
\begin{equation}
	\begin{aligned}
		K_0 &= R,\\
		K_1 &= R e_1 \oplus R e_2 \oplus R e_3,\\
		K_2 &= R e_{23} \oplus R e_{31} \oplus R e_{12},\\
		K_3 &= R e_{123}.
	\end{aligned}
	\label{eq:3d-koszul-bases}
\end{equation}
Note the cyclic ordering chosen for the basis of $K_2$. With respect to these ordered bases, the general Koszul differentials take the matrix forms
\begin{subequations}
	\label{eq:3d-differentials}
	\begin{align}
		\partial_1 &= \begin{pmatrix} f_1 & f_2 & f_3 \end{pmatrix}, \\[1ex]
		\partial_2 &= \begin{pmatrix}
			0    & f_3  & -f_2\\
			-f_3  & 0    & f_1\\
			f_2   & -f_1 & 0
		\end{pmatrix}, \qquad
		\partial_3 = \begin{pmatrix} f_1\\ f_2\\ f_3 \end{pmatrix}.
		\label{eq:KC3}
	\end{align}
\end{subequations}
\noindent
The nilpotency identities $\partial_1\partial_2=0$ and $\partial_2\partial_3=0$ hold automatically for any sequence $\mathbf{f}$, serving as the algebraic embodiment of the geometric principle that the boundary of a boundary is trivial.

To construct 3D toric codes, we specialize the sequence to
\begin{equation}
	f_1=x_1-1, \qquad f_2=x_2-1, \qquad f_3=x_3-1.
\end{equation}
Geometrically, the modules correspond to the cells of the lattice: $K_0$ to vertices, $K_1$ to oriented edges, $K_2$ to oriented plaquettes, and $K_3$ to cubes. In three dimensions, this complex yields two natural choices for the physical qudit module.

\begin{figure}[t]
	\centering
	\includegraphics[width=1\columnwidth]{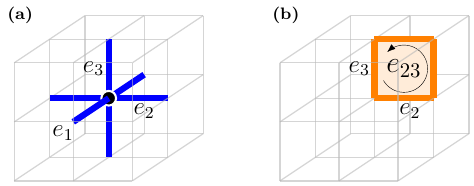}
	\caption{Stabilizer generators for the three-dimensional degree-1 toric code, $\mathrm{TC}_1^{(3)}(\mathbb{Z}_N)$, where physical qudits reside on the lattice edges. The reference edges are denoted by $e_1, e_2$, and $e_3$. (a) The star ($X$-type) check operator, defined by the adjoint map $\varphi_X = \partial_1^\dagger$, acts on the six edges incident to a central vertex. (b) The plaquette ($Z$-type) check operator, defined by the differential $\varphi_Z = \partial_2$, acts on the four edges bounding a plaquette, as illustrated for the $yz$-plane basis element $e_{23} \coloneqq e_2 \wedge e_3$.}
	\label{fig:toric3D}
\end{figure}

For the degree-$1$ model $\mathrm{TC}_1^{(3)}(\mathbb{Z}_N)$, the physical module is $K_1$, meaning the physical qudits reside on the lattice edges. By assigning the stabilizer maps as $\varphi_Z=\partial_2$ and $\varphi_X=\partial_1^\dagger$, we immediately reproduce the standard three-dimensional toric code. Explicitly, the star check operator is specified by
\begin{equation}
	\varphi_X = \partial_1^\dagger = \begin{pmatrix} \overline{f_1} \\ \overline{f_2} \\ \overline{f_3} \end{pmatrix},
\end{equation}
where $\overline{f_i} = x_i^{-1}-1$. Because each $\overline{f_i}$ carries a backward translation $x_i^{-1}$ alongside the constant term, $\varphi_X$ maps a central vertex outward to its six incident edges, defining the $X$-type check illustrated in Fig.~\ref{fig:toric3D}(a). The plaquette check operator, on the other hand, is given directly by the $3 \times 3$ matrix
\begin{equation}
	\varphi_Z = \partial_2 = \begin{pmatrix}
		0 & f_3 & -f_2\\
		-f_3 & 0 & f_1\\
		f_2 & -f_1 & 0
	\end{pmatrix}.
\end{equation}
The three columns of $\varphi_Z$ perfectly map the three fundamental plaquettes ($yz$, $zx$, and $xy$) to their respective boundaries, as depicted in Fig.~\ref{fig:toric3D}(b) for the $yz$-plaquette $e_{23}$.

For the degree-$2$ model $\mathrm{TC}_2^{(3)}(\mathbb{Z}_N)$, the physical module is $K_2$, placing the qudits on the lattice plaquettes. Our construction assigns the stabilizer maps as $\varphi_Z=\partial_3$ and $\varphi_X=\partial_2^\dagger$. Explicitly, the $Z$-type check operator is given by $\varphi_Z = \partial_3 = (f_1\; f_2 \; f_3 )^{\mathsf{T}}$ (see Eq.~\eqref{eq:KC3}), which maps a central cube to its six bounding plaquettes. Conversely, the $X$-type check is given by the adjoint of the $3 \times 3$ matrix $\partial_2$. This construction introduces no new physics relative to the degree-1 model; rather, the two are related by the electric-magnetic duality of the 3D toric code, where $\mathrm{TC}_2^{(3)}(\mathbb{Z}_N) \cong \mathrm{TC}_1^{(3)}(\mathbb{Z}_N)$ under a transformation that swaps the electric ($X$) and magnetic ($Z$) sectors while mapping the primal lattice to its dual.

This physical duality reflects the degree-reversal self-duality of the Koszul
complex. In general,
\begin{equation}
	\mathrm{KC}_j^{(D,s)}(\mathbf f;\mathbb Z_N)
	\cong
	\mathrm{KC}_{s-j}^{(D,s)}(\overline{\mathbf f};\mathbb Z_N),
\end{equation}
where the isomorphism exchanges the \(X\)- and \(Z\)-sectors. For the toric
sequence, \(\overline{f_i}=-x_i^{-1}f_i\), so this recovers
\(\mathrm{TC}_j^{(D)}\cong\mathrm{TC}_{D-j}^{(D)}\).

\subsubsection{Four dimensions ($D=4$)}

The four-dimensional scenario is theoretically illuminating because it is the first dimension where true, physically distinct types of toric codes emerge, providing exact lattice realizations of higher-form gauge theories.

For a general length-four sequence $\mathbf{f}=(f_1, f_2, f_3, f_4)$, the Koszul complex is
\begin{equation}
	0 \longrightarrow K_4 \xrightarrow{\;\partial_4\;} K_3 \xrightarrow{\;\partial_3\;} K_2 \xrightarrow{\;\partial_2\;} K_1 \xrightarrow{\;\partial_1\;} K_0 \longrightarrow 0.
	\label{eq:4d-toric-koszul}
\end{equation}
Continuing with our wedge shorthand, the relevant modules decompose as
\begin{equation}
	\begin{aligned}
		K_1 &= R e_1\oplus R e_2\oplus R e_3\oplus R e_4,\\
		K_2 &= R e_{12}\oplus R e_{13}\oplus R e_{14} \oplus R e_{23}\oplus R e_{24}\oplus R e_{34},\\
		K_3 &= R e_{123}\oplus R e_{124}\oplus R e_{134}\oplus R e_{234}.
	\end{aligned}
\end{equation}
In these ordered bases, the abstract differentials evaluate to
\begin{subequations}
	\label{eq:4d-differentials}
	\begin{align}
		\partial_1 &= \begin{pmatrix} f_1 & f_2 & f_3 & f_4 \end{pmatrix}, \\[1ex]
		\partial_2 &= \begin{pmatrix}
			-f_2 & -f_3 & -f_4 & 0    & 0    & 0\\
			f_1  & 0    & 0    & -f_3 & -f_4 & 0\\
			0    & f_1  & 0    & f_2  & 0    & -f_4\\
			0    & 0    & f_1  & 0    & f_2  & f_3
		\end{pmatrix}, \\[1ex]
		\partial_3 &= \begin{pmatrix}
			f_3  & f_4  & 0    & 0\\
			-f_2 & 0    & f_4  & 0\\
			0    & -f_2 & -f_3 & 0\\
			f_1  & 0    & 0    & f_4\\
			0    & f_1  & 0    & -f_3\\
			0    & 0    & f_1  & f_2
		\end{pmatrix}, \qquad
		\partial_4 = \begin{pmatrix}
			-f_4\\ f_3\\ -f_2\\ f_1
		\end{pmatrix}.
	\end{align} \label{eq:KC4}
\end{subequations}

We specialize to $f_j=x_j-1$ for $1\le j\le 4$ to obtain the 4D toric-code models. Unlike in three dimensions, the choice of module degree here fundamentally alters the resulting topological phase.

If we choose the physical module to be $K_1$, we obtain the degree-$1$ model $\mathrm{TC}_1^{(4)}(\mathbb{Z}_N)$, which places qudits on the lattice edges. This model realizes a $1$-form gauge theory characterized by point-like electric charges and magnetic fluxes forming closed two-dimensional surfaces on the dual lattice. Its physical dual is the degree-$3$ model $\mathrm{TC}_3^{(4)}(\mathbb{Z}_N)$ (qudits on cubes), satisfying the degree-reversal duality $\mathrm{TC}_1^{(4)}(\mathbb{Z}_N)\cong \mathrm{TC}_3^{(4)}(\mathbb{Z}_N)$ upon exchanging the electric and magnetic sectors.

In stark contrast, the most symmetric and topologically distinct physical realization occurs in the degree-$2$ model. Here, the physical module $K_2$ governs qudits living on the two-dimensional plaquettes of the 4D lattice, with stabilizer maps assigned as $\varphi_Z=\partial_3$ and $\varphi_X=\partial_2^\dagger$. The resulting code, $\mathrm{TC}_2^{(4)}(\mathbb{Z}_N)$, realizes a true $2$-form gauge theory where both the electric and magnetic excitations are extended, string-like objects. Strikingly, because $K_2$ sits exactly in the middle of the complex, the two distinct stabilizer types—and their corresponding string-like excitations—are functionally exchanged by the intrinsic degree-reversal self-duality of the Koszul complex itself.

From a quantum information perspective, the extended nature of these string-like excitations is of paramount importance. Because a logical error requires a localized string defect to grow and macroscopically span the entire lattice, the energy barrier to logical failure scales linearly with the system size. The degree-$2$ qubit ($N=2$) model is a canonical self-correcting quantum memory~\cite{dennis2002topological,brown2016quantum}. Below a nonzero critical temperature, suitable encoded observables have relaxation times exponential in the linear system size under weak-coupling Markovian thermal dynamics~\cite{Alicki2010ThermalStability}. The connection to finite-temperature topological order is discussed in Ref.~\cite{hastings2011topological}.

\subsection{Relation to bivariate-bicycle codes}

While general regular sequences allow for elaborate translation structures,
BB codes provide another important contemporary context~\cite{Bravyi2024BB}.
Their check maps always have the two-element Koszul form, but the defining
polynomial pair need not be regular.

In \(D=2\), the regular Koszul-complex models studied here form the
regular-sequence sector of BB stabilizer maps; in higher dimensions, the
construction extends the same pattern to longer Koszul complexes. For a
full-length regular sequence, Koszul exactness enforces the topological
condition, and Corollary~\ref{cor:koszul-multistack-rigidity} reduces the model,
after the translation reduction fixed by its mobility lattice, to a
toric-code stack. This conclusion relies essentially on regularity and does
not apply to a generic BB polynomial pair.

\section{Superselection profiles and the excitation hierarchy}
\label{sec:ext-excitation-framework}

A rigorous classification of topological phases requires identifying excitation types in a strictly presentation-independent manner—ensuring that our labels capture intrinsic physical data rather than artifacts of a chosen stabilizer generating set. We achieve this by formulating these excitations as homological invariants derived via free resolutions. In this section, we establish the precise correspondence between the $\operatorname{Ext}$ functor and the physical phenomena of frustration, local relations, and unbreakable relation violations.

\subsection{Physical dualities on infinite lattices}
\label{subsec:physical-dualities}

To rigorously track the spatial support of physical operators and their resulting excitations, we must distinguish between three notions of duality. Let $F$ be a finite-rank free module over $R = \mathbb{Z}_N[\mathbb{Z}^D]$.

The first is the ordinary $R$-linear dual,
\begin{equation}
	F^\vee
	\coloneqq
	\operatorname{Hom}_R(F,R),
	\label{eq:R-dual}
\end{equation}
which consists of $R$-linear maps and carries the natural left $R$-module structure $(r\cdot \lambda)(m)=r\,\lambda(m)$ for $\lambda\in F^\vee$. For finite-rank free modules, $F^\vee$ remains finite-rank free, and its elements correspond to finitely supported configurations. Because its definition requires no additional involution structure on $R$, this serves as the standard dual utilized in general homological algebra.

The second is the full configuration dual,
\begin{equation}
	F^\ast
	\coloneqq
	\operatorname{Hom}_{\mathbb Z_N}(F,\mathbb Z_N).
	\label{eq:ZN-dual}
\end{equation}
Its elements are arbitrary $\mathbb Z_N$-linear functionals, naturally allowing for infinite spatial support. It carries the contragredient $R$-action
\begin{equation}
	(r\cdot \lambda)(m)
	\coloneqq
	\lambda(\overline r\,m),
	\qquad
	\lambda \in F^\ast,
\end{equation}
where $\overline{x^{\mathbf v}}=x^{-\mathbf v}$ denotes spatial inversion. Notably, $F^\ast$ is the most physically intuitive: it directly reduces to the familiar duality when applied to finite systems. However, it proves mathematically cumbersome when tracking locality.

The third is the finite-support physical dual, or dagger dual,
\begin{equation}
	F^\dagger
	\coloneqq
	\operatorname{Hom}_R(\overline F,R),
	\label{eq:dagger-dual}
\end{equation}
which we introduced in Sec.~\ref{subsec:Pauli}. Recall that $\overline F$ denotes the same abelian group as $F$, but equipped with the twisted $R$-action $r \ast m=\overline r\,m$. Elements of $F^\dagger$ represent finite-support antilinear functionals on $F$, satisfying $\lambda(r\,m)=\overline r\,\lambda(m)$, and are equipped with the $R$-module structure
\begin{equation}
	(r\cdot \lambda )(m) \coloneqq r\, \lambda (m) =
	\lambda (\overline{r}\,m)
\end{equation}
for all $\lambda \in F^\dagger$.

While both $F^\vee$ and $F^\dagger$ embed into $F^\ast$ naturally as the same $\mathbb{Z}_N$-module via composition with the coefficient-of-identity map $\operatorname{tr}:R\to\mathbb Z_N$, only $F^\dagger$ is compatible with the contragredient $R$-action. This compatibility makes $F^\dagger$ the mathematically rigorous choice of finite-support dual for infinite-lattice systems.

By functoriality, duality naturally extends to module homomorphisms. For a given stabilizer map $d_1:F_1\to F_0$, the full $\mathbb Z_N$-linear dual is the map
\begin{equation}
	d_1^\ast:F_0^\ast\to F_1^\ast,
\end{equation}
which sends a potentially infinite Pauli configuration to its complete excitation pattern. In contrast, the dagger dual yields the finite-support excitation map
\begin{equation}
	d_1^\dagger:F_0^\dagger\to F_1^\dagger,
\end{equation}
which is formally defined as $d_1^\dagger \coloneqq \overline{d_1}^\vee$. Under the canonical embedding $F^\dagger\hookrightarrow F^\ast$, the dagger dual $d_1^\dagger$ acts exactly as the restriction of $d_1^\ast$ to finite-support configurations.

\subsection{\texorpdfstring{Pointlike excitations: the $\operatorname{Ext}_R^1$ formulation}{Pointlike excitations: the Ext formulation in degree 1}}
\label{subsec:pointlike-ext}

Armed with the dagger dual to track finite-support physical observables, we now establish the topological data of pointlike excitations. We focus our analysis on the Pauli-$Z$ module without loss of generality; the treatment of the Pauli-$X$ module follows identically. Consider a stabilizer map for a CSS code defined on a $D$-dimensional infinite lattice:
\begin{equation}
	F_1\xrightarrow{\;d_1\;}F_0.
	\label{eq:local-stabilizer-map}
\end{equation}
Here, $F_1$ and $F_0$ are finite-rank free $R$-modules, $F_0$ denotes the Pauli module of physical qudits, and the stabilizer module is given by $\operatorname{im} d_1 \subset F_0$.

If the chosen stabilizer generators have no local relations (i.e., $\ker d_1=0$), the cokernel $\operatorname{coker} d_1^\dagger$ naturally describes the module of \textit{pointlike excitation types} (namely, finite excitation patterns modulo those created by finitely supported Pauli operators). However, when the chosen stabilizer presentation contains local redundancies—or more generally, local relations among generators—this simple cokernel fails to yield an intrinsic excitation module. It becomes presentation-dependent, and redundant generators add spurious summands with no physical content.

\begin{example}[Stabilizer map with redundant generators] \label{eg:redundancy}
	Let $R=\mathbb Z_N[x^{\pm1},y^{\pm1}]$. Consider the standard stabilizer map for the toric code,
	\begin{equation}
		d_Z:R\to R^2, \qquad d_Z= \begin{pmatrix} x-1\\ y-1 \end{pmatrix},
	\end{equation}
	and compare it to the redundant presentation
	\begin{equation}
		d_Z':R^2\to R^2, \qquad d_Z'= \begin{pmatrix} x-1 & 0\\ y-1 & 0 \end{pmatrix}.
	\end{equation}
	Because both maps share the same image in $R^2$, they generate identical physical stabilizer modules. Nevertheless, evaluating their dagger adjoints—which physically represent the syndrome check maps—yields structurally distinct cokernels:
	\begin{equation}
		\operatorname{coker} (d_Z^\dagger: R^2 \to R) \cong_R \frac{R}{(x-1,y-1)},
	\end{equation}
	whereas the redundant presentation gives
	\begin{equation}
		\operatorname{coker} (d_Z^{\prime\dagger}: R^2 \to R^2) \cong_R \frac{R}{(x-1,y-1)}\oplus R.
	\end{equation}
	The additional free summand is entirely an artifact of the redundant
	stabilizer generator, representing no genuine pointlike excitation. This
	local redundancy is encoded by the relation map $d_2^\prime$:
	\begin{equation}
		R \xrightarrow{\;d_2'=\begin{pmatrix}0\\1\end{pmatrix}\;} R^2 \xrightarrow{\;d_Z'\;} R^2.
	\end{equation}
\end{example}

The $\operatorname{Ext}$ functor provides an elegant description of excitations that inherently eliminates such presentation-dependent artifacts. Consider a presentation of the stabilizer structure
\begin{equation}
	F_2 \xrightarrow{\;d_2\;} F_1 \xrightarrow{\;d_1\;} F_0,
	\label{eq:local-relation-sequence}
\end{equation}
which is exact at $F_1$, i.e., $\operatorname{im} d_2 = \ker d_1$. Here,
$d_1$ denotes the stabilizer map, while $d_2$ encodes the local relations
among its generators. (In homological algebra, relations are called
\textit{syzygies}; higher syzygies recursively encode relations among
relations~\cite[Sec.~1.10]{eisenbud1995commutative}.)

Physically, the module of \textit{superselection sectors} (or \textit{pointlike excitation types}) is defined by the quotient:
\begin{equation}
	\frac{\text{admissible finite-support excitation patterns}}{\text{locally generated patterns}}.
\end{equation}
Formally, the denominator is $\operatorname{im} d_1^\dagger$. The numerator is $F_1^\dagger\cap \operatorname{im} d_1^*$: the image $\operatorname{im} d_1^*$ consists of patterns induced by formal, possibly infinite, Pauli configurations, and intersecting with $F_1^\dagger$ enforces finite support so that the excitation pattern can be treated as a pointlike (0‑dimensional) object. This intersection is well-defined via the natural embedding of $F_1^\dagger$ into $F_1^*$. Thus,
\begin{equation}
	\frac{F_1^\dagger\cap \operatorname{im} d_1^*}{\operatorname{im} d_1^\dagger} \label{eq:def_ss}
\end{equation}
serves as the formal mathematical definition of the pointlike superselection sectors.

Haah identified pointlike topological charges with the torsion submodule of
the cokernel of the excitation map~\cite{Haah2013}. Ruba and Yang recast this
charge module as an \(\operatorname{Ext}^1\) group and extended it to a
hierarchy of higher \(\operatorname{Ext}\) charge modules~\cite{Yang2023}.
Here we give a complementary, support-sensitive CSS-map derivation:
distinguishing finite-support syndromes from unrestricted formal configurations
directly identifies the local-relation constraint with physical admissibility
and removes presentation artifacts such as the redundant generator in
Example~\ref{eg:redundancy}.

\begin{theorem}[$\operatorname{Ext}$ formulation of superselection sectors]
	\label{thm:superselection}
	Let $R=\mathbb{Z}_N[\mathbb{Z}^D]$
	and
	\begin{equation*}
		F_2\xrightarrow{\;d_2\;}F_1\xrightarrow{\;d_1\;}F_0
	\end{equation*}
	be a sequence of finite-rank free $R$-modules, exact at $F_1$. Then
	\begin{equation}
		\ker d_2^\dagger = F_1^\dagger\cap \operatorname{im} d_1^*,
		\label{eq:ker_d2}
	\end{equation}
	where the intersection is defined by the natural embedding of $F_1^\dagger$ as a subset of $F_1^*$.

	Consequently, the module of superselection sectors, defined by Eq.~\eqref{eq:def_ss}, can be expressed via the $\operatorname{Ext}$ functor:
	\begin{equation}
		\frac{F_1^\dagger\cap \operatorname{im} d_1^*}{\operatorname{im} d_1^\dagger} \cong_R \operatorname{Ext}_R^1(\overline{\operatorname{coker}d_1},R).
		\label{eq:Ext1-semi-infinite}
	\end{equation}
\end{theorem}

\begin{proof}
	Because $\mathbb{Z}_N$ is a finite Frobenius ring, it is self-injective as a $\mathbb{Z}_N$-module. Hence, the functor $\operatorname{Hom}_{\mathbb{Z}_N}(-,\mathbb{Z}_N)$ is exact. Applying this functor to the exact sequence preserves exactness at the full-configuration level:
	\begin{equation}
		\ker d_2^* = \operatorname{im} d_1^* \subset F_1^*.
		\label{eq:full-config-exactness}
	\end{equation}
	By the restriction property $d_2^\dagger = d_2^*|_{F_1^\dagger}$, we find $\ker d_2^\dagger = F_1^\dagger \cap \ker d_2^*$. Substituting Eq.~\eqref{eq:full-config-exactness} proves Eq.~\eqref{eq:ker_d2}.

	Consequently, the definition of the module of superselection sectors in Eq.~\eqref{eq:def_ss} takes the claimed form
	\begin{equation}
		\frac{F_1^\dagger\cap \operatorname{im} d_1^*}{\operatorname{im} d_1^\dagger} =
		\frac{\ker d_2^\dagger}{\operatorname{im} d_1^\dagger} =
			H^1(F_\bullet^\dagger) \cong_R \operatorname{Ext}_R^1(\overline{\operatorname{coker}d_1},R),
		\label{eq:SS_dd}
	\end{equation}
	where $H^1(F_\bullet^\dagger)$ denotes the cohomology of $F_\bullet^\dagger$ at $F_1^\dagger$.
	The final isomorphism follows from the standard resolution computation of $\operatorname{Ext}$, alongside the observation that the sequence
	\begin{equation}
			\overline{F_2}\xrightarrow{\;\overline{d_2}\;}\overline{F_1}\xrightarrow{\;\overline{d_1}\;}\overline{F_0}\rightarrow \overline{\operatorname{coker}d_1} \rightarrow 0
		\end{equation}
		provides the initial terms of a free resolution for $\overline{\operatorname{coker}d_1}$. Here, the involution (the overline) is needed to reconcile the dagger adjoint $F_\bullet^\dagger = \overline{F_\bullet}^\vee$ with the standard convention of $\operatorname{Ext}$ defined as the derived functor of the $R$-linear dual.
\end{proof}

\begin{remark}[Emergence of topological order]
	Equation~\eqref{eq:full-config-exactness} asserts that every finite-support
	syndrome satisfying the local stabilizer relations is induced by a formal
	unrestricted Pauli configuration. Such a configuration need not define a
	physical quasi-local operation. The topological obstruction is precisely the
	failure, for some admissible syndromes, to find a \textit{finite-support} Pauli
	operator that creates them.
\end{remark}

\begin{remark}[Presentation independence]
		Since conjugation is exact, $\overline{\operatorname{coker}d_1}$ is canonically determined by $\operatorname{coker} d_1 = F_0 / \operatorname{im} d_1$. Thus, the module of superselection sectors is completely and exclusively determined by the physical stabilizer module $\operatorname{im} d_1 \subset F_0$. This rigorously guarantees that our superselection modules are independent of the chosen stabilizer generators, generalizing foundational algebraic insights \cite{Haah2013, Yang2023} while providing a clean, explicit physical formulation for CSS codes.
\end{remark}

\subsection{\texorpdfstring{The $\operatorname{Ext}_R^\ell$ hierarchy of superselection sectors}{The Ext hierarchy of superselection sectors}}
\label{subsec:higher-dim}

Topological stabilizer codes in three spatial dimensions or higher may host extended excitations. We illustrate this using the three-dimensional toric code $\mathrm{TC}_1^{(3)}(\mathbb{Z}_N)$, demonstrating that looplike excitations are formally detected by the second $\operatorname{Ext}$ group.

\begin{figure}
	\includegraphics[width=1\columnwidth]{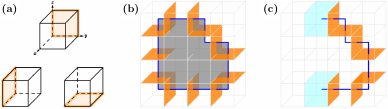}
	\caption{(a) Plaquette terms of the 3D toric code model. (b) A loop excitation created by a membrane operator. (c) The topological obstruction to cutting the loop excitation open.}
	\label{fig:TC3d}
\end{figure}

The $Z$-plaquette interaction shown in Fig.~\ref{fig:TC3d}(a) is described by the stabilizer map
\begin{equation}
	K_2\xrightarrow{\;\varphi_Z\coloneqq \partial_2\;}K_1, \qquad K_1\cong_R K_2\cong_R R^3.
\end{equation}
The subsequent Koszul map $K_3\xrightarrow{\;\partial_3\;}K_2$ encodes the local relation that the product of plaquette terms around any elementary cube is trivial. Dually, a finite plaquette-violation pattern $\varrho\in K_2^\dagger$ is physically admissible precisely when it satisfies
\begin{equation}
	\partial_3^\dagger\varrho=0.
	\label{eq:3d-loop-admissibility}
\end{equation}
Physically, these admissible plaquette-violation patterns correspond to closed magnetic flux loops, as depicted in Fig.~\ref{fig:TC3d}(b).

A closed flux loop can be created locally by applying a finite membrane operator. Therefore, its nontrivial topological nature is not that of a pointlike excitation; rather, it is characterized by the topological inability of the loop to be cut open without creating a localized violation of the underlying stabilizer relations. If we consider an open segment $\zeta\in K_2^\dagger$ of a flux loop, it ceases to be an admissible excitation (since $\partial_3^\dagger \zeta\neq 0$). The violation of the relations among stabilizers, given by $b \coloneqq \partial_3^\dagger\zeta\in K_3^\dagger$, is highly localized near the two endpoints of the open segment (see Fig.~\ref{fig:TC3d}(c)).

The topological type of such an isolated relation violation, which characterizes the obstruction to cutting the loop, is formally captured by the second $\operatorname{Ext}$ module:
\begin{equation}
	\operatorname{Ext}_R^2(\overline{\operatorname{coker}\varphi_Z},R) \cong_R \frac{K_3^\dagger}{\operatorname{im}\partial_3^\dagger} \cong_R \frac{R}{(x-1,y-1,z-1)} \cong_R \mathbb{Z}_N.
	\label{eq:3d-loop-ext2}
\end{equation}
Here \(\mathbb Z_N\) is regarded as an \(R\)-module through the augmentation
\(x_i\mapsto1\), making its translation action trivial.
Each nonzero element in this module signifies the unbreakability of the corresponding flux loop (or stated differently, its indecomposability into a union of pointlike excitations); attempting to terminate it inevitably leaves a nontrivial violation of the relations.

The homological description of pointlike and loop excitations extends through the full $\operatorname{Ext}$ hierarchy. The algebraic structure is encoded in a free resolution of $\operatorname{coker} d_1$:
\begin{equation}
	\cdots \xrightarrow{\;d_4\;} F_3 \xrightarrow{\;d_3\;} F_2 \xrightarrow{\;d_2\;} F_1 \xrightarrow{\;d_1\;} F_0 \rightarrow \operatorname{coker} d_1 \rightarrow 0 ,
	\label{eq:full-resolution}
\end{equation}
where each $F_i$ is a finite‑rank free $R$-module. The stabilizer map $d_1$
defines the first relation layer---direct relations among physical qudits. The syzygy map $d_2$ records relations among
the $d_1$-relations and defines the second relation layer; $d_3$ similarly
defines the third, and so forth.

Accordingly, the ordinary superselection sectors associated with pointlike
excitations form the \textit{first superselection layer}, which
classifies finite-support violations of the first relation layer (the
stabilizers). The
$\operatorname{Ext}_R^2$ module is the \textit{second superselection layer};
it classifies violations of the second relation layer and, in the mobile
setting, captures the obstruction associated with looplike violations of the
first relation layer that cannot be cut into pointlike pieces. More
generally, for each integer $\ell\geq1$, the \(R\)-module
\begin{equation}
	\operatorname{Ext}_R^\ell(\overline{\operatorname{coker}d_1}, R) \cong_R \frac{\ker d_{\ell+1}^\dagger}{\operatorname{im} d_\ell^\dagger} ,
	\label{eq:ext-cohomology}
\end{equation}
is called the \textit{\(\ell\)-th superselection layer}. Its elements are
equivalence classes of finite-support violations of the $\ell$-th relation
layer modulo those generated at lower relation layers. Thus the layer index is
the Ext cohomological degree, and each superselection layer arises from the
corresponding relation layer of the free resolution.

For mobile topological codes---including the finite single-degree setting
below---the \(\ell\)-th superselection layer in
Eq.~\eqref{eq:ext-cohomology} classifies $(\ell-1)$-dimensional excitations:
\begin{equation}
		\operatorname{Ext}_R^\ell(\overline{\operatorname{coker}d_1}, R) \cong_R \frac{\text{admissible $(\ell-1)$-dim excitations}}{\text{decomposable or locally creatable}}. \label{eq:Ext_classication}
\end{equation}
For $\ell\geq2$, this interpretation is inductive: cutting open an intrinsic
$(\ell-1)$-dimensional excitation produces a boundary violation one dimension
lower. Algebraically,
\begin{equation}
		\operatorname{Ext}_R^{(\ell-1)}(\overline{\operatorname{coker}d_2}, R) \cong_R \operatorname{Ext}_R^\ell(\overline{\operatorname{coker}d_1}, R).
\end{equation}
The lower degree belongs to the shifted relation hierarchy beginning with
$d_2$, rather than to a lower superselection layer of the original code.

For completeness, we summarize the full hierarchy of superselection sectors:
\begin{itemize}
	\item \textbf{0-th superselection layer (local order parameters).}
	The module
	\(
	\operatorname{Ext}_R^0
	\bigl(\overline{\operatorname{coker}\varphi_X},R\bigr)
	\cong_R\ker\varphi_X^\dagger
	\)
	consists of the finite-support \(Z\)-type operators that create no
	\(X\)-type excitations. Modulo \(Z\)-stabilizers, these give the local
	order parameters of a classical code. For a topological CSS code,
	topological exactness gives
	\(\ker\varphi_X^\dagger=\operatorname{im}\varphi_Z\), so this quotient
	vanishes.

	\item \textbf{1-st superselection layer (pointlike excitations).} The $\operatorname{Ext}_R^1$ module classifies the genuine pointlike ($0$-dimensional) superselection sectors, namely, finite excitation patterns modulo those that can be locally created.

	\item \textbf{2-nd superselection layer (looplike excitations).} The $\operatorname{Ext}_R^2$ module classifies violations of stabilizer relations, corresponding in the mobile setting to $1$-dimensional (loop or string) defects. A nontrivial class signals an obstruction to ending the loop without a pointlike violation of those relations.

	\item \textbf{\(\ell\)-th superselection layer ($(\ell-1)$-dimensional excitations).} The $\operatorname{Ext}_R^\ell$ module classifies violations of higher relations, corresponding in the mobile setting to extended defects of dimension $\ell-1$. For $\ell\geq2$, opening such an object leaves an $(\ell-2)$-dimensional boundary violation in the shifted relation hierarchy.
\end{itemize}

Although the free resolution $F_\bullet$ depends on the specific choice of stabilizer generators, $\operatorname{Ext}_R^\ell(\overline{\operatorname{coker}d_1}, R)$ are formal invariants of $\operatorname{coker}d_1$. The hierarchy thus provides a purely intrinsic, presentation‑independent topological classification of excitations.

Because we focus on topological CSS codes, for which the degree-zero quotient
above vanishes, we reserve \textit{excitation superselection profile} for the
positive-degree object defined below, abbreviated thereafter to
\textit{superselection profile}.

Up to an index shift, these \(\operatorname{Ext}\) modules are the sectorwise
CSS counterparts of the generalized charge modules of Ref.~\cite{Yang2023}.
We further package them into a derived object, which can retain cross-degree
extension data absent from the graded list alone~\cite[Chap.~10]{weibel1994introduction}.
Here we focus primarily on the rigidity question: when do the individual
\(\operatorname{Ext}\) modules uniquely determine the translation SET order at
the stable Clifford-circuit level? The nonsplit extensions in
Sec.~\ref{sec:non-splitting-extensions} show that they need not do so in the
multi-degree regime, motivating the derived packaging; a systematic treatment
of the cross-degree data is left to future work.

\begin{definition}[Excitation superselection profile]
\label{def:superselection-profile}
For \(\sigma\in\{X,Z\}\), the \textit{excitation superselection profile} of the
\(\sigma\)-sector is the positive-degree derived superselection object
\begin{equation}
	\mathscr E_\sigma
	\coloneqq
	\tau_{\geq1}\mathbf R\!\operatorname{Hom}_R
	\bigl(\overline{\operatorname{coker}\varphi_\sigma},R\bigr).
		\label{eq:superselection-profile}
\end{equation}
In simple terms, \(\mathbf R\!\operatorname{Hom}_R(M,R)\) can be viewed as the
dual of an entire free resolution of \(M\), retaining both its
\(\operatorname{Ext}\) layers and their cross-degree gluing; see
Ref.~\cite[Sec.~10.7]{weibel1994introduction} for details.
Here \(\tau_{\geq1}\) retains only cohomological degrees \(\ell\geq1\). The
\(\ell\)-th \textit{superselection layer} of \(\mathscr E_\sigma\) is the
cohomology module
\begin{equation}
	E_\sigma^\ell
	\coloneqq
	H^\ell(\mathscr E_\sigma)
	\cong_R
	\operatorname{Ext}_R^\ell
	\bigl(\overline{\operatorname{coker}\varphi_\sigma},R\bigr),
	\qquad \ell\geq 1.
	\label{eq:superselection-layers}
\end{equation}
The index \(\ell\) is the cohomological degree and records the corresponding
relation layer. A profile is
\textit{single-degree} if its cohomology is nonzero in at most one positive
degree, including the vanishing profile, and \textit{multi-degree} otherwise.
\end{definition}

Here ``layer'' refers to the relation hierarchy rather than spatial stacking.
Algebraically, each layer is an \(R\)-module: its underlying Abelian group
encodes fusion, whereas its \(R\)-module structure encodes the action of
lattice translations.

For reference, Table~\ref{tab:terminology} summarizes the algebra--physics
dictionary used throughout.
\begin{table*}[tbp]
	\centering
	\footnotesize
	\renewcommand{\arraystretch}{1.2}
	\setlength{\tabcolsep}{0pt}
	\begin{tabular}{@{}p{0.30\textwidth}@{\hspace{2pt}}p{\dimexpr0.70\textwidth-2pt\relax}@{}}
		\toprule
		\textbf{Algebraic object} & \textbf{Terminology and physical role} \\
		\midrule
		\(\operatorname{coker}\varphi_\sigma
		\coloneqq P_\sigma/\operatorname{im}\varphi_\sigma\)
		& \textbf{Reduced Pauli-\(\sigma\) module:} finite-support
		  Pauli-\(\sigma\) operators modulo \(\sigma\)-type stabilizers. \\
		\(\mathscr E_\sigma\coloneqq
		  \tau_{\geq1}\mathbf R\!\operatorname{Hom}_R
		  \bigl(\overline{\operatorname{coker}\varphi_\sigma},R\bigr)\)
		& \textbf{Excitation superselection profile:}
		  records all superselection layers \(\ell\geq1\) and their cross-degree gluing. \\
		\(E_\sigma^\ell\coloneqq H^\ell(\mathscr E_\sigma)
		\cong_R\operatorname{Ext}_R^\ell
		\bigl(\overline{\operatorname{coker}\varphi_\sigma},R\bigr)\)
		& \textbf{\(\ell\)-th superselection layer:} \(\sigma\)-excitation fusion and
		  translations---pointlike \((\ell=1)\) and extended
		  \((\ell>1)\)\textsuperscript{*}. \\
		\(S\coloneqq R/(\mathbf f) = R/(f_1,f_2,\cdots,f_s)\)
		& \textbf{Koszul structure module:} controls, by rigidity, all topological
		  properties of the regular KC model family. \\
		\(R_{\mathbf L}\coloneqq
		R/(x_1^{L_1}-1,\ldots,x_D^{L_D}-1)\)
		& \textbf{Periodic quotient ring:} imposes \(\mathbf L\)-periodic
		  boundary conditions, where \(\mathbf L=(L_1,L_2,\ldots,L_D)\). \\
		\(\mathcal H_j^Z(\mathbf L)
		  \cong_R\operatorname{Tor}_j^R(S,R_{\mathbf L})\)
		& \textbf{Finite-size \(Z\)-logical module:} describes logical
		  operators on the \(\mathbf L\)-periodic torus of a regular Koszul model. \\
		\(R_\Lambda\coloneqq\mathbb Z_N[\Lambda]\subseteq R\)
		& \textbf{Retained-translation ring:} describes codes with translation
		  symmetry restricted to \(\Lambda\subseteq\mathbb Z^D\). \\
		\(\Lambda_E \coloneqq \ker(\rho_E)\)
		& \textbf{Mobility lattice:} kernel of the translation action
		  \(\rho_E:\mathbb Z^D\to\operatorname{Aut}_{\mathbb Z_N}(E)\). \\
		\bottomrule
	\end{tabular}
	\caption{Algebra--physics dictionary for the central algebraic
		objects used throughout. Here \(\sigma\in\{X,Z\}\) labels the Pauli sector.
		\textsuperscript{*}\,For finite profiles
		(\(\lvert E_\sigma^\ell\rvert<\infty\) for all \(\ell\geq1\)), the
		\(\ell\)-th layer describes conventional \((\ell-1)\)-dimensional
		excitations (looplike for \(\ell=2\), membrane-like for \(\ell=3\), etc.).}
	\label{tab:terminology}
\end{table*}

For the finite single-degree profiles used below, an $(\ell-1)$-dimensional
excitation can be constructed from higher relation-layer violations through a
step-by-step spatial
dimension-raising process (corresponding algebraically to a cohomological
descent). A nontrivial class in $\operatorname{Ext}_R^\ell$ is represented by
a finite-support violation pattern in the $\ell$-th relation layer. Sweeping
this violation through the lattice generates a $1$-dimensional chain of
violations in the $(\ell-1)$-th relation layer. Sweeping that chain, in turn,
yields a $2$-dimensional sheet of violations in the $(\ell-2)$-th relation
layer. Continuing this procedure produces an indecomposable
$(\ell-1)$-dimensional violation of the original stabilizer relations $d_1$,
i.e., a macroscopic extended excitation. This explicit construction bridges
the abstract algebraic class in
$\operatorname{Ext}_R^\ell(\overline{\operatorname{coker}d_1},R)$ to a
concrete excitation pattern in the code. We provide the precise algorithmic
formulation of this descent for a single-degree superselection profile whose
nonzero layer is finite in Appendix~\ref{app:cohomological_descent}.

While the algebraic machinery of cohomological descent applies universally to translation-invariant stabilizer codes, the \textit{geometric} interpretation of the resulting extended excitations requires nuance in the context of fracton models. In conventional topological phases, sweeping local violations produces flexible topological strings or membranes. In high-dimensional generalizations of fracton models, however, the analogous spatial dimension-raising process often yields extended excitation patterns supported on discrete fractal geometries (e.g., Sierpinski configurations). Consequently, while a class in $\operatorname{Ext}_R^\ell$ strictly captures algebraic descent through higher relation layers, the heuristic of an ``$(\ell-1)$-dimensional extended excitation'' blurs, requiring a more careful treatment of the geometric notion of dimension. We therefore label this hierarchy by Ext degree rather than by an assumed geometric dimension.

\subsection{Connection with generalized gauging}
\label{subsec:gauging-connection}

This relation-layer hierarchy also clarifies the connection with gauging. In
generalized lattice gauge theory and subsystem gauging, gauge variables are
assigned not only to the original local terms but also to local relations
among them~\cite{Williamson2016Ungauging,Vijay2016,
ShirleySlagleChen2019Gauging}. The free resolution in
Eq.~\eqref{eq:full-resolution} provides an algebraic blueprint for iterating
this construction through higher relation layers.

Gauging converts relation-layer violations of the original model into
dynamical excitations. In the present algebraic formulation, this motivates
interpreting
$\operatorname{Ext}_R^\ell
(\overline{\operatorname{coker}d_1},R)$ as the pointlike excitation module
obtained by gauging the $(\ell-1)$-th relation layer. This viewpoint provides
a common language for cellular topological orders, fractonic stabilizer
models, and generalized subsystem gauging.

\subsection{Translation-symmetry-enriched topological order}
\label{subsec:set-order}

For a topological CSS code, the exactness condition
in Eq.~\eqref{eq:topological-exactness-CSS} makes the two stabilizer sectors
mutually determined as submodules: fixing \(\varphi_X\) fixes
\(\operatorname{im}\varphi_Z=\ker\varphi_X^\dagger\), and conversely. We
therefore use the \(X\)-sector as our chosen description. This convention
removes redundant stabilizer data; it does not identify the physically
distinct \(X\)- and \(Z\)-type excitations.

For any fixed module \(E=E_X^\ell\), translations act through
\begin{equation}
	\rho_E\colon
	\mathbb Z^D
	\longrightarrow
	\operatorname{Aut}_{\mathbb Z_N}(E),
	\qquad
	\mathbf v\longmapsto(e\mapsto x^{\mathbf v}e).
\end{equation}
Following Ref.~\cite{Chen2025PRL}, which introduced the anyon
\textit{mobility sublattice} for bivariate-bicycle codes, we define the
\textit{mobility lattice} of \(E\) by
\begin{equation}
	\Lambda_E\coloneqq\ker(\rho_E)\subseteq\mathbb Z^D.
	\label{eq:mobility}
\end{equation}
Thus \(\Lambda_E\) is the maximal translation subgroup that fixes every
class in \(E\).
For \(E=E_X^1\), \(\Lambda_E\) is the sectorwise form of that notion;
Ref.~\cite{He2026PRB} extended the mobility-group analysis to
\(\mathbb Z_N\) stabilizer codes. Equation~\eqref{eq:mobility} further
applies to every superselection layer.

Without assuming concentration in a single degree, the mobility lattice of
the \(X\)-sector superselection layers is
\begin{equation}
	\Lambda_{\mathrm{ss}}
	\coloneqq
	\bigcap_{\ell\geq 1}\Lambda_{E_X^\ell}.
\end{equation}
It is the maximal translation subgroup acting trivially on every layer of the
positive-degree \(X\)-sector superselection profile.

When the simultaneous translation action on these layers has finite image,
\(\Lambda_{\mathrm{ss}}\) has finite index in \(\mathbb Z^D\). Generalizing
the anyon periods introduced in Ref.~\cite{Chen2025PRL}, for each coordinate
direction \(a\), define the \textit{directional mobility period}
\begin{equation}
	J_a\coloneqq
	\min\bigl\{J\in\mathbb Z_{>0}\bigm|
	x_a^J e=e,\; \forall\,\ell\geq 1,\ \forall\,e\in E_X^\ell\bigr\}.
	\label{eq:coordinate-periods}
\end{equation}
The corresponding mobility-period vector
\begin{equation}
	\mathbf J=(J_1,J_2,\dots,J_D)
\end{equation}
determines the rectangular translation subgroup
\(\Lambda_{\mathrm{rect}}
=J_1\mathbb Z\oplus\cdots\oplus J_D\mathbb Z
\subseteq\Lambda_{\mathrm{ss}}\).
This inclusion can be strict, because \(\Lambda_{\mathrm{ss}}\) may also
contain skewed translation vectors.

\section{Superselection profiles of regular Koszul-complex models}
\label{sec:regular-koszul-excitations}

Applying the homological framework developed in Sec.~\ref{sec:ext-excitation-framework}, we now compute the superselection sectors for the regular Koszul-complex models of Sec.~\ref{sec:koszul-complex-codes}. The result is particularly simple: each of the \(X\)- and \(Z\)-sector profiles has a single nonzero superselection layer, given by the Koszul structure module $S \coloneqq R/(\mathbf{f})$, up to spatial inversion.

\subsection{Computation of the superselection profile}

\begin{fact}[Superselection profiles of Koszul-complex models]
	\label{fact:koszul-excitation}
	Let $\mathbf{f}=(f_1,\dots,f_s)$ be a regular sequence in $R=\mathbb{Z}_{N}[\mathbb{Z}^D]$. For $1\le j\le s-1$, the Koszul-complex model $\mathrm{KC}_j^{(D,s)}(\mathbf{f}; \mathbb{Z}_N)$ is topological. Furthermore, the only nonzero superselection layers in strictly positive degrees are
	\begin{align}
		\operatorname{Ext}_R^j(\overline{\operatorname{coker}\varphi_X}, R)
		&\cong_R S, \label{eq:X-sector-Ext-rigidity}\\
		\operatorname{Ext}_R^{s-j}(\overline{\operatorname{coker}\varphi_Z}, R)
		&\cong_R \overline{S}, \label{eq:Z-sector-Ext-rigidity}
	\end{align}
	where $\varphi_X = \partial_j^\dagger$ and $\varphi_Z = \partial_{j+1}$, with $\partial_\bullet$ denoting the Koszul differential maps.
\end{fact}

\begin{proof}
	As shown in Sec.~\ref{sec:koszul-complex-codes} (specifically, Eqs.~\eqref{eq:koszul-homology} and \eqref{eq:koszul-cohomology}), the regularity of the sequence $\mathbf{f}$ guarantees the exactness of the Koszul complex and its dual at all intermediate degrees: both $H_i(K_\bullet)$ and $H^i(K_\bullet^\dagger)$ vanish for $1 \le i \le s-1$. This dictates the absence of finite-support $Z$-logicals and $X$-logicals, confirming that $\mathrm{KC}_j^{(D,s)}(\mathbf{f}; \mathbb{Z}_N)$ is topological.

	We now proceed to compute the superselection layers. Set
	\(M_X\coloneqq\operatorname{coker}\varphi_X\) and
	\(M_Z\coloneqq\operatorname{coker}\varphi_Z\) for brevity.

	\textbf{The $\varphi_Z$ excitations.}
	Since \(M_Z=\operatorname{coker}\partial_{j+1}\), conjugating the
	corresponding tail of the Koszul resolution gives
	\[
	0 \to \overline{K_s} \longrightarrow \dots \longrightarrow \overline{K_{j+2}} \xrightarrow{\;\;\overline{\partial_{j+2}}\;\;} \overline{K_{j+1}} \xrightarrow{\;\;\overline{\partial_{j+1}}\;\;} \overline{K_j} \to \overline{M_Z} \to 0.
	\]
	Applying $\operatorname{Hom}_R(-,R)$ to this resolution: the $\ell$-th $\operatorname{Ext}$ module for $\ell>0$ corresponds to the cohomology of $K_\bullet^\dagger (\mathbf{f})$ at position $j+\ell$, which can be obtained from Eq.~\eqref{eq:koszul-cohomology}. Thus,
	\begin{equation}
		\operatorname{Ext}_R^\ell(\overline{M_Z},R)\cong_R
		\begin{cases}
			\overline{S}, & \ell=s-j,\\[2pt]
			0, & \ell\neq s-j,
		\end{cases}\qquad \ell>0.
	\end{equation}

	\textbf{The $\varphi_X$ excitations.}
	By construction, $\varphi_X = \partial_j^\dagger$. We analyze the module
	$\overline{M_X}
	\cong_R\operatorname{coker}\partial_j^\vee$.
	By the self-duality of the Koszul complex ($K_\bullet^\vee \cong_R K_{s-\bullet}$), where the map $\partial_j^\vee$ identifies with the differential $\partial_{s-j+1} : K_{s-j+1} \to K_{s-j}$, up to the standard signs, we have $\overline{M_X} \cong_R \operatorname{coker} \partial_{s-j+1}$. Again, since $\mathbf{f}$ is regular,
	$\overline{M_X}$ is thus resolved by the tail of $K_\bullet(\mathbf{f})$:
	\[
	\begin{aligned}
	0\longrightarrow K_s\longrightarrow\cdots\longrightarrow K_{s-j+2}
	&\xrightarrow{\partial_{s-j+2}}K_{s-j+1}\\
	&\xrightarrow{\partial_{s-j+1}}K_{s-j}
	\longrightarrow\overline{M_X}\longrightarrow0.
	\end{aligned}
	\]
	Applying $\operatorname{Hom}_R(-, R)$ to this resolution: the $\ell$-th $\operatorname{Ext}$ module for $\ell>0$ corresponds to $H^{s-j+\ell}(K_\bullet^\vee(\mathbf{f}))$, which is isomorphic to $H_{j-\ell}(K_\bullet(\mathbf{f}))$
	by self-duality $K_\bullet^\vee \cong_R K_{s-\bullet}$ and hence given by Eq.~\eqref{eq:koszul-homology}. Thus,
	\begin{equation}
		\operatorname{Ext}_R^\ell\!\left(\overline{M_X},R\right)
		\cong_R
		\begin{cases}
			S, & \ell =j,\\
			0, & \ell \neq j,
		\end{cases}
		\qquad \ell >0.
	\end{equation}
	This completes the proof.
\end{proof}

\begin{example}[3D toric code]
	As a concrete double-check of the above computation, consider the 3D toric code $\mathrm{TC}_1^{(3)}(\mathbb{Z}_N)$ on a cubic lattice with $\mathbb{Z}_N$ qudits on the links. The generating sequence is
	\begin{equation}
		f_i = x_i - 1 \qquad \text{for} \quad i=1,2,3.
	\end{equation}
	Accordingly, the Koszul structure module evaluates to $S \cong \mathbb{Z}_N$. Fact~\ref{fact:koszul-excitation} then yields
	\begin{equation}
		\operatorname{Ext}_R^1(\overline{\operatorname{coker}\varphi_X}, R) \cong_R \mathbb{Z}_N,
		\qquad
		\operatorname{Ext}_R^2(\overline{\operatorname{coker}\varphi_Z}, R) \cong_R \mathbb{Z}_N,
	\end{equation}
	recovering precisely the expected $\mathbb{Z}_N$-valued pointlike charge and looplike flux sectors of the 3D toric code.
\end{example}

\subsection{Topological order and symmetry enrichment}

The \(\operatorname{Ext}\) superselection modules encode excitation fusion and
lattice-translation action: their underlying Abelian groups label excitation
classes and fusion, while multiplication by monomials records translations.
Thus their \(R\)-module structures capture the fusion-and-translation part of
translation SET data; in a multi-degree profile, additional information can
reside in cross-degree gluing.

To facilitate the analysis, the qudits can be reduced to their prime-power components. Given the prime factorization
\begin{equation}
	N=\prod_\alpha p_\alpha^{m_\alpha},
\end{equation}
the Chinese Remainder Theorem yields the ring isomorphism
\begin{equation}
	\mathbb{Z}_N[x_1^{\pm1},\dots,x_D^{\pm1}]
	\cong
	\prod_\alpha
	\mathbb{Z}_{p_\alpha^{m_\alpha}}[x_1^{\pm1},\dots,x_D^{\pm1}].
	\label{eq:CRT}
\end{equation}
Consequently, an onsite relabeling of the local degrees of freedom decomposes
any stabilizer model into decoupled $\mathbb{Z}_{p_\alpha^{m_\alpha}}$
sectors, and stable translation-invariant Clifford equivalence can be checked
componentwise; see Appendix~\ref{app:crt-stabilizer-decomposition}. In
particular, for Koszul-complex models one obtains
\begin{equation}
	\mathrm{KC}_{j}^{(D,s)}(\mathbf{f};\mathbb{Z}_N)
	\cong
	\bigoplus_\alpha
	\mathrm{KC}_{j}^{(D,s)}(\mathbf{f}^{(\alpha)};\mathbb{Z}_{p_\alpha^{m_\alpha}}),
	\label{eq:CRT-Koszul-decomposition}
\end{equation}
where $\mathbf{f}^{(\alpha)}$ denotes the $\alpha^\text{th}$ prime-power component of $\mathbf{f}$. Regularity is checked componentwise under this decomposition. Hence, it suffices to analyze the $\mathbb{Z}_{p^m}$ situation without loss of generality.

\subsubsection{Fusion rules of superselection sectors}

In each prime-power sector, the Abelian group structure of the excitation modules simplifies significantly: the Koszul structure module $S$ is free over $\mathbb{Z}_{p^m}$.

\begin{lemma}[Freeness over $\mathbb{Z}_{p^m}$ for a regular sequence]
	\label{lem:free-Zpm}
	Let
	\[
	R = \mathbb{Z}_{p^m}[x_1^{\pm1},\dots,x_D^{\pm1}]
	\]
	and let $\mathbf{f}=(f_1,\dots,f_s)$ be an $R$-regular sequence. Then the Koszul structure module $S \coloneqq R/(\mathbf{f})$ is free as a $\mathbb{Z}_{p^m}$-module, possibly of infinite rank.
\end{lemma}

\begin{proof}
	Because $\mathbf{f}$ is an $R$-regular sequence, its Koszul complex $K_\bullet(\mathbf{f})$ provides a free resolution of $S$:
	\[
	0 \longrightarrow K_s \xrightarrow{\;\partial_s\;} \dots \xrightarrow{\;\partial_1\;} K_0 \longrightarrow S \longrightarrow 0.
	\]
	Each $K_i$ is a finite-rank free $R$-module. Moreover, $R$ itself is a free $\mathbb{Z}_{p^m}$-module, with a basis given by the Laurent monomials $x^{\mathbf{v}}$ for $\mathbf{v}\in\mathbb{Z}^D$. Thus, each $K_i$ is a free $\mathbb{Z}_{p^m}$-module, possibly of infinite rank. Hence, the resolution above demonstrates that $S$ has finite projective dimension as a $\mathbb{Z}_{p^m}$-module.

	The ring $\mathbb{Z}_{p^m}$ is an Artinian local Frobenius ring, and is thus self-injective. Over a self-injective ring, every module of finite projective dimension is projective. Therefore, $S$ is a projective $\mathbb{Z}_{p^m}$-module. Finally, because $\mathbb{Z}_{p^m}$ is local, every projective $\mathbb{Z}_{p^m}$-module is free. Thus, $S$ is free as a $\mathbb{Z}_{p^m}$-module.
\end{proof}

Our primary focus lies in cases where the regular sequence $\mathbf{f}=(f_1,\dots,f_s)$ achieves the full length $s=D$, exactly matching the Krull dimension of $R$. The resulting models are purely topological and devoid of genuine fractonic excitations. This absence of fractons is guaranteed by the finiteness of $S$, which can be established directly in the general $\mathbb{Z}_N$ setting.

\begin{lemma}[Finiteness for a full-length regular sequence]
	\label{lem:finite-full-length-ZN}
	Let
	\[
	R=\mathbb{Z}_N[x_1^{\pm1},\dots,x_D^{\pm1}]
	\]
	with $N>0$, and let $\mathbf{f}=(f_1,\dots,f_D)$ be a full-length $R$-regular sequence. Then the Koszul structure module $S \coloneqq R/(\mathbf{f})$ is a finite $\mathbb{Z}_N$-module.
\end{lemma}

\begin{proof}
	Assume the ideal $(\mathbf{f})$ is proper (i.e., $S \neq 0$), as the claim is trivial otherwise.

	The key observation is that quotienting by a proper regular sequence of length $D$ reduces the Krull dimension of $R$ by exactly $D$:
	\[
	\dim S = \dim R - D = D - D = 0.
	\]
	Because $R$ is Noetherian, its quotient $S$ is also Noetherian. Any zero-dimensional Noetherian ring is necessarily Artinian.

	By construction, $S$ is a finitely generated algebra over the finite ring $\mathbb{Z}_{N}$. Because $S$ is Artinian, the quotient $S/\operatorname{Jac}(S)$---where $\operatorname{Jac}(S)$ denotes the Jacobson radical---is a finite product of fields. Furthermore, because these fields are finitely generated algebras over the finite ring $\mathbb{Z}_N$, they must be finite fields. Finally, the filtration by powers of the radical, $\operatorname{Jac}(S)^k$, has finite successive quotients and must terminate in finitely many steps. Consequently, $S$ contains only a finite number of elements, making it a finite $\mathbb{Z}_{N}$-module.
\end{proof}

Combining Lemmas~\ref{lem:free-Zpm} and \ref{lem:finite-full-length-ZN} with Fact~\ref{fact:koszul-excitation} establishes the following fusion rules:

\begin{fact}[Fusion rules for full-length regular models]
	\label{fact:fusion-rules}
	For the full-length regular Koszul-complex model $\mathrm{KC}_j^{(D)}(\mathbf{f};\mathbb{Z}_{p^m})$, with $1\le j\le D-1$, the nonzero superselection layers
	\[
	\operatorname{Ext}_R^j(\overline{\operatorname{coker}\varphi_X}, R) \quad \text{and} \quad \operatorname{Ext}_R^{D-j}(\overline{\operatorname{coker}\varphi_Z}, R)
	\]
	are isomorphic as Abelian groups to the Koszul structure module $S$. Furthermore, because $S$ is both finite and free over $\mathbb{Z}_{p^m}$, we have the $\mathbb Z_{p^m}$-module isomorphism
	\begin{equation}
		S \cong_{\mathbb Z_{p^m}} \mathbb{Z}_{p^m}^Q \label{eq:S}
	\end{equation}
	for some integer $Q \ge 0$. These topological fusion rules coincide exactly with those of $Q$ copies of the $D$-dimensional toric code $\mathrm{TC}_j^{(D)}(\mathbb{Z}_{p^m})$.
\end{fact}

\subsubsection{Translation enrichment of regular models}

For a full-length regular Koszul-complex model,
Fact~\ref{fact:koszul-excitation} gives \(E_X^j\cong_R S\),
while all other positive-degree \(X\)-sector layers vanish. Thus the
construction of Sec.~\ref{subsec:set-order} specializes directly to \(S\):
\(\Lambda_{\mathrm{ss}}=\Lambda_{E_X^j}\), which we denote by
\(\Lambda_S\). Translation by \(\mathbf v\) acts on \(S\) as multiplication
by \([x^{\mathbf v}]\), so the directional mobility period \(J_a\) is the
smallest positive integer satisfying
\([x_a]^{J_a}=[1]\) in \(S\); it is finite whenever \(S\) is finite. The
\(Z\)-sector module is \(\overline S\); spatial inversion reverses the
translation action but leaves its mobility lattice and directional mobility periods
unchanged.

The rank \(Q\) of \(S\) fixes the fusion group
\(S\cong_{\mathbb Z_{p^m}}\mathbb Z_{p^m}^Q\), while multiplication by
\([x_1],\dots,[x_D]\) records the translation enrichment. A
\(\mathbb Z_{p^m}\)-basis of \(S\), the associated translation matrices,
the mobility-period vector \(\mathbf J\), and \(\Lambda_S\) can be computed algorithmically
using Gr\"obner-basis techniques~\cite{Chen2025PRL, He2026PRB}.

\section{Finite-size logical structure and translation enrichment}
\label{sec:finite-size-logical}

We have seen how translation SET order is reflected in the superselection sectors. In this section, we introduce the $\operatorname{Tor}$ formalism to demonstrate how the finite-size dependence of the logical space (or ground state degeneracy) serves as another direct physical manifestation of this symmetry enrichment.

Recall from Sec.~\ref{subsec:periodic-boundary-conditions} that a finite periodic system of size $\mathbf{L} = (L_1,\dots,L_D)$ is obtained via the base ring change
\begin{equation}
	R \longrightarrow R_{\mathbf{L}} \coloneqq R/(\mathbf{b}),
\end{equation}
where $\mathbf{b} \coloneqq (x_1^{L_1}-1,\dots,x_D^{L_D}-1)$ is the regular sequence encoding the periodic boundary conditions.

The relevant periodic Koszul complex is obtained by tensoring the infinite-lattice complex with $R_{\mathbf{L}}$, which enforces the periodic geometry by effectively replacing the base ring $R$.

\begin{definition}
	Following standard homological notation, for any sequence $\mathbf{f}$ and any $R$-module $M$, we denote the Koszul complex of $\mathbf{f}$ with coefficients in $M$ by
	\begin{equation}
		K_\bullet(\mathbf{f}; M) \coloneqq K_\bullet(\mathbf{f}) \otimes_R M.
	\end{equation}
\end{definition}

The base-changed complex for the periodic code is therefore naturally written as $K_\bullet(\mathbf{f}; R_{\mathbf{L}})$.

\begin{definition}[Periodic Koszul-complex code]
	We denote the finite periodic code associated with the Koszul-complex stabilizer model on a lattice of size $\mathbf{L}$ by $\mathrm{KC}_j^{(D,s)}(\mathbf{f};\mathbf{L})$, where $s$ is the length of $\mathbf{f}$ and $j$ denotes the specific homological degree of the model. For brevity, when $\mathbf{f}$ is clear from context, we omit it and write $\mathrm{KC}_j^{(D,s)}(\mathbf{L})$. In the full-length regular case where $s=D$, we drop the redundant index and denote the code simply by $\mathrm{KC}_j^{(D)}(\mathbf{L})$.
\end{definition}

\subsection{Tor formulation of logical module}

A CSS code has dual $X$- and $Z$-logical operators, so counting $Z$-logical operators suffices to compute the size of the logical space.
The module of such operators is the homology
\[
\mathcal{H}_j^Z(\mathbf{L})
\coloneqq
H_j\bigl(K_\bullet(\mathbf{f}; R_{\mathbf{L}}) \bigr)
=
\frac{\ker(\partial_j\otimes_R\operatorname{id}_{R_{\mathbf{L}}})}{\operatorname{im}(\partial_{j+1}\otimes_R \operatorname{id}_{R_{\mathbf{L}}})},
\]
called the \textit{$Z$-logical module} of $\mathrm{KC}_j^{(D,s)}(\mathbf{L})$.

As established in Sec.~\ref{sec:regular-koszul-complex-models}, the Koszul complex $K_\bullet(\mathbf{f})$ of a regular sequence $\mathbf{f}$ provides a free resolution of the Koszul structure module $S = R/(\mathbf{f})$.
The $\operatorname{Tor}$ functor is precisely the homology obtained by tensoring a free resolution with the second argument, hence:

\begin{fact}[Tor formulation of the logical module]
	\label{fact:tor-isomorphism}
	For a regular sequence $\mathbf{f}$,
	\begin{equation}
		\mathcal{H}_j^Z(\mathbf{L})
		\;\cong_R\;
		\operatorname{Tor}_j^R\bigl(S, R_{\mathbf{L}}\bigr) = \operatorname{Tor}_j^R\bigl(R/(\mathbf{f}), R/(\mathbf{b})\bigr),
		\label{eq:L_Tor}
	\end{equation}
	where $S \coloneqq R/(\mathbf{f})$, $R_{\mathbf{L}} \coloneqq R/(\mathbf{b})$, and $\mathbf{b} \coloneqq (x_1^{L_1}-1,\dots,x_D^{L_D}-1)$.
\end{fact}
\begin{proof}
	Because $\mathbf{f}$ is regular, the Koszul complex $K_\bullet(\mathbf{f})$ provides a free resolution of $S$. Computing the $\operatorname{Tor}$ functor via this resolution yields $\operatorname{Tor}_j^R(S,R_{\mathbf{L}}) \cong_R H_j\bigl(K_\bullet(\mathbf{f}; R_{\mathbf{L}})\bigr)$, which is precisely $\mathcal{H}_j^Z(\mathbf{L})$ by definition.
\end{proof}

\subsection{Logical-excitation correspondence}

Applying the symmetry $\operatorname{Tor}_j^R(M,N)\cong_R \operatorname{Tor}_j^R(N,M)$ immediately yields a dual expression for the logical space.

\begin{fact}[Logical-excitation correspondence]
	\label{fact:logical-excitation-correspondence}
	For a regular sequence $\mathbf{f}$,
	\begin{equation}
		\mathcal{H}_j^Z(\mathbf{L})
		\;\cong_R\;
		\operatorname{Tor}_j^R\bigl(R/(\mathbf{b}), R/(\mathbf{f})\bigr)
		\;\cong_R\;
		H_j\bigl(K_\bullet(\mathbf{b}; S)\bigr),
		\label{eq:L_Tor_sym}
	\end{equation}
	where $\mathbf{b} \coloneqq (x_1^{L_1}-1,\dots,x_D^{L_D}-1)$ and $S \coloneqq R/(\mathbf{f})$.
\end{fact}
\begin{proof}
	The first isomorphism is the natural symmetry of the Tor functor applied to Fact~\ref{fact:tor-isomorphism}. For the second, because the boundary sequence $\mathbf{b}$ is regular by construction, its Koszul complex $K_\bullet(\mathbf{b})$ provides a free resolution of $R/(\mathbf{b})$. Computing the Tor functor by tensoring this resolution with $S$ yields $H_j\bigl(K_\bullet(\mathbf{b}; S)\bigr)$.
\end{proof}

The symmetry of the Tor functor reveals a profound physical duality. By exchanging the roles of the regular sequences, we obtain an alternative computation of the $Z$-logical space that directly relates to the topological $\varphi_X$-excitations, which are labeled by
\[
\operatorname{Ext}_R^j(\overline{\operatorname{coker}\varphi_X},R) \;\cong_R\; S
\]
(i.e., the $\varphi_X$ superselection layer given in Fact~\ref{fact:koszul-excitation}).

\subsection{Relations among logical dimensions}
\label{subsec:relations-logical-modules}

Using the Chinese remainder decomposition in Eq.~\eqref{eq:CRT-Koszul-decomposition}, we restrict our analysis to $R = \mathbb{Z}_{p^m}[x_1^{\pm1},\dots,x_D^{\pm1}]$ without loss of generality throughout this subsection.

\subsubsection{Generic Koszul-complex codes}

We first derive universal relations for \textit{any} sequence $\mathbf{f}$ (not necessarily regular or full-length).
Let the length of the sequence be $s$; the target model is then $\mathrm{KC}_j^{(D,s)}$.
It is instructive to consider the entire family of $Z$-logical modules $\mathcal{H}_j^Z(\mathbf{L})$ for all degrees $0\le j\le s$ simultaneously. The extreme cases $j=0,s$ yield classical codes (repetition-code like) but still strictly obey the relations below.

The size of the logical space in degree $j$ is measured by the logical dimension
\begin{equation}
	k_j(\mathbf{L}) \;\coloneqq\; \log_{p}\bigl|\mathcal{H}_j^Z(\mathbf{L})\bigr|,
\end{equation}
which is an integer because every finite $\mathbb{Z}_{p^m}$-module has a cardinality that is a power of $p$.

\begin{fact}[Duality relation]
	For any sequence \(\mathbf{f}\) (not necessarily regular) of length $s$, any periodic size \(\mathbf{L}\), and any \(0 \le j \le s\), there is an isomorphism of \(R\)-modules:
	\begin{equation}
		\overline{\mathcal{H}_j^Z(\mathbf{L})^\ast} \;\cong_R\; \mathcal{H}_{s-j}^Z(\mathbf{L}),
		\label{eq:module-duality}
	\end{equation}
	where the overlined dual $\overline{(-)^\ast} = \overline{\operatorname{Hom}_{\mathbb{Z}_{p^m}}(-, \mathbb{Z}_{p^m})}$ is equipped with the $R$-action $(r * \lambda)(m) = \lambda(rm)$. Consequently,
	\begin{equation}
		k_j(\mathbf{L}) = k_{s-j}(\mathbf{L}).
		\label{eq:k_duality}
	\end{equation}
\end{fact}

\begin{proof}
	The ring $R_{\mathbf{L}}$ is the group ring of a finite abelian group, hence a Frobenius $\mathbb{Z}_{p^m}$-algebra. The Koszul complex $K_\bullet(\mathbf{f}; R_{\mathbf{L}})$ carries a perfect degree-reversing pairing induced by the exterior product composed with the coefficient-of-identity map $\operatorname{tr}\colon R_{\mathbf{L}}\to\mathbb{Z}_{p^m}$. This pairing descends to homology and yields the stated isomorphism.
\end{proof}

\begin{fact}[Euler relation]
	\label{fact:euler-logical}
	The logical dimensions satisfy
	\begin{equation}
		\sum_{j=0}^{s} (-1)^j k_j(\mathbf{L}) = 0 . \label{eq:k_Euler}
	\end{equation}
\end{fact}
\begin{proof}
	The chain groups are $K_j(\mathbf{f}; R_{\mathbf{L}}) \cong_{R_{\mathbf L}} R_{\mathbf{L}}^{\binom{s}{j}}$.
	Hence
	\[
	\sum_{j=0}^{s} (-1)^j \log_{p}\bigl|K_j(\mathbf{f}; R_{\mathbf{L}})\bigr|
	= \log_{p}|R_{\mathbf{L}}| \sum_{j=0}^{s} (-1)^j \binom{s}{j} = 0 .
	\]
	Because \(\log_p |\cdot|\) is additive over short exact sequences of finite modules (i.e., \(\log_p|B| = \log_p|A| + \log_p|C|\) for \(0 \to A \to B \to C \to 0\)), the Euler characteristic is preserved when passing from the chain complex to its homology. This establishes Eq.~\eqref{eq:k_Euler}.
\end{proof}

\begin{remark}
	The duality and Euler relations determine the full set
	$\{k_j\}_{j=0}^s$ from its first $\lceil s/2\rceil$ members. For
	$s=2n$, duality leaves the middle value $k_n$, which the Euler relation
	fixes; for $s=2n+1$, the first $n+1$ values determine the rest, and the
	Euler relation is then automatic.
\end{remark}

For \(s=2\), Eqs.~\eqref{eq:k_duality} and~\eqref{eq:k_Euler} give
\(k_2(\mathbf L)=k_0(\mathbf L)\) and
\(k_1(\mathbf L)=2k_0(\mathbf L)\). Our earlier study of the fractal Ising
model associated with Haah's cubic code provides an explicit three-dimensional
check of this general Euler relation in the non-full-length case; see Eq.~(20)
of Ref.~\cite{Canossa2024Exotic} and the discussion below it.

\subsubsection{Full-length regular Koszul-complex codes}
\label{subsubsec:full-length-logicals}

We now specialize to models
\(\mathrm{KC}_j^{(D)}(\mathbf{f};\mathbf{L})\)
defined by a full-length regular sequence \(\mathbf f\), so that \(s=D\).
By Eq.~\eqref{eq:S}, the Koszul structure module
\(S\coloneqq R/(\mathbf f)\) is free of finite rank \(Q\) over
\(\mathbb Z_{p^m}\), with
\begin{equation}
	S\cong_{\mathbb Z_{p^m}}(\mathbb Z_{p^m})^Q,
	\qquad
	|S|=p^{mQ}.
	\label{eq:full-length-structure-size}
\end{equation}

The dual form of Eq.~\eqref{eq:L_Tor_sym} reads
\begin{equation}
	\mathcal H_j^Z(\mathbf L)
	\cong_R H_j\bigl(K_\bullet(\mathbf b;S)\bigr).
	\label{eq:LZ-full-length-KbS}
\end{equation}
Thus the finite-size logical module is determined entirely by the action of
the periodicity operators \(\mathbf b\) on the finite module \(S\).

\begin{fact}[Uniform upper bound]
	\label{fact:logical-upper-bound}
	For a full-length regular code
	\(\mathrm{KC}_j^{(D)}(\mathbf f;\mathbf L)\), the
	logical module obeys the system-size-independent bound
	\begin{equation}
		|\mathcal H_j^Z(\mathbf L)|
		\leq
		|S|^{\binom{D}{j}},
		\quad\text{equivalently}\quad
		k_j(\mathbf L)\leq mQ\binom{D}{j}.
		\label{eq:logical-upper-bound}
	\end{equation}
\end{fact}

\begin{proof}
	The degree-\(j\) chain group of \(K_\bullet(\mathbf b;S)\) is
	\(S^{\binom{D}{j}}\). Its homology is a subquotient of this chain group,
	so it cannot contain more elements.
\end{proof}

Here the directional mobility periods defined in
Eq.~\eqref{eq:coordinate-periods} take the form
\begin{equation}
	J_a\coloneqq
	\min\bigl\{\ell\in\mathbb Z_{>0}\bigm|
	[x_a]^\ell=[1]\text{ in }S\bigr\}.
	\label{eq:regular-coordinate-periods}
\end{equation}
They are finite because \(S\) is finite and each \([x_a]\) is invertible.

\begin{fact}[GCD reduction]
	\label{fact:logical-gcd-reduction}
	Let \(g_a=\gcd(L_a,J_a)\) and
	\(\mathbf g=(g_1,\dots,g_D)\). Then
	\begin{equation}
		\mathcal H_j^Z(\mathbf L)
		\cong_R\mathcal H_j^Z(\mathbf g),
		\qquad\text{and hence}\qquad
		k_j(\mathbf L)=k_j(\mathbf g).
		\label{eq:logical-gcd-reduction}
	\end{equation}
\end{fact}

\begin{proof}
	Fix \(a\), and write \(u=[x_a]\), \(L=L_a\), \(J=J_a\), and
	\(g=g_a\). Since \(g\mid L\), the element \(u^L-[1]\) is divisible by
	\(u^g-[1]\). Conversely, choose \(r>0\) such that
	\(r(L/g)\equiv1\pmod{J/g}\). Since \(u^J=[1]\), we have
	\(u^g=(u^L)^r\), so \(u^g-[1]\) is divisible by \(u^L-[1]\).
	These elements therefore generate the same principal ideal in the finite
	subring of \(S\) generated by \(u\), and hence are associates. Applying
	this argument in every direction, the Koszul sequences for
	\(\mathbf L\) and \(\mathbf g\) differ entrywise by units. Rescaling the
	Koszul generators by these units gives a chain isomorphism and therefore
	the claimed homology isomorphism.
\end{proof}

\begin{corollary}[Periodicity]
	\label{cor:logical-periodicity}
	For each direction \(a\), \(k_j(\mathbf L)\) is \(J_a\)-periodic in
	\(L_a\):
	\begin{equation}
		k_j(L_1,\dots,L_a+J_a,\dots,L_D)
		=
		k_j(L_1,\dots,L_a,\dots,L_D).
		\label{eq:logical-coordinate-periodicity}
	\end{equation}
\end{corollary}

\begin{remark}[Toric-code saturation]
	\label{rem:toric-saturation}
	For the standard toric code,
	\(\mathbf f=(x_1-1,\dots,x_D-1)\), so
	\(S\cong\mathbb Z_{p^m}\) and every \([x_a]\) acts as the identity.
	Thus \(\mathbf b=0\) in \(S\), all Koszul differentials vanish, and
	\begin{equation}
		\mathcal H_j^Z(\mathbf L)
		\cong_R S^{\binom{D}{j}}
		\cong_{\mathbb Z_{p^m}}(\mathbb Z_{p^m})^{\binom{D}{j}},
		\qquad
		k_j(\mathbf L)=m\binom{D}{j}.
		\label{eq:toric-logical-saturation}
	\end{equation}
	This saturates the bound in Eq.~\eqref{eq:logical-upper-bound}.
\end{remark}

\begin{remark}[Topological frustration]
	\label{rem:topological-frustration}
	The finite torus can support fewer logical degrees of freedom than the
	maximum allowed by \(S\), so the bound in
	Eq.~\eqref{eq:logical-upper-bound} need not be saturated. Following
	Ref.~\cite{Chen2025PRL}, we call this deficit
	\textit{topological frustration}.
	For full-length regular Koszul models, it is a translation-enriched finite-size effect
	determined by the action of the \([x_a]\) on \(S\), with the system-size
	dependence entering through the integers \(\gcd(L_a,J_a)\).
\end{remark}

In summary, the logical modules of a full-length regular Koszul code are
bounded independently of system size and depend on \(\mathbf L\) only through
\(\mathbf g=(\gcd(L_1,J_1),\dots,\gcd(L_D,J_D))\). Their bounded periodicity
is a finite-size signature of the rigid superselection structure of these
models.

\section{Single-degree superselection rigidity and realization}
\label{sec:structure_module}

The preceding sections determined the superselection modules of regular
Koszul-complex models and showed that these modules control their finite-size
ground-state degeneracy. We now ask a stronger question: \textit{under what conditions does the superselection profile determine a translation-invariant topological CSS code up to a circuit-level transformation?}
We provide an affirmative answer for regular Koszul-complex models by developing a general rigidity-and-realizability theory for topological CSS codes with single-degree profiles.

Recall from Definition~\ref{def:superselection-profile} that the \(X\)-sector
profile \(\mathscr E_X\) is single-degree when at most one of its positive-degree
cohomology modules \(E_X^\ell\) is nonzero; this includes the vanishing profile.
Fact~\ref{fact:koszul-excitation} shows that regular Koszul-complex models have
this form.
Here, \textit{single-degree superselection rigidity} means that,
under the hypotheses of Theorem~\ref{thm:topological-rigidity},
a single-degree profile in one Pauli sector determines the code's
stable translation-invariant Clifford class.

Within the finite-resolution class, Theorems~\ref{thm:topological-rigidity}
and~\ref{thm:single-degree-realizability} provide a complete
existence-and-uniqueness theory for single-degree profiles. We first establish rigidity,
derive criteria for phase identification and minimal coarse graining, and
apply them to regular Koszul-complex models. We then extend the Koszul
construction to arbitrary admissible single-degree profiles, preparing the
transition to nonsplit extensions. We begin by establishing a code--module
dictionary that underlies both the single-degree rigidity proved here and the
multi-degree analysis in Sec.~\ref{sec:non-splitting-extensions}.

\subsection{Code--module dictionary for stable Clifford equivalence}

The standard equivalence relation for gapped quantum phases is characterized
by finite-depth local unitary (FDLU) circuits~\cite{Chen2010, Hastings2005,
Osborne2007, Verstraete2005}. Throughout, a \textit{Clifford circuit} means a
finite-depth circuit of Clifford gates with uniformly bounded geometric range.
We focus on phase identifications realizable by translation-invariant Clifford
circuits, allowing the addition or removal of product-state ancillas.

\subsubsection{Translation-invariant Clifford circuits}
\label{subsubsec:TI-clifford}

For CSS-preserving operations, the basic circuit element is the two-qudit $\operatorname{CNOT}$ gate~\footnote{For $\mathbb Z_N$ qudits, this generalized CNOT is also commonly called the controlled-addition or qudit SUM gate; for $N=2$, it reduces to the usual qubit CNOT.}
($|z_1,z_2\rangle \mapsto |z_1, (z_1+z_2) \bmod N\rangle$), whose adjoint
action on the generalized Pauli operators is
\begin{equation}
	\begin{aligned}
		X \otimes I &\mapsto X \otimes X, \qquad & I \otimes X &\mapsto I \otimes X,\\
		Z \otimes I &\mapsto Z \otimes I, \qquad & I \otimes Z &\mapsto Z^{-1} \otimes Z .
	\end{aligned}
\end{equation}
Applying this gate to control-target pairs $(e_j, x^{\mathbf u} e_i)$ simultaneously at every lattice translate---with $i\neq j$ to keep the supports disjoint---yields a translation-invariant circuit layer. Because the operation is $R$-linear, it induces a module automorphism of the physical Pauli modules. In the canonical bases $\{e_i^X\}$ and $\{e_i^Z\}$, this automorphism is represented by the elementary polynomial matrices
\begin{equation}
	W_X = I + x^{\mathbf u} E_{i,j}, \qquad W_Z = I - x^{-\mathbf u} E_{j,i}, \label{eq:Wxz}
\end{equation}
where $E_{i,j}$ denotes the matrix with a single $1$ in the $(i,j)$ entry and zeros elsewhere. These matrices automatically satisfy the symplectic compatibility condition $W_X^\dagger W_Z = I$, which encodes preservation of the Pauli commutation relations.

The constraint  $W_Z = (W_X^\dagger)^{-1}$ implies that a general finite-depth CNOT circuit on $P_X \oplus P_Z$ is completely determined by its $X$-sector action, described by the \textbf{elementary linear group}
\begin{equation}
	\operatorname{E}_q(R) \coloneqq \bigl\langle\, I + rE_{i,j} \mid r\in R,\; i\neq j \,\bigr\rangle,
\end{equation}
which is generated by matrices $W_X$ of the form in Eq.~\eqref{eq:Wxz}, and hence includes matrices $I+rE_{i,j}$ as well as their finite products.

Mathematically, for any fixed $q$, the elementary group $\operatorname{E}_q(R)$ is only a \textit{proper} subgroup of the general linear group $\operatorname{GL}_q(R)$.
Furthermore, \(\operatorname{E}_2(R)\) need not be invariant under arbitrary changes of basis. This presents an apparent algebraic hurdle: describing Clifford circuits requires this rigid subgroup, whereas standard mathematical classification allows arbitrary invertible \(R\)-linear changes of coordinates, encoded by $\operatorname{GL}_q(R)$.

Fortunately, we do not need to distinguish between these two groups when classifying topological phases. As we establish in the following section, we are rescued by the physical freedom to add unentangled ancillas. Through the algebraic mechanism of \textit{stabilization}, the basis-dependent restrictions of finite-depth circuits dissolve entirely, safely allowing us to utilize basis-independent mathematics.

\subsubsection{Stable equivalence of stabilizer maps}

In practice, a topological CSS code is completely specified by an $X$-sector stabilizer map
\begin{equation}
	\varphi_X \colon R^a \to R^q,
\end{equation}
acting between coordinate modules defined with respect to the canonical bases of the stabilizer and Pauli modules. The $Z$-sector is completely determined by the topological condition $\im \varphi_Z = \ker \varphi_X^\dagger$. Explicitly, $\varphi_X$ is represented by a polynomial matrix $\varphi_X \in M_{q \times a}(R)$, whose columns directly encode the Pauli configurations for a set of stabilizer generators.

The problem of classifying topological CSS codes therefore reduces to classifying matrices over $R$.

Standard topological equivalence allows adiabatic evolution as well as adding or removing trivial degrees of freedom. In terms of stabilizer matrices, these physically allowed operations act as the following matrix transformations:
\begin{enumerate}
	\item \textbf{Finite-depth Clifford circuit:} Applying a translation-invariant finite-depth CNOT circuit yields an automorphism $W_X \in \operatorname{E}_q(R)$, transforming the matrix as $\varphi_X \mapsto W_X \varphi_X$.
	\item \textbf{Stabilizer generator relabeling:} Applying an isomorphism $\Phi \in \operatorname{GL}_a(R)$ corresponds to a change of basis for the stabilizer generators, transforming the matrix as $\varphi_X \mapsto \varphi_X \Phi$.
	\item \textbf{Trivial generator addition:} Adding trivial generators $0_{0,\eta}\colon R^\eta \to 0$ leaves the stabilizer module unchanged, yielding $\varphi_X \mapsto \varphi_X \oplus 0_{0,\eta}$.
	\item \textbf{Trivial code addition ($Z$-polarized):} Stacking the lattice with unentangled ancilla qudits polarized in the $Z$-basis (which possess trivial $Z$-stabilizers but no $X$-stabilizers) adds a trivial block $0_{\beta,0}\colon 0 \to R^{\beta}$, yielding $\varphi_X \mapsto \varphi_X \oplus 0_{\beta,0}$.
	\item \textbf{Trivial code addition ($X$-polarized):} Stacking the lattice with unentangled ancilla qudits polarized in the $X$-basis adds a trivial block $I_{\alpha}\colon R^{\alpha} \to R^{\alpha}$, yielding $\varphi_X \mapsto \varphi_X \oplus I_{\alpha}$.
\end{enumerate}

Because translation invariance (encoded by $R$-linearity) is preserved during all of these matrix deformations, this framework classifies translation-invariant topological codes.

Formalizing these physical transformations leads to algebraic notions of equivalence for \(R\)-linear maps. First, restricting to operations 1 and 2---which preserve the number of physical qudits and stabilizer generators---yields the following equivalence.

\begin{definition}[Clifford \(R\)-equivalence]
	Let $\varphi, \varphi' \in M_{q \times a}(R)$. We call $\varphi$ and $\varphi'$ \textit{Clifford \(R\)-equivalent} and write $\varphi \sim_{R,\mathrm{Cl}} \varphi'$ if there exist
	\begin{equation}
		W \in \operatorname{E}_q(R), \qquad \Phi \in \operatorname{GL}_a(R)
	\end{equation}
	such that $W \varphi \Phi = \varphi'$.
\end{definition}

The ring subscript records the translation symmetry retained by the
comparison. Clifford \(R\)-equivalence restricts the left transformation to
\(\operatorname{E}_q(R)\), whereas \(R\)-equivalence, defined next, allows
\(\operatorname{GL}_q(R)\).

\begin{definition}[\(R\)-equivalence]
	Let $\varphi, \varphi' \in M_{q \times a}(R)$. We call $\varphi$ and $\varphi'$ \textit{\(R\)-equivalent}, or \textit{equivalent over \(R\)}, and write $\varphi \sim_R \varphi'$ if there exist
	\begin{equation}
		\Xi \in \operatorname{GL}_q(R), \qquad \Phi \in \operatorname{GL}_a(R)
	\end{equation}
	such that $\Xi \varphi \Phi = \varphi'$.
\end{definition}

By incorporating the addition or removal of trivial generators and unentangled qudits (Operations 3, 4, and 5), we can extend these relations to matrices of different dimensions. Algebraically, this corresponds to taking direct sums with matrix blocks of the form $I_\alpha \oplus 0_{\beta, \eta}$ (where $0_{\beta, \eta}$ is the $\beta \times \eta$ zero matrix), representing trivial maps between free modules.

\begin{definition}[Free stable Clifford \(R\)-equivalence for matrices]
	Let $\varphi \in M_{q \times a}(R)$ and $\varphi' \in M_{q' \times a'}(R)$. We say that $\varphi$ and $\varphi'$ are \textit{freely stably Clifford \(R\)-equivalent}, and write $\varphi \sim_{R,\mathrm{Cl},\mathrm{st}} \varphi'$, if there exist non-negative integers $\alpha, \beta, \eta$ and $\alpha', \beta', \eta'$ such that
	\begin{equation}
		\varphi \oplus I_{\alpha} \oplus 0_{\beta, \eta} \sim_{R,\mathrm{Cl}} \varphi' \oplus I_{\alpha'} \oplus 0_{\beta', \eta'}.
	\end{equation}
\end{definition}

\begin{definition}[Free stable \(R\)-equivalence for matrices]
	Similarly, two matrices $\varphi$ and $\varphi'$ are \textit{freely stably \(R\)-equivalent}, denoted $\varphi \sim_{R,\mathrm{st}} \varphi'$, if there exist trivial blocks such that
	\begin{equation}
		\varphi \oplus I_{\alpha} \oplus 0_{\beta, \eta} \sim_R \varphi' \oplus I_{\alpha'} \oplus 0_{\beta', \eta'}.
	\end{equation}
\end{definition}

Physically, $\varphi_X \sim_{R,\mathrm{Cl},\mathrm{st}} \varphi_X'$ implies that the two topological CSS codes can be mapped into one another via a finite-depth translation-invariant Clifford circuit, up to the addition or removal of unentangled ancilla qudits and redundant stabilizer generators.

A remarkable consequence of this algebraic setup is that, in the freely stable limit, free stable Clifford \(R\)-equivalence coincides exactly with free stable \(R\)-equivalence, vastly simplifying the mathematical classification. This coincidence follows from the standard elementary factorization of the balanced block \(\operatorname{diag}(\Xi,\Xi^{-1})\), a consequence of Whitehead's lemma; see Milnor~\cite[Sec.~3]{milnor1971introduction} and Weibel~\cite[Chap.~III, Secs.~1.2--1.3]{weibel2013kbook}.

\begin{proposition}[Coincidence of algebraic and Clifford stable \(R\)-equivalence]
	\label{lem:stable-clifford-equivalence}
	Let $\varphi$ and $\varphi'$ be matrices over $R$. Then $\varphi$ and $\varphi'$ are freely stably \(R\)-equivalent ($\varphi \sim_{R,\mathrm{st}} \varphi'$) if and only if they are freely stably Clifford \(R\)-equivalent ($\varphi \sim_{R,\mathrm{Cl},\mathrm{st}} \varphi'$).
\end{proposition}

\begin{proof}
		Since $\operatorname{E}_n(R)$ is a subgroup of $\operatorname{GL}_n(R)$ for any dimension $n$, every free stable Clifford \(R\)-relation automatically implies a free stable \(R\)-equivalence relation.

		Conversely, assume $\varphi \sim_{R,\mathrm{st}} \varphi'$. By definition, there exist suitably padded matrices $\tilde{\varphi}$ and $\tilde{\varphi}'$ (obtained by adding trivial identity and zero blocks to $\varphi$ and $\varphi'$, respectively) of the same dimensions $\tilde{q} \times \tilde{a}$, and invertible matrices $\Xi \in \operatorname{GL}_{\tilde{q}}(R)$ and $\Phi \in \operatorname{GL}_{\tilde{a}}(R)$ such that
	\begin{equation}
		\Xi \tilde{\varphi} \Phi = \tilde{\varphi}'.
	\end{equation}
	We stabilize both sides further by taking the direct sum with the identity matrix $I_{\tilde{q}}$, yielding the matrices $\tilde{\varphi} \oplus I_{\tilde{q}}$ and $\tilde{\varphi}' \oplus I_{\tilde{q}}$. We may then define the block diagonal transformation matrices
	\begin{equation}
		W := \operatorname{diag}(\Xi, \Xi^{-1})
		\qquad \text{and} \qquad
		\tilde{\Phi} := \operatorname{diag}(\Phi, \Xi).
	\end{equation}
	It is straightforward to verify that applying these transformations yields
	\begin{equation}
		W (\tilde{\varphi} \oplus I_{\tilde{q}}) \tilde{\Phi}
		= \operatorname{diag}(\Xi \tilde{\varphi} \Phi,\; I_{\tilde{q}})
		= \tilde{\varphi}' \oplus I_{\tilde{q}}.
	\end{equation}
	By directly invoking Whitehead's Lemma, the block-diagonal matrix $W = \operatorname{diag}(\Xi, \Xi^{-1})$ admits an elementary factorization, forcing it strictly into the elementary group $\operatorname{E}_{2\tilde{q}}(R)$. Since the right transformation matrix $\tilde{\Phi}$ remains invertible in $\operatorname{GL}_{\tilde{a}+\tilde{q}}(R)$, the extended matrices $\tilde{\varphi} \oplus I_{\tilde{q}}$ and $\tilde{\varphi}' \oplus I_{\tilde{q}}$ satisfy the Clifford \(R\)-equivalence relation ($\sim_{R,\mathrm{Cl}}$). Because these matrices are constructed purely by adding trivial blocks to $\varphi$ and $\varphi'$, it follows that $\varphi$ and $\varphi'$ are freely stably Clifford \(R\)-equivalent ($\sim_{R,\mathrm{Cl},\mathrm{st}}$).
\end{proof}

Consequently, the phase identification of
topological CSS codes is captured by the mathematically simpler relation of
free stable \(R\)-equivalence (\(\sim_{R,\mathrm{st}}\)). This reduction provides a conceptual advantage: free stable Clifford \(R\)-equivalence
(\(\sim_{R,\mathrm{Cl},\mathrm{st}}\)) is formulated in terms of elementary operations in a chosen basis of the physical Pauli module, whereas \(\sim_{R,\mathrm{st}}\) permits arbitrary invertible changes of coordinates.
We may therefore work basis-independently and freely pass between a concrete matrix and the abstract \(R\)-linear map it represents. Henceforth, we use ``matrix over \(R\)'' and ``\(R\)-linear map'' interchangeably.

\subsubsection{Free stable \(R\)-isomorphism of modules}

Although \(\sim_{R,\mathrm{st}}\) is basis-independent,
it is still expressed in terms of presentation maps that retain explicit choices of stabilizer generators.  Passing to their
cokernels recasts the problem in module-theoretic language and makes the
tools of homological algebra available.  For a topological CSS code with
stabilizer maps $\varphi_X$ and $\varphi_Z$, the cokernels
$\operatorname{coker}\varphi_X=P_X/\operatorname{im}\varphi_X$ and
$\operatorname{coker}\varphi_Z=P_Z/\operatorname{im}\varphi_Z$ describe
$X$- and $Z$-type Pauli operators modulo local stabilizers.  We call them the
\textit{reduced Pauli-$X$} and \textit{reduced Pauli-$Z$ modules}, respectively.

\begin{definition}[Free stable \(R\)-isomorphism for modules]
	Two $R$-modules $M$ and $N$ are \textit{freely stably \(R\)-isomorphic}, written $M \cong_{R,\mathrm{st}} N$, if there exist integers $m,n \ge 0$ such that $M \oplus R^m \cong_R N \oplus R^n$.
\end{definition}

\begin{proposition}[Coincidence of matrix and module stable relations]
	\label{prop:stable-module-matrix}
	For matrices $\varphi \colon R^a \to R^q$ and $\varphi' \colon R^{a'} \to R^{q'}$ over $R$, the following are equivalent:
	\begin{enumerate}
			\item The matrices $\varphi$ and $\varphi'$ are freely stably \(R\)-equivalent ($\varphi \sim_{R,\mathrm{st}} \varphi'$).
		\item The modules $\operatorname{coker}\varphi$ and $\operatorname{coker}\varphi'$ are freely stably \(R\)-isomorphic ($\operatorname{coker}\varphi \cong_{R,\mathrm{st}} \operatorname{coker}\varphi'$).
	\end{enumerate}
\end{proposition}

\begin{proof}
	The implication $(1) \Rightarrow (2)$ follows because free stable \(R\)-equivalence allows padding $\varphi$ and $\varphi'$ with identity blocks ($I_\beta$) or zero blocks ($0_{\alpha, \eta}$), which distributes over the cokernel and maps directly to adding free module summands ($R^\alpha$) to the cokernel.

	The reverse implication $(2) \Rightarrow (1)$ requires an explicit algebraic construction: adjoining new generators and relations to build extended block matrices, and applying elementary row and column operations to establish strict \(R\)-equivalence. We defer this constructive proof to Appendix~\ref{app:stable-module-matrix-proof}.
\end{proof}

\begin{remark}[Code--module dictionary]
	\label{rem:physical-dictionary-modules}
	For topological CSS stabilizer maps $\varphi_X$ and $\varphi_X'$,
	Propositions~\ref{prop:stable-module-matrix} and
	\ref{lem:stable-clifford-equivalence} give
	\begin{equation*}
		\begin{aligned}
			\operatorname{coker}\varphi_X
			\cong_{R,\mathrm{st}}
			\operatorname{coker}\varphi_X'
			&\;\Longleftrightarrow\;
			\varphi_X\sim_{R,\mathrm{st}}\varphi_X'\\
			&\;\Longleftrightarrow\;
			\varphi_X\sim_{R,\mathrm{Cl},\mathrm{st}}\varphi_X'.
		\end{aligned}
	\end{equation*}
	The last relation means that the two codes are connected by a
	translation-invariant finite-depth Clifford circuit after allowing
	unentangled ancillas and changes of stabilizer generators.
\end{remark}

\subsection{Single-degree stable Clifford rigidity}

Using our code--module dictionary, we first examine when a single-degree
superselection profile determines the stable translation-invariant Clifford
class of a topological CSS code. This setting admits a direct matrix-level
proof of rigidity and an explicit representative of that class.

Topological exactness fixes the $Z$-stabilizer submodule via $\mathrm{im}\,\varphi_Z = \ker\varphi_X^\dagger$, reducing the analysis entirely to the $X$-sector.
We use generalized Schanuel's lemma in the following map form,
together with Propositions~\ref{prop:stable-module-matrix}
and~\ref{lem:stable-clifford-equivalence}, to relate partial free
resolutions to stable Clifford equivalence.

\begin{lemma}[Generalized Schanuel's lemma for maps]
    \label{lem:generalized-schanuel}
    Let $R$ be a unital commutative ring, let $n \ge 1$, and let $N$ be an $R$-module. Suppose
    \begin{equation}
        F_n \xrightarrow{\alpha_n} F_{n-1} \to \cdots \to F_0 \to N \to 0
    \end{equation}
    and
    \begin{equation}
        G_n \xrightarrow{\beta_n} G_{n-1} \to \cdots \to G_0 \to N \to 0
    \end{equation}
    are two sequences of finitely generated free $R$-modules, exact except possibly at $F_n$ and $G_n$.
    Then the leftmost maps
	\begin{equation}
		\alpha_n \colon F_n \to F_{n-1},
		\qquad
		\beta_n \colon G_n \to G_{n-1}
	\end{equation}
	are freely stably \(R\)-equivalent ($\alpha_n \sim_{R,\mathrm{st}} \beta_n$).
\end{lemma}

\begin{proof}
	The cokernels of $\alpha_n$ and $\beta_n$ are two $(n-1)$-st syzygies of $N$. The generalized Schanuel lemma for modules, reviewed in Appendix~\ref{app:schanuel-lemma} (see also Lam~\cite[\S~5A, Cor.~5.5]{Lam1999}), implies that these cokernels are freely stably \(R\)-isomorphic. Proposition~\ref{prop:stable-module-matrix} then implies $\alpha_n \sim_{R,\mathrm{st}} \beta_n$.
\end{proof}

This resolution-independence naturally leads to the algebraic foundation of our rigidity result, which we state below in both coordinate-free module form and explicit matrix form.

\begin{proposition}[Single-degree rigidity: module form]
	\label{prop:isolated-ext-module}
	\label{thm:isolated-ext-module}
	\label{thm:stable-determination-isolated-ext}
	Let $R$ be a unital commutative Noetherian ring such that every finitely generated projective $R$-module is free. Let $M$ be a finitely generated $R$-module admitting a finite-length free resolution. Fix $j \ge 1$, and assume that the positive-degree Ext modules of $M$ are concentrated in degree $j$, meaning that
	\begin{equation}
		\operatorname{Ext}_R^\ell(M,R)=0
		\qquad (\ell>0,\ \ell\neq j).
	\end{equation}
	Set $E\coloneqq \operatorname{Ext}_R^j(M,R)$, and let
	\begin{equation}
		G_j\xrightarrow{\partial_j}G_{j-1}
		\longrightarrow\cdots\longrightarrow
		G_0\longrightarrow E\longrightarrow 0
	\end{equation}
	be any partial resolution of $E$ by finitely generated free
	$R$-modules, exact except possibly at $G_j$. Then
	\begin{equation}
		M\cong_{R,\mathrm{st}}\operatorname{coker}(\partial_j^\vee).
	\end{equation}
\end{proposition}

\begin{proof}
	Under the stated hypotheses, the assumed finite-length free resolution
	may be chosen with finitely generated free terms. In particular, $M$
	has finite projective dimension.
	Suppose first that $E=0$. Then every positive-degree Ext module of $M$ vanishes. The projective-dimension identity
	\begin{equation}
		\operatorname{pd}_R(M)=\max\{\ell:\operatorname{Ext}_R^\ell(M,R)\neq 0\}
	\end{equation}
	implies $\operatorname{pd}_R(M)=0$. Thus $M$ is finitely generated and projective, hence free by hypothesis. Lemma~\ref{lem:generalized-schanuel}, applied to the given resolution of $E=0$ and the all-zero partial resolution, shows that $\partial_j$ is freely stably \(R\)-equivalent to the empty map. Dualizing and applying Proposition~\ref{prop:stable-module-matrix} show that $\operatorname{coker}(\partial_j^\vee)$ is freely stably \(R\)-isomorphic to zero. Since $M$ is free, the claimed stable \(R\)-isomorphism follows.

	Now suppose that $E\neq 0$. The same projective-dimension identity gives $\operatorname{pd}_R(M)=j$. We may therefore choose a free resolution of exact length $j$ by finitely generated free modules,
	\begin{equation}
		0\longrightarrow F_j\xrightarrow{\delta_j}F_{j-1}
		\longrightarrow\cdots\longrightarrow
		F_1\xrightarrow{\delta_1}F_0\longrightarrow M\longrightarrow 0.
	\end{equation}
	Dualizing this resolution computes $\operatorname{Ext}_R^\bullet(M,R)$. The single-degree hypothesis therefore makes the augmented sequence
	\begin{equation}
		F_0^\vee\xrightarrow{\delta_1^\vee}F_1^\vee
		\longrightarrow\cdots\longrightarrow
		F_j^\vee\longrightarrow E\longrightarrow 0
	\end{equation}
	a length-$j$ partial resolution of $E$ by finitely generated free
	modules, exact except possibly at $F_0^\vee$.
	Lemma~\ref{lem:generalized-schanuel} gives $\delta_1^\vee\sim_{R,\mathrm{st}}\partial_j$. Dualizing again, and using the canonical double-dual identifications for finitely generated free modules, gives $\delta_1\sim_{R,\mathrm{st}}\partial_j^\vee$. Proposition~\ref{prop:stable-module-matrix} now yields
	\begin{equation}
		M\cong_R\operatorname{coker}(\delta_1)
		\cong_{R,\mathrm{st}}\operatorname{coker}(\partial_j^\vee),
	\end{equation}
	as required.
\end{proof}

\begin{corollary}[Single-degree rigidity: matrix form]
	\label{cor:isolated-ext-matrix}
	Under the hypotheses of Proposition~\ref{prop:isolated-ext-module}, let $\phi\colon R^a\to R^q$ be any presentation of $M$, so that $M\cong_R\operatorname{coker}\phi$. Then
	\begin{equation}
		\phi\sim_{R,\mathrm{st}}\partial_j^\vee.
	\end{equation}
\end{corollary}

\begin{proof}
	Proposition~\ref{prop:isolated-ext-module} gives
	\begin{equation}
		\operatorname{coker}\phi
		\cong_{R,\mathrm{st}}
		\operatorname{coker}(\partial_j^\vee).
	\end{equation}
	Proposition~\ref{prop:stable-module-matrix} then gives $\phi\sim_{R,\mathrm{st}}\partial_j^\vee$.
\end{proof}

We now apply this algebraic result to translation-invariant CSS codes. By the
Chinese-remainder reduction in Eq.~\eqref{eq:CRT}, it suffices to work over a
prime-power component \(R=\mathbb Z_{p^m}[\mathbb Z^D]\).
We apply the matrix statement to \(\overline{\varphi_\sigma}\).
Spatial inversion then converts the algebraic dual \((\vee)\) into the physical dagger dual
\((\dagger)\). For notational convenience, the proof and subsequent
applications use the \(X\)-sector; the corresponding \(Z\)-sector statements follow by
exchanging \(X\) and \(Z\) throughout.

\begin{theorem}[Single-degree stable Clifford rigidity]
	\label{thm:topological-rigidity}
	Let $p$ be prime and $m\geq 1$, and let $(\varphi_X,\varphi_Z)$ define a translation-invariant topological CSS code over
	$R=\mathbb{Z}_{p^m}[\mathbb{Z}^D]$.
	For fixed \(\sigma\in\{X,Z\}\), assume that $\operatorname{coker}\varphi_\sigma$ admits a finite-length free resolution, a condition that is automatic when $m=1$. Fix $j\geq 1$ and suppose that the \(\sigma\)-sector superselection profile is single-degree, supported in degree $j$, meaning
	\begin{equation}
		\operatorname{Ext}_R^\ell
		\bigl(\overline{\operatorname{coker}\varphi_\sigma},R\bigr)=0
		\qquad (\ell>0,\ \ell\neq j).
	\end{equation}
	Set
	$E\coloneqq \operatorname{Ext}_R^j
		\bigl(\overline{\operatorname{coker}\varphi_\sigma},R\bigr)$,
	and let
	\begin{equation}
		G_j\xrightarrow{\partial_j}G_{j-1}
		\longrightarrow\cdots\longrightarrow
		G_0\longrightarrow E\longrightarrow0
	\end{equation}
	be any partial resolution of $E$ by finitely generated free
	$R$-modules, exact except possibly at $G_j$. Then
	\begin{equation}
		\varphi_\sigma\sim_{R,\mathrm{Cl},\mathrm{st}}\partial_j^\dagger. \label{eq:rigidity}
	\end{equation}
\end{theorem}

\begin{proof}
	By symmetry between the Pauli sectors, it suffices to prove the statement
	for \(\sigma=X\).
	The ring $R$ is unital, commutative, and Noetherian. To verify the projective-freeness hypothesis needed for Corollary~\ref{cor:isolated-ext-matrix}, set $I=pR$. Then $I^m=0$ and $R/I\cong\mathbb Z_p[\mathbb Z^D]$. If $P$ is a finitely generated projective $R$-module, then $P/IP$ is a finitely generated projective $R/I$-module and hence is free by the Laurent-polynomial extension of the Quillen--Suslin theorem~\cite{Swan1978}. Choose an isomorphism $(R/I)^r\cong_{R/I}P/IP$ and lift it to a map $f\colon R^r\to P$. Since $f$ is an isomorphism modulo $I$, its cokernel $C$ satisfies $C=IC$, and hence $C=I^mC=0$. Thus $f$ is surjective. Because $P$ is projective, $f$ splits, giving $R^r\cong_R K\oplus P$ with $K=\ker f$. Reduction modulo $I$ gives $K/IK=0$, so $K=I^mK=0$. Therefore $f$ is an isomorphism and $P$ is free.

	When $m=1$, the ring $R=\mathbb Z_p[\mathbb Z^D]$ is a localization of $\mathbb Z_p[x_1,\ldots,x_D]$. Hilbert's syzygy theorem and localization imply that $R$ has global dimension at most $D$, so every finitely generated $R$-module has finite projective dimension. By the projective-freeness just proved, it therefore admits a finite-length free resolution, which justifies the assertion in the theorem.

	These observations verify the ring-theoretic hypotheses of Corollary~\ref{cor:isolated-ext-matrix}. We now apply it to the conjugate stabilizer map $\overline{\varphi_X}$, which presents the conjugate module:
	\begin{equation}
		\overline{\operatorname{coker}\varphi_X}
		\cong_R
		\overline{P_X}/\operatorname{im}\overline{\varphi_X}.
	\end{equation}
	Conjugation preserves finite-length free resolutions, so the corollary gives
	\begin{equation}
		\overline{\varphi_X}\sim_{R,\mathrm{st}}\partial_j^\vee.
	\end{equation}
	Applying spatial inversion entrywise preserves free stable \(R\)-equivalence and, since $\overline{\partial_j^\vee}=\partial_j^\dagger$, yields $\varphi_X\sim_{R,\mathrm{st}}\partial_j^\dagger$. Finally, Proposition~\ref{lem:stable-clifford-equivalence} upgrades this relation to $\varphi_X\sim_{R,\mathrm{Cl},\mathrm{st}}\partial_j^\dagger$.
\end{proof}

\phantomsection
\paragraph*{Scope and limits of Clifford classification.}
\label{subsubsec:clifford-scope}
\label{rem:prime-power-rigidity-scope}

For prime qudits (\(m=1\)), the finite-length free-resolution hypothesis
is automatic, so
Theorem~\ref{thm:topological-rigidity} applies to every topological CSS
code with a single-degree profile. For prime-power qudits (\(m>1\)),
rigidity under stable Clifford equivalence is restricted to codes
satisfying this hypothesis. Outside that class, the profile need not
determine the stable Clifford class;
Appendix~\ref{app:pm-obstruction} demonstrates this already for the
vanishing profile. General FDLU equivalence provides a faithful circuit
characterization of topological phase equivalence. We nevertheless retain
the more restrictive Clifford setting because it enables a simple algebraic
treatment and suffices for the phase identifications established here.
The eight toric-code extensions in
Sec.~\ref{sec:non-splitting-extensions} are already distinguished by their
finite-torus GSD functions under general translation-invariant FDLUs,
with primitive translations fixed and product-state ancillas allowed.
These examples motivate further study of the profile data preserved by
such circuits in general.

\subsection{Code equivalence and minimal coarse graining}

We now draw the physical consequences of
Theorem~\ref{thm:topological-rigidity}. First, we identify when two
topological CSS codes are connected by a stable Clifford circuit. We then
determine the minimal coarse graining needed to reach a toric-code-stack
normal form and interpret this normal form from an
entanglement-renormalization perspective.

\subsubsection{Code equivalence from single-degree data}

For prime qudits ($m=1$), matching single-degree \(X\)-sector
superselection profiles give an unconditional criterion for stable,
translation-invariant Clifford equivalence between topological CSS codes.

\begin{corollary}[Stable translation-invariant Clifford equivalence]
    \label{cor:clifford-equivalence-isolated}
    Let \((\varphi_{1,X},\varphi_{1,Z})\) and \((\varphi_{2,X},\varphi_{2,Z})\) be two translation-invariant topological CSS codes over $R=\mathbb{Z}_p[\mathbb{Z}^D]$. For fixed \(\sigma\in\{X,Z\}\), suppose that the \(\sigma\)-sector superselection profiles of both codes are single-degree and supported in the same degree $j\ge 1$, and that their \(j\)-th superselection layers are isomorphic:
    \begin{equation}
        \operatorname{Ext}_R^j\!\bigl(\overline{\operatorname{coker}\varphi_{1,\sigma}}, R\bigr)
        \cong_R
        \operatorname{Ext}_R^j\!\bigl(\overline{\operatorname{coker}\varphi_{2,\sigma}}, R\bigr).
    \end{equation}
    Then $\varphi_{1,\sigma}\sim_{R,\mathrm{Cl},\mathrm{st}}\varphi_{2,\sigma}$, namely, the two codes are connected by a finite-depth, translation-invariant Clifford circuit up to the addition or removal of unentangled product-state ancillas.
\end{corollary}

\begin{proof}
    Let \(E\) denote the common isomorphism class of the two \(j\)-th
    superselection layers. Fix a partial resolution of \(E\) by finitely generated free
    \(R\)-modules, exact except possibly at \(G_j\), with leftmost boundary
    map \(\partial_j\colon G_j\to G_{j-1}\). By
    Theorem~\ref{thm:topological-rigidity}, both \(\sigma\)-sector maps are stably
    Clifford equivalent to the same dual map:
    \(\varphi_{1,\sigma}\sim_{R,\mathrm{Cl},\mathrm{st}}\partial_j^\dagger\) and
    \(\varphi_{2,\sigma}\sim_{R,\mathrm{Cl},\mathrm{st}}\partial_j^\dagger\). Transitivity
    gives
    \(\varphi_{1,\sigma}\sim_{R,\mathrm{Cl},\mathrm{st}}\varphi_{2,\sigma}\).
\end{proof}

\begin{remark}[Hadamard symmetry and sector exchange]
    Because transversal Hadamard gates exchange $X$- and $Z$-sectors, this Clifford equivalence applies symmetrically. One may freely cross-compare the $X$-sector data of one code with the $Z$-sector data of another (e.g., comparing $\varphi_{1,X}$ with $\varphi_{2,Z}$) to establish stable Clifford equivalence.
\end{remark}

\begin{remark}[Low-dimensional topological CSS codes on prime qudits]
    \label{rem:rigidity-interpretation}
    Every nonzero superselection layer lies in \(1\leq\ell\leq D-1\),
    as established in Sec.~\ref{subsubsec:code-construction}
    [Eq.~\eqref{eq:spatial-grade-bound}].
    For \(D=1\), there is no allowed positive degree. The profile therefore
    vanishes, and
    Theorem~\ref{thm:topological-rigidity} implies that every such code is
    stably Clifford trivial: it can be reduced to a product state after
    adding or removing product-state ancillas.

    For \(D=2\), the same bound implies that every nonvanishing
    \(X\)-sector superselection profile is single-degree and supported in degree \(j=1\).
    Moreover, Eq.~\eqref{eq:topological-dimension-bound} implies
    \(\lvert E_X^1\rvert<\infty\), so the pointlike excitations
    are anyonic rather than fractonic.
    The first superselection layer
    \begin{equation}
        \operatorname{Ext}_R^1\!\bigl(\overline{\operatorname{coker}\varphi_X},R\bigr)
    \end{equation}
    records the fusion rules and the full translation action on the anyon
    sectors, including nontrivial translation-induced permutations. It
    uniquely determines the stable translation-invariant Clifford class: any two such codes with isomorphic first superselection layers
    are equivalent up to a finite-depth, translation-invariant Clifford
    circuit and product-state ancillas.
\end{remark}

For prime-power qudits (\(m>1\)), the same conclusions about stable
Clifford equivalence hold provided the relevant reduced modules admit
finite-length free resolutions. Without this hypothesis, even a vanishing
profile can remain Clifford nontrivial
(Appendix~\ref{app:pm-obstruction}). We expect the same phase-identification
conclusions to hold under general FDLU equivalence without the finite-length
resolution hypothesis; establishing this extension is left to future work.

\subsubsection{Minimal coarse graining to toric-code stacks}
\label{subsubsec:minimal-coarse-graining}

When the $X$-sector superselection profile is single-degree, supported in degree $j$, with a finite \(j\)-th superselection layer $E$, its topological excitation data closely mirror a stack of toric codes. This raises two natural physical questions: First, does the code belong to the same intrinsic topological order as a toric-code stack? Second, if so, can this equivalence be realized by a translation-invariant, finite-depth, finite-range Clifford circuit (up to product-state ancillas), and at what minimal coarse-graining scale? The following corollary answers both questions: the minimal coarse-graining scale corresponds to the layerwise mobility lattice (Sec.~\ref{subsec:set-order}), which reduces to $\Lambda_{\mathrm{ss}}=\Lambda_E$ under the single-degree hypothesis.

\begin{corollary}[Minimal coarse graining to toric-code stacks]
    \label{thm:minimal-coarse-graining}
    Let $D\geq 2$ and $1\leq j\leq D-1$, and let $\varphi_X$ define a
    translation-invariant topological CSS code over
    $R=\mathbb{Z}_p[\mathbb{Z}^D]$ whose $X$-sector superselection profile is single-degree and supported in degree $j$. Set
    $E\coloneqq
    \operatorname{Ext}_R^j(\overline{\operatorname{coker}\varphi_X},R)$,
    let $\Lambda_E$ be its mobility lattice as in
    Eq.~\eqref{eq:mobility}, and suppose that
    $Q\coloneqq\dim_{\mathbb{Z}_p}E<\infty$. After coarse-graining to a
    finite-index translation subgroup $\Lambda\subset\mathbb{Z}^D$, the code is equivalent to $Q$ decoupled copies of the degree-$j$ toric code $\mathrm{TC}_j^{(D)}(\mathbb{Z}_p)$ via a $\Lambda$-invariant Clifford circuit and product-state ancillas if and only if $\Lambda\subseteq\Lambda_E$.

    Equivalently, the minimal coarse-graining that permits decoupling into the toric-code normal form $\mathrm{TC}_j^{(D)}(\mathbb{Z}_p)^{\oplus Q}$ corresponds precisely to $\Lambda_E$.
\end{corollary}

\begin{proof}
	Suppose first that \(\Lambda\subseteq\Lambda_E\), and set
	\(R_\Lambda=\mathbb Z_p[\Lambda]\). Restriction of scalars preserves
	topological exactness. Because \(R\) is a finite free Frobenius extension
	of \(R_\Lambda\), it also identifies the positive-degree Ext modules over
	\(R_\Lambda\) with those over \(R\), carrying the restricted translation
	action. Since \(\Lambda\subseteq\Lambda_E\), this action is trivial on
	\(E\), so
	\(E\cong_{R_\Lambda}\mathbb Z_p^Q\), where the right-hand side carries
	the trivial \(\Lambda\)-action. This is the
	\(j\)-th superselection layer of the \(X\)-sector profile of \(Q\)
	decoupled degree-\(j\) toric codes.
	Corollary~\ref{cor:clifford-equivalence-isolated}, applied over
	\(R_\Lambda\), gives the claimed Clifford equivalence.

	Conversely, a translation-compatible Clifford circuit preserves the
	action of the retained translations on superselection sectors. That
	action is trivial for the standard toric-code stack, so every
	\(\mathbf v\in\Lambda\) acts trivially on \(E\). Hence
	\(\Lambda\subseteq\Lambda_E\).
\end{proof}

As expected on physical grounds, the minimal coarse graining is independent
of the chosen Pauli sector: the \(X\)- and \(Z\)-sector layers carry
contragredient translation actions and hence define the same mobility lattice
\(\Lambda_E\).

Because \(E\) is finite, \(\Lambda_E\) has finite index. Taking
\(\Lambda=\Lambda_E\) in the corollary gives a stable Clifford equivalence
to \(Q\) decoupled degree-\(j\) toric codes. Since blocking and product-state
stabilization preserve the intrinsic phase and Clifford circuits are FDLUs,
the code therefore lies in the intrinsic phase of this toric-code stack.

Within the CSS setting, this corollary both extends and sharpens
earlier classifications of two-dimensional topological stabilizer codes,
including the structural results of Bombin and
collaborators~\cite{Bombin2012Universal,Bombin2014Structure} and Haah's
algebraic treatment~\cite{Haah2016Algebraic,Haah2021Classification}.
Those works established
toric-code-stack normal forms after suitable coarse graining but generally
required further coarse graining beyond the maximal anyon-preserving
translation subgroup, leaving open whether this additional symmetry reduction
was necessary. For finite single-degree profiles, the corollary gives the
exact answer: the compatible retained translation subgroups are precisely
\(\Lambda\subseteq\Lambda_E\). Thus \(\Lambda_E\) is the maximal retained
symmetry and determines the minimal coarse graining. The result applies in
every dimension \(D\geq2\) and every allowed degree
\(1\leq j\leq D-1\). In two dimensions,
Remark~\ref{rem:rigidity-interpretation} gives the stronger
SET statement before coarse graining: the full \(R\)-module
structure retains the translation action on the anyon sectors, including
nontrivial anyon permutations, and determines the stable
translation-invariant Clifford class. Restricting translations to
\(\Lambda_E\) trivializes this action and yields the toric-code-stack normal
form. In the overlapping two-dimensional qubit setting, this criterion agrees
with the independent result of Wang \textit{et al.}~\cite{Wang2026Decoupling}.
Their constructive approach provides an explicit polynomial-time algorithm
for the decoupling map, together with bounds on operator spreading.

For prime-power qudits, the same toric-code-stack normal form---and hence the
same intrinsic-phase conclusion---holds under two additional algebraic
hypotheses. Over \(R=\mathbb Z_{p^m}[\mathbb Z^D]\) with \(m>1\), assume in
addition that \(\operatorname{coker}\varphi_X\) admits a finite-length free
resolution and that
\(
	E\coloneqq
	\operatorname{Ext}_R^j
	\bigl(\overline{\operatorname{coker}\varphi_X},R\bigr)
\)
is free of finite rank \(Q\) over \(\mathbb Z_{p^m}\). The target is then
\(Q\) decoupled copies of
\(\mathrm{TC}_j^{(D)}(\mathbb Z_{p^m})\), and the criterion remains
\(\Lambda\subseteq\Lambda_E\). The argument is unchanged, with
\(\mathbb Z_p^Q\) replaced by \(\mathbb Z_{p^m}^Q\) and
Theorem~\ref{thm:topological-rigidity} applied over
\(\mathbb Z_{p^m}[\Lambda]\). The Clifford-circuit reduction established
here applies under these hypotheses.

\subsubsection{Entanglement-renormalization perspective}
\label{subsubsec:entanglement-renormalization}

Rigidity constrains the possible form of real-space entanglement
renormalization directly from superselection data. For a topological qubit CSS
code with a single finite nonzero \(X\)-sector layer \(E\) in degree \(j\), the
layer degree (equivalently, excitation dimension \(j-1\) in this mobile
setting), fusion rules, and lattice-translation action fix the stable
translation-invariant Clifford class. After coarse graining to a finite-index
subgroup \(\Lambda\subseteq\Lambda_E\), the restricted layer has trivial
translation action and is isomorphic to \(\mathbb Z_2^Q\), where
\(Q=\dim_{\mathbb Z_2}E\). It therefore matches the layer of \(Q\) decoupled
degree-\(j\) toric codes, fixing the toric-code-stack normal form up to stable
Clifford equivalence.
The equivalence is implemented by an exact stabilizer
entanglement-renormalization step: a finite-depth local Clifford circuit
disentangles product-state qudits, which are removed before rescaling the
lattice. For the two-dimensional toric code, elementary vertex and plaquette
moves change the cellulation through CNOT circuits while preserving the
topological degrees of freedom~\cite{AguadoVidal2008Entanglement,Aguado2011Disentanglement}.
The standard hypercubic \(\mathbb Z_2\) toric codes in two, three, and four
dimensions likewise furnish exact, nonbifurcating entanglement-RG fixed
points~\cite{Haah2014Bifurcation}.

This perspective also clarifies when the microscopic cellulation may be
suppressed. Suppose that two bounded-geometry cellulations are related by a
bounded-depth, uniformly local sequence of elementary cellular expansions
and collapses. Their cellular chain complexes are then related by local basis
changes and the addition or removal of contractible two-term summands. At the
physical code degree, these summands correspond to product-state qudits,
while the basis changes are implemented by local CNOT gates. The two cellular
models are therefore stably Clifford equivalent and represent the same
intrinsic phase. We suppress the cellulation only among representatives
related in this sense; when translation symmetry is retained, the cellular
moves must additionally be \(\Lambda\)-periodic.

Fracton stabilizer models can exhibit qualitatively different
renormalization behavior~\cite{Dua2020}. Haah's cubic code has a bifurcating flow in which
coarse graining produces an additional inequivalent long-range-entangled
factor~\cite{Haah2014Bifurcation}. For the X-cube model, adding or removing a
foliation leaf is implemented by adjoining or removing a two-dimensional
toric-code layer through a local CNOT
circuit~\cite{ShirleySlagleWangChen2018Fracton}. This uses a larger
stabilization resource than product-state ancillas. By contrast, the codes
covered by Corollary~\ref{thm:minimal-coarse-graining} flow, after the
prescribed finite coarse graining, to ordinary nonbifurcating
toric-code-stack fixed points. In particular, this applies to the full-length
regular Koszul-complex models considered below.

\subsection{Toric-stack normal forms for full-length regular Koszul models}
\label{subsec:koszul-multistack-rigidity}

We now return to the regular Koszul sector, where both the single-degree
property and the finite-length free-resolution hypothesis follow from the
Koszul resolution itself. For full-length regular sequences,
Fact~\ref{fact:fusion-rules} further shows that the structure module is finite
and free over \(\mathbb Z_{p^m}\), so the preceding minimal-coarse-graining
result reduces every such model to a toric-code stack.

\begin{corollary}[Toric-code-stack normal form for full-length regular models]
	\label{cor:koszul-multistack-rigidity}
	Let \(D\geq2\), let \(p\) be prime, \(m\geq1\),
	\(R=\mathbb Z_{p^m}[\mathbb Z^D]\), and let \(\mathbf f\in R^D\) be a
	regular sequence. Set \(S\coloneqq R/(\mathbf f)\) and
	\begin{equation}
		\begin{aligned}
			Q&\coloneqq\operatorname{rank}_{\mathbb Z_{p^m}}S,\\
			\Lambda_S&\coloneqq
			\bigl\{\mathbf v\in\mathbb Z^D:(x^{\mathbf v}-1)S=0\bigr\}.
		\end{aligned}
	\end{equation}
	For \(1\leq j\leq D-1\) and every finite-index translation subgroup
	\(\Lambda\subseteq\mathbb Z^D\), write
	\(R_\Lambda\coloneqq\mathbb Z_{p^m}[\Lambda]\). Viewed as codes over
	\(R_\Lambda\),
	\begin{equation}
		\mathrm{KC}_j^{(D)}(\mathbf f;\mathbb Z_{p^m})
		\sim_{R_\Lambda,\mathrm{Cl},\mathrm{st}}
		\left[\mathrm{TC}_j^{(D)}(\mathbb Z_{p^m})\right]^{\oplus Q}
	\end{equation}
	if and only if \(\Lambda\subseteq\Lambda_S\), where
	\(\mathrm{TC}_j^{(D)}(\mathbb Z_{p^m})\) denotes any degree-\(j\) toric-code
	realization with trivial \(\Lambda\)-translation action on its
	\(j\)-th superselection layer.
\end{corollary}

\begin{proof}
	Fact~\ref{fact:koszul-excitation} identifies the sole nonzero
	positive-degree \(X\)-sector superselection layer with \(S\), and the
	Koszul complex supplies a finite-length free resolution. By
	Fact~\ref{fact:fusion-rules}, \(S\) is finite free of rank \(Q\) over
	\(\mathbb Z_{p^m}\). The prime-power extension of
	Corollary~\ref{thm:minimal-coarse-graining} therefore applies. Since
	\(E=S\) and \(\Lambda_E=\Lambda_S\), it gives the stated equivalence exactly
	when \(\Lambda\subseteq\Lambda_S\).
\end{proof}

\subsection{Single-degree realizability and construction}
\label{subsubsec:code-construction}

The preceding normal forms concern regular Koszul models. We now extend
their construction beyond regular-sequence quotients and ask which
single-degree profiles can occur in topological CSS codes. Motivated by
the normal form in Eq.~\eqref{eq:rigidity}, we give a realization criterion
and an explicit construction using free resolutions of the prescribed
layer that need not have finite length. Detailed proofs of the results in
this subsection are collected in Appendix~\ref{app:profile-realization}.

For a finitely generated \(R\)-module \(E\), recall that its grade is defined
by~\cite[Prop.~18.4]{eisenbud1995commutative}
\begin{equation}
	\operatorname{grade}_R(E)
	\coloneqq
	\inf\bigl\{i\geq0\bigm|
	\operatorname{Ext}_R^i(E,R)\neq0\bigr\}.
	\label{eq:module-grade}
\end{equation}
We use the convention \(\operatorname{grade}_R(0)=\infty\).

Topological exactness imposes the following grade bound.
\begin{lemma}[Grade bound for superselection layers]
	\label{lem:topological-grade-bound}
	Let \(p\) be prime, \(m\geq1\), and
	\(R=\mathbb Z_{p^m}[\mathbb Z^D]\).
	For \(\ell\geq1\) and \(\sigma\in\{X,Z\}\), every translation-invariant
	topological CSS code over \(R\) satisfies
	\begin{equation}
		\operatorname{grade}_R(E_\sigma^\ell)>\ell.
		\label{eq:topological-grade-bound}
	\end{equation}
	No single-degree assumption is required.
\end{lemma}

\begin{remark}[Geometric interpretation]
\label{rem:grade-geometry}
For this Laurent ring, grade measures support codimension:
\begin{equation}
	\dim M=D-\operatorname{grade}_R(M),
	\label{eq:grade-support-dimension}
\end{equation}
where \(M\) is any nonzero finitely generated \(R\)-module and
\(\dim\) denotes Krull dimension. With the convention
\(\dim0=-\infty\), Lemma~\ref{lem:topological-grade-bound}
equivalently states
\begin{equation}
	\dim E_\sigma^\ell<D-\ell.
	\label{eq:topological-dimension-bound}
\end{equation}
This recasts the charge-module bound of Ruba and
Yang~\cite[Prop.~25]{Yang2023} in our notation.
\end{remark}

We formulate the realization problem in terms of \(X\)-sector profiles;
the results below hold equally with \(X\) and \(Z\) interchanged.
For fixed \(R\) and \(j\geq1\), let \(\mathrm{CSS}_j^X(E)\) denote the family of
translation-invariant topological CSS codes \(\mathcal C\) over \(R\) whose
\(X\)-sector superselection layers satisfy
\begin{equation}
	E_X^\ell(\mathcal C)\cong_R
	\begin{cases}
		E, & \ell=j,\\
		0, & \ell>0,\ \ell\neq j.
	\end{cases}
	\label{eq:realized-profile}
\end{equation}
Physically, \(\mathrm{CSS}_j^X(E)\neq\varnothing\) means that the prescribed
single-degree \(X\)-sector superselection profile is \textit{realizable} by a
topological CSS code.

\begin{theorem}[Single-degree realizability]
	\label{thm:single-degree-realizability}
	Let \(p\) be prime, \(m\geq1\), and
	\(R=\mathbb Z_{p^m}[\mathbb Z^D]\).
	For any finitely generated \(R\)-module \(E\) and \(j\geq1\),
	\begin{equation}
		\mathrm{CSS}_j^X(E)\neq\varnothing
		\quad\Longleftrightarrow\quad
		\operatorname{grade}_R(E)>j.
		\label{eq:realizability-grade}
	\end{equation}
	In the nonempty case, Construction~\ref{constr:single-degree-realization}
	yields a code for which
	\(\operatorname{coker}\varphi_X\) admits a free resolution of length at
	most \(j\). If \(E\neq0\), the constructed reduced Pauli-\(X\) module
	has projective dimension exactly \(j\).
\end{theorem}

\begin{construction}[Realization of single-degree profiles]
\label{constr:single-degree-realization}
For \(\operatorname{grade}_R(E)>j\),
choose any free resolution of \(E\) by finitely generated free modules,
\begin{equation}
	\cdots\longrightarrow
	G_{j+1}\xrightarrow{\partial_{j+1}}G_j
	\xrightarrow{\partial_j}G_{j-1}
	\longrightarrow\cdots\longrightarrow
	G_0\longrightarrow E\longrightarrow0,
\end{equation}
and define the adjacent stabilizer maps by
\begin{equation}
	\begin{aligned}
		\varphi_X
		&\coloneqq
		\partial_j^\dagger
		\colon G_{j-1}^\dagger\longrightarrow G_j^\dagger,\\
		\varphi_Z
		&\coloneqq
		\partial_{j+1}
		\colon G_{j+1}\longrightarrow G_j.
	\end{aligned}
	\label{eq:profile-realization-maps}
\end{equation}
The resolution identity gives commuting checks, while the grade condition
makes the dual complex exact through the required degrees. Together, these
facts show that the maps in Eq.~\eqref{eq:profile-realization-maps} define a
code in \(\mathrm{CSS}_j^X(E)\), as asserted in
Theorem~\ref{thm:single-degree-realizability}.
\end{construction}

Although the chosen resolution of \(E\) may be infinite, the resulting
\(\operatorname{coker}\varphi_X\) has a finite-length free resolution, as
required by the rigidity theorem.

This construction extends regular Koszul models to arbitrary admissible
excitation modules by replacing the Koszul resolution with a general free
resolution. For an \(R\)-regular sequence \(\mathbf f\) of length \(s\),
let \(S=R/(\mathbf f)\neq0\). Its Koszul resolution recovers
\begin{equation}
	\mathrm{KC}_j^{(D,s)}(\mathbf f;\mathbb Z_{p^m})
	\in\mathrm{CSS}_j^X\!\left(R/(\mathbf f)\right).
	\label{eq:koszul-realization-family}
\end{equation}
Here \(\operatorname{grade}_R(S)=s\), so the realizability criterion gives
precisely \(1\leq j<s\).

The finite-size correspondences in Facts~\ref{fact:tor-isomorphism}
and~\ref{fact:logical-excitation-correspondence} also extend to
Construction~\ref{constr:single-degree-realization}, with
\[
\mathcal H_Z(\mathbf L)\cong_R
\operatorname{Tor}_j^R(E,R_{\mathbf L})\cong_R
H_j(K_\bullet(\mathbf b;E)).
\]
Indeed, the periodic \(Z\)-logical module is
\(H_j(G_\bullet\otimes_R R_{\mathbf L})\), and the regular boundary
sequence \(\mathbf b\) supplies a Koszul resolution of \(R_{\mathbf L}\).

More generally, a prescribed single-degree \(X\)-sector profile determines all
companion \(Z\)-layers, which may occur in several degrees (e.g., in
\(\mathrm{TC}_1^{(3)}(\mathbb Z_2)\oplus\text{Haah's cubic code}\)).
\begin{proposition}[Companion-sector layers]
	\label{prop:companion-sector-layers}
	Let \(R=\mathbb Z_{p^m}[\mathbb Z^D]\), with \(p\) prime and \(m\geq1\).
	For every \(\mathcal C\in\mathrm{CSS}_j^X(E)\),
	\begin{equation}
		E_Z^\ell(\mathcal C)
		\cong_R\operatorname{Ext}_R^{j+\ell}(\overline E,R),
		\qquad \ell\geq1.
		\label{eq:companion-sector-layers}
	\end{equation}
	No finite-length free resolution of \(\operatorname{coker}\varphi_X\)
	is assumed.
\end{proposition}
For Construction~\ref{constr:single-degree-realization}, the formula
follows by dimension shifting. For regular Koszul models, it recovers
Fact~\ref{fact:koszul-excitation}.

Since a nonzero module has nonnegative Krull dimension,
Eq.~\eqref{eq:topological-dimension-bound} restricts every nonzero layer to
\begin{equation}
	1\leq\ell\leq D-1.
	\label{eq:spatial-grade-bound}
\end{equation}
In particular, for \(D=2\), only the first superselection layer can
be nonzero, corresponding to pointlike excitations.

Theorem~\ref{thm:single-degree-realizability} establishes single-degree
realizability without assuming the finite-resolution hypothesis of rigidity.
Construction~\ref{constr:single-degree-realization} supplies an explicit
representative satisfying the finite-resolution hypothesis, so
Theorem~\ref{thm:topological-rigidity} places all codes in
\(\mathrm{CSS}_j^X(E)\) satisfying that hypothesis in a single stable
translation-invariant Clifford class. In particular, this class is
independent of the chosen resolution.

Finally, for a finite family of finitely generated \(R\)-modules \(E^j\) with
\(\operatorname{grade}_R(E^j)>j\), stacking the single-degree constructions
realizes the split \(X\)-sector profile
\begin{equation}
	\mathscr E_X\simeq\bigoplus_{j=1}^{D-1}E^j[-j],
\end{equation}
where \(\simeq\) denotes isomorphism in the derived category of
\(R\)-modules, and \(E^j[-j]\) places \(E^j\) in cohomological degree \(j\),
following the shift convention \((C[k])^\ell=C^{\ell+k}\).
Stacking realizes the prescribed layers, but not prescribed nonsplit
attachments between them. These split profiles provide a baseline for
the nonsplit extensions in the next section.

\section{Beyond single-degree rigidity: nonsplit extensions}
\label{sec:non-splitting-extensions}

We now consider multi-degree superselection profiles, which first become
possible in spatial dimensions \(D\geq3\). The admissible-degree
bound~\eqref{eq:spatial-grade-bound} implies that the positive-degree profile
vanishes in \(D=1\), while only degree one is allowed in \(D=2\). In
\(D=3\), degrees one and two can coexist, so the single-degree hypothesis of
Theorem~\ref{thm:topological-rigidity} is no longer automatic. Classification
must then account for possible gluing between different superselection layers.

Our code--module dictionary
(Remark~\ref{rem:physical-dictionary-modules}) applies beyond the single-degree
regime, reducing stable translation-invariant code classification to a
module-classification problem in homological algebra. Our full
superselection profile \(\mathscr E_\sigma\), including cross-degree
gluing data, is precisely the truncated derived dual of
\(\overline{\operatorname{coker}\varphi_\sigma}\) and determines the
corresponding stable module class via Yoshino's contravariant
equivalence~\cite[Theorem~1.6]{Yoshino1993MaximalBuchsbaum}. Combining our
dictionary with this algebraic correspondence extends
Theorem~\ref{thm:topological-rigidity} to multi-degree profiles: among
topological CSS codes over \(R=\mathbb Z_{p^m}[\mathbb Z^D]\) for which
\(\operatorname{coker}\varphi_\sigma\) admits a finite-length free
resolution, the derived
isomorphism class of \(\mathscr E_\sigma\) is a complete invariant of
CSS-preserving stable translation-invariant Clifford equivalence with
primitive translations fixed~\footnote{Heller's criterion identifies
isomorphism in Yoshino's stable category with projective stable
isomorphism~\cite[Theorem~2.2]{Heller1960LoopSpace}. Over \(R\), reduction
modulo \(pR\), Swan's theorem over \(R/pR\), and lifting across the nilpotent
ideal \(pR\) show that finitely generated projective modules are
free~\cite{Swan1978}; hence projective stabilization is precisely the free
stabilization in our code--module dictionary.}.

Here we focus on the physical aspect of this algebraic classification.
The individual superselection layers \(E_\sigma^\ell\) encode excitation
types, fusion rules, and translation actions, whereas the physical content
of their cross-degree gluing is not immediately apparent. We therefore ask
whether gluing has observable consequences when the individual layers and
their translation actions are fixed. To answer this question, we construct degree-\((1,2)\)
extensions of regular Koszul-complex models,
including nonsplit couplings, and compute their common \(E_\sigma^\ell\) in
both Pauli sectors. We then specialize to qubit toric codes, obtaining eight
models with identical layers \(E_\sigma^\ell\cong\mathbb Z_2\) and trivial
translation action for \(\sigma=X,Z\) and \(\ell=1,2\), yet different
system-size dependences of the ground-state degeneracy. These distinct GSD functions
rule out FDLU equivalence with the original translation symmetry fixed.
Finite coarse graining maps all eight models to the same decoupled
toric-code stack, so they share an intrinsic topological order while
realizing distinct translation SET orders. Thus, cross-degree gluing carries
detectable translation-SET information beyond the degreewise data.

\subsection{Degree-\texorpdfstring{\((1,2)\)}{(1,2)} extensions of Koszul-complex models}

\begin{construction}[Degree-\((1,2)\) extensions of Koszul-complex models]
\label{constr:degree-12-koszul-extension}
\label{def:degree-12-koszul-extension}
Let \(N\geq2\) and write
\begin{equation}
	R=\mathbb Z_N[x_1^{\pm1},x_2^{\pm1},x_3^{\pm1}],
\end{equation}
where \(x_i\) generates translation along the reference direction \(e_i\),
as illustrated in Fig.~\ref{fig:toric3D}.
Let \(\mathbf f=(f_1,f_2,f_3)\in R^3\) be a regular sequence, and let
\((K_\bullet(\mathbf f),\partial_\bullet)\) be its Koszul complex, using the
ordered bases and differentials in
Eqs.~\eqref{eq:3d-koszul-bases}--\eqref{eq:3d-differentials}.  Choose a coupling vector
\begin{equation}
	\mathbf a=\begin{pmatrix}a_1\\a_2\\a_3\end{pmatrix}\in R^3.
\end{equation}
Define the \(Z\)-sector CSS complex \(\mathcal D_{\bullet,\mathbf a}\) by
\begin{equation}
	K_2 \oplus K_3
	\xrightarrow{\mathcal{D}_{2,\mathbf{a}} \coloneqq \begin{pmatrix} \partial_2 & \mathbf{a} \\ 0 & \partial_3 \end{pmatrix}\;}
	K_1 \oplus K_2
	\xrightarrow{\mathcal{D}_{1,\mathbf{a}} \coloneqq \begin{pmatrix} \partial_1 & -\mathbf{a}^{\mathsf{T}} \\ 0 & \partial_2 \end{pmatrix}\;}
	K_0 \oplus K_1,
	\label{eq:AE-extension}
\end{equation}
and the stabilizer maps
\begin{equation}
	\varphi_{Z,\mathbf a}
	\coloneqq \mathcal{D}_{2,\mathbf{a}}
	=
	\begin{pmatrix}
		\partial_2&\mathbf a\\
		0&\partial_3
	\end{pmatrix},
	\quad
	\varphi_{X,\mathbf a}
	\coloneqq \mathcal{D}_{1,\mathbf{a}}^\dagger
	=
	\begin{pmatrix}
		\partial_1^\dagger&0\\
		-\overline{\mathbf a}&\partial_2^\dagger
	\end{pmatrix}.
	\label{eq:BE-extension}
\end{equation}

	The CSS code specified by Eqs.~\eqref{eq:AE-extension}
	and~\eqref{eq:BE-extension} is denoted by
	\begin{equation}
		\mathcal C_{\mathbf a}(\mathbf f;\mathbb Z_N)
		\coloneqq
		\mathrm{KC}_1^{(3)}(\mathbf f;\mathbb Z_N)
		\mathbin{\oplus_{\mathbf a}}
		\mathrm{KC}_2^{(3)}(\mathbf f;\mathbb Z_N)
	\end{equation}
	and called a \textit{degree-\((1,2)\) extension of Koszul-complex models}.
	Here \(\oplus_{\mathbf a}\) records the off-diagonal coupling;
	\(\mathbf a=0\) gives the ordinary direct sum.
\end{construction}

We call the extension \textit{split} if a translation-invariant change of CSS
presentation removes the off-diagonal coupling and reduces
\(\mathcal C_{\mathbf a}\) to the direct sum \(\mathcal C_{\mathbf 0}\);
otherwise, it is \textit{nonsplit}.
Unless stated otherwise, split and nonsplit refer to the original translation
ring \(R\), before any coarse graining.

We call extensions assembled from Koszul-complex models in different degrees
\textit{multi-degree extensions of Koszul-complex models}.
When all blocks are toric-code models, we call the result a
\textit{multi-degree toric-code extension}.

\begin{fact}[Topologicality of degree-\((1,2)\) extensions of Koszul-complex models]
	\label{prop:AE-BE-topological}
	For every regular sequence \(\mathbf f\in R^3\) and every
	\(\mathbf a\in R^3\), the model
	\(\mathcal C_{\mathbf a}(\mathbf f;\mathbb Z_N)\) is a
	translation-invariant topological CSS code.
\end{fact}

\begin{proof}
	Direct multiplication gives
	\begin{equation}
		\mathcal D_{1,\mathbf a}\mathcal D_{2,\mathbf a}
		=
		\begin{pmatrix}
			\partial_1\partial_2 &
			\partial_1\mathbf a-\mathbf a^{\mathsf T}\partial_3\\
			0&\partial_2\partial_3
		\end{pmatrix}
		=0,
	\end{equation}
	where the diagonal blocks vanish by the Koszul identities and the
	off-diagonal block vanishes because
	\(\partial_1\mathbf a=\mathbf a^{\mathsf T}\partial_3\).  Equivalently,
	\(\varphi_{X,\mathbf a}^\dagger\varphi_{Z,\mathbf a}=0\).

	If \((g_1,g_2)\in\ker\varphi_{X,\mathbf a}^\dagger\), Koszul exactness
	gives \(g_2=\partial_3h_3\) and
	\(g_1-\mathbf a h_3=\partial_2h_2\) for some \(h_2,h_3\).  Hence
	\begin{equation}
		(g_1,g_2)
		=(\partial_2h_2+\mathbf a h_3,\partial_3h_3)
		=\varphi_{Z,\mathbf a}(h_2,h_3).
	\end{equation}
	Dually, exactness of the conjugate Koszul complex gives
	\begin{equation}
		(g_1,g_2)
		=(\partial_1^\dagger h_0,
		  \partial_2^\dagger h_1-\overline{\mathbf a}h_0)
		=\varphi_{X,\mathbf a}(h_0,h_1)
	\end{equation}
	for every \((g_1,g_2)\in\ker\varphi_{Z,\mathbf a}^\dagger\).
	The reverse inclusions follow from commutativity, proving exactness in both
	sectors.
\end{proof}

\begin{remark}[Extension classes]
	Let \(S\coloneqq R/(\mathbf f)\) be the Koszul structure module of
	Definition~\ref{def:koszul-structure-module}.  If
	\([\mathbf a]=[\mathbf a'] \) in \( S^3\), then
	\(\mathbf a'-\mathbf a\in(\mathbf f)R^3\).  The corresponding
	upper-triangular chain isomorphism relabels the stabilizer generators at its
	endpoints and is implemented on the physical Pauli module by a
	translation-invariant CSS Clifford circuit.  Thus the two representatives
	need not define identical stabilizer subgroups in fixed Pauli coordinates,
	but they define Clifford-equivalent codes.  Suppressing the fixed arguments
	\((\mathbf f;\mathbb Z_N)\),
	\begin{equation}
		[\mathbf a]=[\mathbf a']
		\quad\Longrightarrow\quad
		\mathcal C_{\mathbf a}\sim_{R,\mathrm{Cl}}\mathcal C_{\mathbf a'}
		\quad\Longrightarrow\quad
		\mathcal C_{\mathbf a}\sim_{R,\mathrm{Cl},\mathrm{st}}\mathcal C_{\mathbf a'}.
	\end{equation}
	In short, \(\mathbf a\) specifies a concrete stabilizer presentation, but for
	studying phases within this family it is enough to retain its extension class
	\([\mathbf a]\in S^3\).
\end{remark}

The off-diagonal construction is flexible: compatible chain maps can couple
Koszul-complex models in different degrees or based on different sequences.
The present case \(\mathbf f=\mathbf g\) and \((i,j)=(1,2)\) is especially
simple, with extension classes \([\mathbf a]\in S^3\); more general couplings
need not admit such a single-vector parametrization.  Mapping cones and related
homological constructions also appear in studies of CSS-code
construction~\cite{Audoux2014}, stabilizer-weight
reduction~\cite{Hastings2021WeightReduction,Sabo2024WeightReduction},
decoding~\cite{Kubica2022SingleShot}, and fault-tolerant logical
measurements~\cite{Ide2025HomologicalMeasurement}.  Here we use the coupling to
study topological properties, particularly extension data invisible to the
individual \(X\)- and \(Z\)-sector layers.  A treatment in full generality is
left to future work.

\subsection{Common superselection layers}

Let
\(M_{\sigma,\mathbf a}(\mathbf f)\coloneqq
\operatorname{coker}\varphi_{\sigma,\mathbf a}\) for \(\sigma=X,Z\), and
retain the notation
\(S=R/(\mathbf f)\).

\begin{fact}[Common superselection layers in both Pauli sectors]
	\label{prop:block-extension-profile}
	For a fixed regular sequence \(\mathbf f\), the positive-degree layers of
	\(\mathcal C_{\mathbf a}(\mathbf f;\mathbb Z_N)\) are independent of
	\(\mathbf a\) in both Pauli sectors and are given by
	\begin{align}
		E_X^\ell
		&\cong_R\operatorname{Ext}_R^\ell
		\bigl(\overline{M_{X,\mathbf a}(\mathbf f)},R\bigr)
		\cong_R
		\begin{cases}
			S,&\ell=1,2,\\
			0,&\ell\geq3.
		\end{cases}
		\label{eq:block-extension-profile}\\
		E_Z^\ell
		&\cong_R\operatorname{Ext}_R^\ell
		\bigl(\overline{M_{Z,\mathbf a}(\mathbf f)},R\bigr)
		\cong_R
		\begin{cases}
			\overline S,&\ell=1,2,\\
			0,&\ell\geq3.
		\end{cases}
		\label{eq:block-extension-profile-Z}
	\end{align}
	Under these isomorphisms, translations act through the natural
	\(R\)-module structures of \(S\) and \(\overline S\), respectively.
\end{fact}

\begin{proof}
	The module \(M_{X,\mathbf a}(\mathbf f)\) admits the finite free resolution
	\begin{equation}
		0\longrightarrow K_0^\dagger
		\xrightarrow{\;\binom{0}{\partial_1^\dagger}\;}
		K_0^\dagger\oplus K_1^\dagger
		\xrightarrow{\;\varphi_{X,\mathbf a}\;}
		K_1^\dagger\oplus K_2^\dagger
		\longrightarrow M_{X,\mathbf a}(\mathbf f)\longrightarrow0.
		\label{eq:block-extension-resolution}
	\end{equation}
	Its exactness follows from Koszul exactness.  Applying
	\(\operatorname{Hom}_R(-,R)\) to the conjugated resolution gives
	\begin{equation}
		\operatorname{Ext}_R^2
		\bigl(\overline{M_{X,\mathbf a}(\mathbf f)},R\bigr)
		\cong_R R/\operatorname{im}\partial_1
		=R/(\mathbf f)=S.
	\end{equation}
	In degree one, Koszul exactness and
	\(\mathbf a^{\mathsf T}\partial_3=\partial_1\mathbf a\) reduce every
	cocycle to \((b,0)\), with \(b\) defined modulo
	\(\operatorname{im}\partial_1\); hence
	\(\operatorname{Ext}_R^1\cong_R S\).  The resolution has length two, so the
	higher Ext modules vanish.  These identifications are \(R\)-linear, giving
	both copies of \(S\) their natural translation action.

	For the \(Z\)-sector module, Koszul exactness similarly gives
	\begin{equation}
		0\longrightarrow K_3
		\xrightarrow{\;\binom{\partial_3}{0}\;}
		K_2\oplus K_3
		\xrightarrow{\;\varphi_{Z,\mathbf a}\;}
		K_1\oplus K_2
		\longrightarrow M_{Z,\mathbf a}(\mathbf f)\longrightarrow0.
		\label{eq:block-extension-resolution-Z}
	\end{equation}
	Indeed, \(\varphi_{Z,\mathbf a}(u,v)=0\) first implies
	\(\partial_3v=0\), hence \(v=0\), and then
	\(\partial_2u=0\), hence \(u=\partial_3w\).
	After conjugating and dualizing, the cochain maps are
	\begin{equation}
		K_1^\dagger\oplus K_2^\dagger
		\xrightarrow{\;\varphi_{Z,\mathbf a}^\dagger\;}
		K_2^\dagger\oplus K_3^\dagger
		\xrightarrow{\;(\partial_3^\dagger\;0)\;}
		K_3^\dagger.
	\end{equation}
	Thus the second layer is
	\(\operatorname{coker}\partial_3^\dagger\cong_R\overline S\).
	For the first layer, every cocycle has the form
	\((\partial_2^\dagger u,c)\) and is cohomologous to
	\((0,c-\overline{\mathbf a}^{\mathsf T}u)\).  This residual class is
	well defined modulo \(\operatorname{im}\partial_3^\dagger\), since the
	adjoint compatibility relation
	\(\overline{\mathbf a}^{\mathsf T}\partial_1^\dagger
	=\partial_3^\dagger\overline{\mathbf a}\) removes the ambiguity in \(u\).
	Conversely, every \((0,c)\) is a cocycle. Hence the first layer is
	\(\operatorname{coker}\partial_3^\dagger\cong_R\overline S\), independently of
	\(\mathbf a\); the resolution length gives the vanishing in higher degrees.
	These identifications are \(R\)-linear, so both copies of \(\overline S\)
	carry their natural translation action.
\end{proof}

\subsection{Eight qubit toric-code extensions: size-dependent GSD and translation enrichment}

We now specialize to qubits, \(N=2\), and to the toric-code sequence
\begin{equation}
	\mathbf f_{\mathrm{TC}}=(x_1-1,x_2-1,x_3-1).
\end{equation}
Then \(S=R/(\mathbf f_{\mathrm{TC}})\cong\mathbb Z_2\), so the extension
classes are represented by the eight constant vectors
\begin{equation}
	\mathbf a=\begin{pmatrix}a_1\\a_2\\a_3\end{pmatrix}
	\in\mathbb Z_2^3.
\end{equation}

\begin{figure}[t]
	\centering
	\includegraphics[width=\columnwidth]{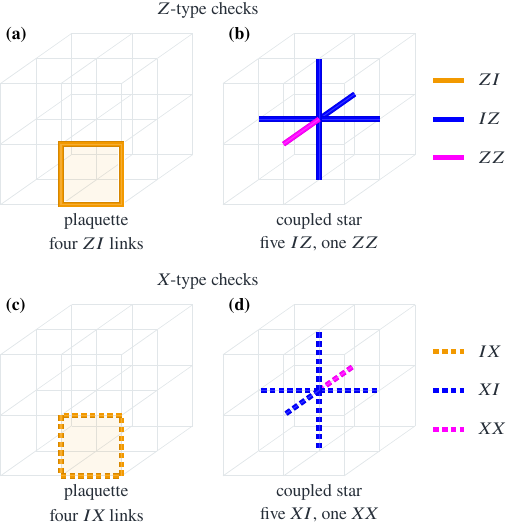}
	\caption{Representative checks of
		\(\mathrm{TC}_1^{(3)}(\mathbb Z_2)\mathbin{\oplus_{\mathbf a}}
		\mathrm{TC}_2^{(3)}(\mathbb Z_2)\), obtained by coupling two 3D toric
		codes along \(e_1\), with \(\mathbf a=e_1=(1,0,0)^{\mathsf T}\).
		With \(\mathrm{TC}_2\) represented on the dual lattice, each displayed
		link carries two qubits, ordered as
		\((\mathrm{TC}_1,\mathrm{TC}_2)\); \(ZI=Z\otimes I\), and similarly
		for the other labels. Solid (dashed) lines show \(Z\) (\(X\)) checks.
		The coupled factors \(ZZ\) in (b) and \(XX\) in (d) lie on the
		\(+e_1\) and \(-e_1\) arms, respectively.}
	\label{fig:degree-12-toric-extension-checks}
\end{figure}

\begin{definition}[Qubit degree-\((1,2)\) toric-code extensions]
	\label{def:degree-12-toric-extension}
	For \(\mathbf a\in\mathbb Z_2^3\), set
	\begin{equation}
		\mathcal C_{\mathbf a}
		\coloneqq
		\mathcal C_{\mathbf a}(\mathbf f_{\mathrm{TC}};\mathbb Z_2)
		=
		\mathrm{TC}_1^{(3)}(\mathbb Z_2)
		\mathbin{\oplus_{\mathbf a}}
		\mathrm{TC}_2^{(3)}(\mathbb Z_2).
		\label{eq:coupled-toric-family}
	\end{equation}
\end{definition}

Figure~\ref{fig:degree-12-toric-extension-checks} shows representative local
\(X\)- and \(Z\)-type checks for \(\mathbf a=e_1\), with
\(\mathrm{TC}_2\) represented on the dual lattice.

Fact~\ref{prop:block-extension-profile} shows that these eight models have
identical \(E_\sigma^\ell\) for \(\sigma=X,Z\) and \(\ell=1,2\). Since
\(S\cong\overline S\cong\mathbb Z_2\), translations act trivially on every
layer. In either Pauli sector, each model therefore has one nontrivial
pointlike class in the first layer and one nontrivial looplike class in the
second.
We now distinguish the models through their finite-size ground-state
degeneracies.

Place the models on a three-torus of size
\(\mathbf L=(L_1,L_2,L_3)\).  Define
\begin{equation}
	\rho_{\mathbf a}(\mathbf L)
	\coloneqq
	\begin{cases}
		0, & a_1L_2L_3=a_2L_3L_1=a_3L_1L_2=0 \text{ in }\mathbb{Z}_2,\\
		1, & \text{otherwise}.
	\end{cases}
	\label{eq:rho-E}
\end{equation}

\begin{fact}[Size-dependent ground-state degeneracy]
	\label{prop:GSD-eight-E}
	For the model \(\mathcal C_{\mathbf a}\), the number of logical qubits on
	the \(\mathbf L\)-torus is
	\begin{equation}
		\label{eq:GSD-E-formula}
		k_{\mathbf a}(\mathbf L)
		=\log_2\operatorname{GSD}_{\mathbf a}(\mathbf L)
		=6-2\rho_{\mathbf a}(\mathbf L).
	\end{equation}
	The ground-state degeneracy (GSD) of each model is summarized in
	Table~\ref{tab:parity-resolved-gsd}.
\end{fact}

\begin{proof}
	By Eq.~\eqref{eq:AE-extension}, the finite-torus \(Z\)-logical space is
	the first homology
	\begin{equation}
		\mathcal H_{Z,\mathbf a}(\mathbf L)
		\coloneqq H_1\bigl(\mathcal D_{\bullet,\mathbf a}\otimes_R R_{\mathbf L}\bigr)
		=\frac{\ker\bigl(\mathcal D_{1,\mathbf a}\otimes_R R_{\mathbf L}\bigr)}
		{\operatorname{im}\bigl(\mathcal D_{2,\mathbf a}\otimes_R R_{\mathbf L}\bigr)}.
	\end{equation}
	Below, all maps are understood over \(R_{\mathbf L}\).

	For \(\mathbf a=0\), the two toric-code blocks decouple.  Set
	\begin{equation}
		\begin{aligned}
			\ell_i&\coloneqq\lambda_i e_i, &
			\lambda_i&\coloneqq1+x_i+\cdots+x_i^{L_i-1},\\
			m_{jk}&\coloneqq\lambda_j\lambda_k e_{jk},
			&e_{jk}&\coloneqq e_j\wedge e_k,
		\end{aligned}
	\end{equation}
	corresponding to the familiar loop and membrane logicals of the 3D toric code. Then
	\begin{equation}
		\mathcal H_{Z,\mathbf0}(\mathbf L)
		=\operatorname{span}_{\mathbb Z_2}\{\ell_1,\ell_2,\ell_3\}
		\oplus\operatorname{span}_{\mathbb Z_2}\{m_{23},m_{31},m_{12}\}.
	\end{equation}
	Indeed, \((x_i-1)\lambda_i=x_i^{L_i}-1=0\) in
	\(R_{\mathbf L}\).

	Now turn on \(\mathbf a\). Then
	\((g_1,g_2)\in(K_1\oplus K_2)\otimes_R R_{\mathbf L}\)
	is a cycle precisely when
	\begin{equation}
		\partial_2g_2=0,
		\qquad
		\partial_1g_1=\mathbf a^{\mathsf T}g_2.
	\end{equation}
	Modulo the ordinary toric-code boundaries, write the membrane component as
	\begin{equation}
		[g_2]=c_1[m_{23}]+c_2[m_{31}]+c_3[m_{12}],
		\quad
		\mathbf c\coloneqq(c_1,c_2,c_3)^{\mathsf T}\in\mathbb Z_2^3.
	\end{equation}
	The equation for \(g_1\) has a solution precisely when
	\(\mathbf a^{\mathsf T}g_2\) is trivial in
	\(H_0\cong\mathbb Z_2\).  Since \(x_i=1\) in \(H_0\), direct substitution
	gives
	\begin{equation}
		[\mathbf a^{\mathsf T}g_2]
		=\mathbf v_{\mathbf a}(\mathbf L)^{\mathsf T}\mathbf c\,[1],
		\qquad
		\mathbf v_{\mathbf a}(\mathbf L)
		\coloneqq
		\begin{pmatrix}
			a_1L_2L_3\\
			a_2L_3L_1\\
			a_3L_1L_2
		\end{pmatrix}.
	\end{equation}
	Thus \(\mathbf c\) must solve
	\(\mathbf v_{\mathbf a}(\mathbf L)^{\mathsf T}\mathbf c=0\), leaving
	\(3-\rho_{\mathbf a}(\mathbf L)\) independent membrane combinations.
	For each allowed \([g_2]\), fix one solution \(g_1^{(0)}\).  The difference
	between any two solutions is a \(\partial_1\)-cycle, so
	\begin{equation}
		g_1\equiv g_1^{(0)}+\sum_{i=1}^3b_i\ell_i
		\pmod{\operatorname{im}\partial_2},
		\qquad b_i\in\mathbb Z_2.
	\end{equation}
	The remaining identification among the coefficients \(b_i\) comes from a
	coupled boundary with \(\partial_3h_3=0\).  Modulo ordinary boundaries, such
	an \(h_3\) is a multiple of the volume cycle
	\(\omega\coloneqq\lambda_1\lambda_2\lambda_3e_{123}\), whose boundary is
	\begin{equation}
		\begin{aligned}
			\mathcal D_{2,\mathbf a}(0,\omega)&=(\mathbf a\omega,0),\\
			[\mathbf a\omega]
			&=a_1L_2L_3[\ell_1]+a_2L_3L_1[\ell_2]
			  +a_3L_1L_2[\ell_3].
		\end{aligned}
	\end{equation}
	Thus \(\mathbf b\sim\mathbf b+\mathbf v_{\mathbf a}(\mathbf L)\), where
	\(\mathbf b=(b_1,b_2,b_3)^{\mathsf T}\).  Hence each allowed \([g_2]\) has
	\(3-\rho_{\mathbf a}(\mathbf L)\) independent \(g_1\) classes above it.
	Combining the two counts gives
	\begin{equation}
		k_{\mathbf a}(\mathbf L)
		=\dim_{\mathbb Z_2}\mathcal H_{Z,\mathbf a}(\mathbf L)
		=2\bigl(3-\rho_{\mathbf a}(\mathbf L)\bigr)
		=6-2\rho_{\mathbf a}(\mathbf L).
	\end{equation}
	Since \(\operatorname{GSD}_{\mathbf a}(\mathbf L)=2^{k_{\mathbf a}(\mathbf L)}\),
	the claim follows.
\end{proof}

\begin{table}[t]
	\centering
	\small
	\renewcommand{\arraystretch}{1.15}
	\begin{tabular}{@{}c@{\qquad}p{0.58\columnwidth}@{}}
		\toprule
		\(\mathbf a\) & Condition for
		\(\operatorname{GSD}_{\mathbf a}(\mathbf L)=2^4\) \\
		\midrule
		\((0,0,0)^{\mathsf T}\) & Never \\
		\((1,0,0)^{\mathsf T}\) & \(L_2\) and \(L_3\) are odd \\
		\((0,1,0)^{\mathsf T}\) & \(L_1\) and \(L_3\) are odd \\
		\((0,0,1)^{\mathsf T}\) & \(L_1\) and \(L_2\) are odd \\
		\((1,1,0)^{\mathsf T}\) & \((L_2,L_3)\) both odd, or \((L_1,L_3)\) both odd \\
		\((1,0,1)^{\mathsf T}\) & \((L_2,L_3)\) both odd, or \((L_1,L_2)\) both odd \\
		\((0,1,1)^{\mathsf T}\) & \((L_1,L_3)\) both odd, or \((L_1,L_2)\) both odd \\
		\((1,1,1)^{\mathsf T}\) & At least two of \(L_1,L_2,L_3\) are odd \\
		\bottomrule
	\end{tabular}
	\caption{Ground-state degeneracies of the eight qubit degree-\((1,2)\)
		toric-code extensions.  Listed are the parity classes with
		\(\operatorname{GSD}_{\mathbf a}(\mathbf L)=2^4\); all remaining classes
		have \(\operatorname{GSD}_{\mathbf a}(\mathbf L)=2^6\).}
	\label{tab:parity-resolved-gsd}
\end{table}

With the original translation symmetry fixed,
Table~\ref{tab:parity-resolved-gsd} shows
that the size-dependent GSD distinguishes all eight qubit models
\(\mathcal C_{\mathbf a}\).  In particular, every \(\mathbf a\neq0\) defines
a nonsplit extension, since a split extension would have
\(\operatorname{GSD}=2^6\) on every torus.  Models with different GSD
functions cannot be related by a translation-invariant FDLU, even after
adding or removing product-state ancillas, because such a circuit preserves
the ground-space dimension on every sufficiently large torus.

Size-dependent torus degeneracy is familiar from models in which translations
act nontrivially on the anyons, beginning with the even--odd effects in Wen's
plaquette model and extending to explicit translation-induced anyon
permutations~\cite{Wen2003QuantumOrders,RaoSodemann2021,Watanabe2023GSD,
RakovszkyKhemani2024Products,Chen2025PRL,He2026PRB}. The
mechanism here is different: translations act trivially on every individual
\(E_\sigma^\ell\), while the GSD detects the cross-degree extension data.

\begin{fact}[Uniform coarse splitting of qubit toric-code extensions]
	\label{fact:coarse-splitting}
	For every \(\mathbf a\in\mathbb Z_2^3\), uniform
	\(2\times2\times2\) blocking gives
	\begin{equation}
		\mathcal C_{\mathbf a}
		\sim_{R_{\Lambda},\mathrm{Cl}}
		\mathcal C_{\mathbf0},
		\qquad \Lambda=2\mathbb Z^3.
	\end{equation}
	Thus all eight translation SET orders share the intrinsic topological order
	of the decoupled stack
	\begin{equation}
		\mathcal C_{\mathbf0}
		=
		\mathrm{TC}^{(3)}_1(\mathbb Z_2)
		\oplus
		\mathrm{TC}^{(3)}_2(\mathbb Z_2).
	\end{equation}
\end{fact}

Appendix~\ref{app:coarse-splitting} constructs the Clifford circuit and
shows explicitly how uniform blocking removes every coupling. Thus the
individual layers \(E_\sigma^\ell\) in both Pauli sectors and their translation
actions do not determine the full SET order with the original translation
symmetry fixed:
the coupling \(\mathbf a\), invisible to those layers, remains detectable
through the size-dependent GSD. Fact~\ref{fact:coarse-splitting} also shows
that these SET distinctions disappear after uniform blocking and
are therefore not intrinsic.

\section{Discussion and outlook}
\label{sec:postnikov-outlook}

Regular Koszul-complex models connect infinite-lattice excitation structure
with finite-size logical information through the same structure module
\(S=R/(\mathbf f)\). The \(\operatorname{Ext}\) calculation determines their
superselection layers, including fusion and lattice translations, while
\(\operatorname{Tor}\) measures the arithmetic overlap of \(S\) with periodic
boundary conditions. For full-length regular sequences,
Corollary~\ref{cor:koszul-multistack-rigidity} gives a circuit-level normal
form: the mobility lattice of \(S\) specifies the exact retained translations
for a stable Clifford reduction to a toric-code stack. These models therefore
provide a computable setting in which excitation data control both circuit
equivalence and finite-size logical structure.

Together, single-degree rigidity
(Theorem~\ref{thm:topological-rigidity}) and realizability
(Theorem~\ref{thm:single-degree-realizability}) establish a classification
result for topological CSS codes with a single-degree profile within the
finite-resolution class, extending beyond the Koszul family. For a
fixed Pauli sector and degree \(j\geq1\), a finitely generated \(R\)-module
\(E\) is realizable precisely when \(\operatorname{grade}_R(E)>j\); within
the finite-resolution class, its realization is unique up to stable
translation-invariant Clifford equivalence. For prime qudits, the resolution
hypothesis is automatic, so these results give a complete
existence-and-uniqueness theory for single-degree profiles.
Our matrix-level Schanuel-lemma method supplies
uniqueness, while Construction~\ref{constr:single-degree-realization}
uses a general free resolution to produce a representative satisfying
the rigidity hypothesis.

For prime-qudit codes with such a single-degree profile and a finite nonzero layer,
Corollary~\ref{thm:minimal-coarse-graining} shows that the
minimal coarse graining is a property of the underlying translation SET
order, rather than of a particular stabilizer presentation.
Proposition~\ref{prop:companion-sector-layers} also determines every
companion \(Z\)-layer from the prescribed single-degree \(X\)-sector profile,
without requiring the \(Z\)-sector profile to be single-degree.

The eight three-dimensional qubit extensions in
Sec.~\ref{sec:non-splitting-extensions} expose the limit of this principle.
Their individual layers \(E_\sigma^\ell\), for \(\sigma=X,Z\) and
\(\ell=1,2\), agree and have trivial translation action, but their
ground-state-degeneracy functions distinguish them with the original
translations fixed. The missing information is carried by the cross-degree
gluing retained in the full superselection profile \(\mathscr E_\sigma\).
Uniform \(2\times2\times2\) blocking maps all eight to the same decoupled
toric-code stack by a finite-depth Clifford circuit
(Fact~\ref{fact:coarse-splitting}). Thus their differences are translation SET
distinctions within one intrinsic topological order, not eight different
intrinsic orders.

This determination problem with translations fixed complements recent
bulk--boundary and \(L\)-theoretic approaches~\cite{RubaYang2025Witt,
GeikoShuklin2025Invertible,YangYu2026Classification}, which use broader
equivalence relations involving condensation, gapped interfaces,
stabilization, and coarse graining. The trivial-charge example in
Appendix~\ref{app:pm-obstruction} shows why the finite-resolution hypothesis
cannot simply be dropped from the cyclic prime-power Clifford theorem.

Further exploration of the physics encoded by general superselection
profiles is an interesting direction for future work.
Lemma~\ref{lem:topological-grade-bound} constrains each layer,
and stacking codes from Construction~\ref{constr:single-degree-realization}
realizes finite collections of admissible layers as split profiles.
Prescribing nonsplit attachments requires additional compatibility data,
which finite Postnikov stages and
Maurer--Cartan equations organize. A related
physical problem is to identify observables that detect the attachments.
Separately, understanding which profile
data remain invariant under more general local-unitary equivalences is
needed to go beyond the Clifford setting. For cyclic prime-power qudits,
it remains open whether matching single-degree profiles imply stable
translation-invariant FDLU equivalence without the finite-resolution hypothesis.

\begin{acknowledgments}
This work is supported by the National Natural Science Foundation of China
(Grants No.~12522502, No.~12474145, and No.~12447101) and the Strategic
Priority Research Program of the Chinese Academy of Sciences (Grant
No.~XDB1680200).
\end{acknowledgments}

\textit{Note added.---}Near the completion of this work, we became aware of independent work by Wang \textit{et al.}~\cite{Wang2026Decoupling}. The two works overlap in establishing the same existence result for decoupling two-dimensional translation-invariant topological CSS codes on qubits under the maximal anyon-preserving translation symmetry. Their approach is constructive, providing an explicit polynomial-time algorithm to build the decoupling unitary, which directly complements our general rigidity result (Corollary~\ref{thm:minimal-coarse-graining}) based on Schanuel's lemma. We coordinated the timing of our arXiv postings and thank the authors for sharing their results prior to public release.

After the first arXiv version of this work appeared, Li and Williamson
established toric-stack equivalence under finite coarse graining for
qubit CSS codes built from three-term segments of length-$D$ free
resolutions whose $ZX$ duals are also
resolutions~\cite[Theorem~1]{LiWilliamson2026}.
A nice connection can be established between the two
works: applying algebraic techniques to their setup yields the
single-degree superselection profile $\mathscr E_X\simeq M[-n]$ ($M$
placed in cohomological degree $n$, with $M$ the resolved finite
module and $n$ the segment's central degree).

\appendix

\section{Extended excitation from cohomological descent}
\label{app:cohomological_descent}

The cohomological descent takes a particularly simple and intuitive form for a single-degree profile whose nonzero layer $\operatorname{Ext}_R^\ell (\overline{\operatorname{coker}d_1}, R)$ is a finite set. In this scenario, the directional mobility periods $\{J_i\}$ are all finite, allowing the use of binomials $\Delta_i \coloneqq x_i^{J_i} - 1$ to generate legitimate sweepings.

Let $e \in F_\ell^\dagger$ represent a nontrivial class $[e] \in \operatorname{Ext}_R^\ell$. We can algorithmically construct the corresponding physical excitation by sweeping $e$ along $\ell$ independent directions:
\begin{enumerate}
	\item \textbf{1D sweep:} Because $\Delta_1$ annihilates the class $[e]$, the translated dipole $\Delta_1 e$ is exact. Thus, there exists a 1D chain $c_1 \in F_{\ell-1}^\dagger$ such that
	\begin{equation}
		\Delta_1 e = d_\ell^\dagger c_1.
	\end{equation}

	\item \textbf{2D sweep:} Choose $c_2\in F_{\ell-1}^\dagger$ analogously,
	so that $d_\ell^\dagger c_2=\Delta_2e$. Antisymmetrizing in the
	orientation used below gives
	$d_\ell^\dagger(\Delta_1c_2-\Delta_2c_1)
	=\Delta_1\Delta_2e-\Delta_2\Delta_1e=0$.
	Because the intermediate superselection layer vanishes
	($\operatorname{Ext}_R^{\ell-1}=0$), this closed pattern is exact. Hence
	there exists a 2D sheet $c_{12} \in F_{\ell-2}^\dagger$ such that
	\begin{equation}
		\Delta_1 c_2 - \Delta_2 c_1 = d_{\ell-1}^\dagger c_{12}.
	\end{equation}

	\item \textbf{General descent:} We iterate this alternating sum for $k$ directions. At the $k^{\text{th}}$ step, the vanishing of $\operatorname{Ext}_R^{\ell-k+1}$ guarantees the existence of a $k$-dimensional pattern $c_{1\dots k} \in F_{\ell-k}^\dagger$ satisfying
	\begin{equation}
		d_{\ell-k+1}^\dagger c_{1\dots k} = \sum_{a=1}^k (-1)^{a-1} \Delta_a c_{1\dots \hat{a}\dots k},
	\end{equation}
	where the hat denotes the omission of the index $a$.
	At odd $N$, the alternating signs in this equation are essential.
\end{enumerate}

Terminating this descent at $k = \ell$ yields an element $c_{12\dots \ell} \in F_0^\dagger$. This element represents a Pauli operator $\xi$ supported on an $\ell$-dimensional hypercube. It generates an $(\ell-1)$-dimensional extended excitation along the boundary of that hypercube. By patching together $\xi$ and its translations, one can construct operators that generate arbitrarily large $(\ell-1)$-dimensional extended excitations. This explicit formula directly bridges the abstract algebraic class in $\operatorname{Ext}_R^\ell$ to concrete realizations of $(\ell-1)$-dimensional extended excitations in the code.

When $\operatorname{Ext}_R^\ell (\overline{\operatorname{coker}d_1}, R)$ is infinite, we may instead need to use sweepings that correspond to fractonic moves~\cite{Haah2011PRA,Yoshida2013,Vijay2016,Song2024PRL} or their higher-dimensional generalizations.

\section{Chinese remainder decomposition of cyclic-qudit stabilizer models}
\label{app:crt-stabilizer-decomposition}

This appendix supplies the physical counterpart of the ring decomposition in
Eq.~\eqref{eq:CRT}.

\begin{proposition}[Physical Chinese remainder decomposition]
	\label{prop:physical-crt-decomposition}
	Let $N=\prod_\alpha n_\alpha$, where the $n_\alpha$ are pairwise
	coprime. An onsite basis relabeling identifies the cyclic-$N$ Pauli module
	with the product of the cyclic-$n_\alpha$ Pauli modules, or equivalently
	identifies their Pauli operators modulo scalar phases. Under this
	identification, translation-invariant stabilizer models, local Clifford
	circuits, and stable Clifford equivalence decompose and recombine
	componentwise for any fixed translation subgroup.
\end{proposition}

\begin{proof}
	Set $N_\alpha\coloneqq N/n_\alpha$ and choose $u_\alpha$ such that
	$u_\alpha N_\alpha\equiv1\pmod{n_\alpha}$. The onsite unitary
	\begin{equation}
		U_{\mathrm{CRT}}\lvert z\rangle_N
		=
		\bigotimes_\alpha
		\lvert z\bmod n_\alpha\rangle_{n_\alpha}
		\label{eq:onsite-crt-unitary}
	\end{equation}
	is a basis relabeling by the Chinese Remainder Theorem. Writing
	$\omega_n\coloneqq\exp(2\pi\mathrm{i}/n)$, CRT reconstruction gives
	\begin{equation}
		\begin{aligned}
			z
			&\equiv
			\sum_\alpha N_\alpha u_\alpha(z\bmod n_\alpha)
			\pmod N,\\
			\omega_N^z
			&=
			\prod_\alpha
			\omega_{n_\alpha}^{u_\alpha(z\bmod n_\alpha)}.
		\end{aligned}
	\end{equation}
	Therefore
	\begin{equation}
		U_{\mathrm{CRT}}X_NU_{\mathrm{CRT}}^\dagger
		=
		\bigotimes_\alpha X_{n_\alpha},
		\qquad
		U_{\mathrm{CRT}}Z_NU_{\mathrm{CRT}}^\dagger
		=
		\bigotimes_\alpha Z_{n_\alpha}^{u_\alpha}.
		\label{eq:crt-pauli-factorization}
	\end{equation}
	Each $u_\alpha$ is a unit modulo $n_\alpha$, so these operators generate
	the full product of the component Pauli modules.

	Algebraically, the elements
	$e_\alpha\coloneqq N_\alpha u_\alpha\in\mathbb Z_N$ are orthogonal
	idempotents satisfying
	\begin{equation}
		e_\alpha e_\beta
		=
		\delta_{\alpha\beta}e_\alpha,
		\qquad
		\sum_\alpha e_\alpha=1.
	\end{equation}
	Hence every $\mathbb Z_N$-module $M$ splits as
	$M=\bigoplus_\alpha e_\alpha M$, and the same coefficientwise splitting
	applies over the Laurent ring. It does not change spatial support, so
	locality and translation invariance are preserved. Componentwise local
	Clifford circuits $V_\alpha$ recombine as
	\begin{equation}
		U_{\mathrm{CRT}}^\dagger
		\left(\bigotimes_\alpha V_\alpha\right)
		U_{\mathrm{CRT}}.
	\end{equation}
	Finally,
	\begin{equation}
		U_{\mathrm{CRT}}\lvert0\rangle_N
		=
		\bigotimes_\alpha\lvert0\rangle_{n_\alpha},
		\qquad
		U_{\mathrm{CRT}}\lvert+\rangle_N
		=
		\bigotimes_\alpha\lvert+\rangle_{n_\alpha}.
	\end{equation}
	After padding component sectors with trivial ancillas when necessary,
	stabilization therefore decomposes and recombines in the same way. Thus two
	models are stably equivalent by a circuit invariant under a fixed
	translation subgroup if and only if all their prime-power reductions are.
\end{proof}

\section{Proof of Proposition~\ref{prop:stable-module-matrix}}
\label{app:stable-module-matrix-proof}

In this appendix, we prove the $(2) \Rightarrow (1)$ direction of Proposition~\ref{prop:stable-module-matrix}: if two presentation matrices have freely stably \(R\)-isomorphic cokernels, they are freely stably \(R\)-equivalent. We first establish the result for the special case of strictly \(R\)-isomorphic cokernels.

\begin{lemma}[Stable \(R\)-equivalence from strict \(R\)-isomorphism]
	\label{lem:strict-iso-stable-eq}
	Let $\varphi \colon R^b \to R^q$ and $\varphi' \colon R^{b'} \to R^{q'}$ be presentation matrices. If $\operatorname{coker}\varphi \cong_R \operatorname{coker}\varphi'$, then $\varphi \sim_{R,\mathrm{st}} \varphi'$.
\end{lemma}

\begin{proof}
	Let $M \coloneqq \operatorname{coker}\varphi \cong_R \operatorname{coker}\varphi'$. Let $\{x_i\}_{i=1}^q$ and $\{y_j\}_{j=1}^{q'}$ be the images in $M$ of the standard bases of $R^q$ and $R^{q'}$, respectively. Because both sets generate $M$, there exist matrices $C \in M_{q \times q'}(R)$ and $D \in M_{q' \times q}(R)$ such that
	\begin{equation}
		\begin{split}
			(y_1,\dots,y_{q'}) &= (x_1,\dots,x_q)C, \\
			(x_1,\dots,x_q) &= (y_1,\dots,y_{q'})D.
		\end{split}
	\end{equation}

	Starting from $\varphi$, we adjoin the $\{y_j\}$ generators and their defining relations to obtain the matrix
	\begin{equation}
		\varphi_1 =
		\begin{pmatrix}
			\varphi & -C\\
			0 & I_{q'}
		\end{pmatrix}.
	\end{equation}
	By elementary row operations, $\varphi_1$ is strictly \(R\)-equivalent to $\varphi \oplus I_{q'}$. Symmetrically, adjoining the $\{x_i\}$ generators to $\varphi'$ yields
	\begin{equation}
		\varphi'_1 =
		\begin{pmatrix}
			I_q & 0\\
			-D & \varphi'
		\end{pmatrix},
	\end{equation}
	which is strictly \(R\)-equivalent to $I_q \oplus \varphi'$.

	By construction, $\varphi_1$ and $\varphi'_1$ present the same module $M$ with respect to the identical generating set $\{x_1,\dots,x_q, y_1,\dots,y_{q'}\}$. Therefore, they share the exact same column space in $R^{q+q'}$. Hence, there exist matrices $T$ and $S$ such that
	\begin{equation}
		\varphi'_1 = \varphi_1 T, \qquad \varphi_1 = \varphi'_1 S.
	\end{equation}

	Consider the block matrix $C_0 = \begin{pmatrix}\varphi_1 & \varphi'_1\end{pmatrix}$. Right multiplication by the invertible matrix
	\begin{equation}
		\begin{pmatrix}
			I & -T \\
			0 & I
		\end{pmatrix}
	\end{equation}
	transforms $C_0$ into $\begin{pmatrix}\varphi_1 & 0_{q+q', q+b'}\end{pmatrix}$. Thus, $C_0$ is strictly \(R\)-equivalent to $\varphi_1 \oplus 0_{0, q+b'}$ (where $0_{0, q+b'}$ denotes $q+b'$ trivial zero columns). Symmetrically, right multiplication utilizing $S$ via analogous column operations transforms $C_0$ into $\begin{pmatrix}\varphi'_1 & 0_{q+q', b+q'}\end{pmatrix}$, proving $C_0 \sim_R \varphi'_1 \oplus 0_{0, b+q'}$.

	By transitivity of strict \(R\)-equivalence, we have
	\begin{equation}
		\varphi_1 \oplus 0_{0, q+b'} \sim_R \varphi'_1 \oplus 0_{0, b+q'}.
	\end{equation}
	Substituting the equivalences $\varphi_1 \sim_R \varphi \oplus I_{q'}$ and $\varphi'_1 \sim_R \varphi' \oplus I_q$, we obtain
	\begin{equation}
		\varphi \oplus I_{q'} \oplus 0_{0, q+b'} \sim_R \varphi' \oplus I_q \oplus 0_{0, b+q'}.
	\end{equation}
	This is exactly \(\varphi \sim_{R,\mathrm{st}} \varphi'\).
\end{proof}

\begin{proof}[Proof of Proposition~\ref{prop:stable-module-matrix}, $(2) \Rightarrow (1)$]
	Assume that the two cokernels are freely stably \(R\)-isomorphic. By
	definition, there exist integers
	$\alpha,\alpha'\geq0$ such that
	\begin{equation}
		\operatorname{coker}\varphi \oplus R^\alpha \cong_R \operatorname{coker}\varphi' \oplus R^{\alpha'}.
	\end{equation}
	As noted in the main text, appending zero rows to a presentation matrix takes the direct sum of its cokernel with free summands. Thus, this relation translates to a strict \(R\)-isomorphism between the cokernels of extended matrices:
	\begin{equation}
		\operatorname{coker}(\varphi \oplus 0_{\alpha, 0}) \cong_R \operatorname{coker}(\varphi' \oplus 0_{\alpha', 0}).
	\end{equation}
	Applying Lemma~\ref{lem:strict-iso-stable-eq} to these extended matrices yields
	\begin{equation}
		\varphi \oplus 0_{\alpha, 0} \sim_{R,\mathrm{st}} \varphi' \oplus 0_{\alpha', 0}.
	\end{equation}
	Because appending zero blocks is a valid free stable \(R\)-equivalence operation, this reduces to $\varphi \sim_{R,\mathrm{st}} \varphi'$.
\end{proof}

\section{Schanuel's lemma and its generalization}
\label{app:schanuel-lemma}

This appendix reviews Schanuel's lemma and its generalization to higher
syzygies---the algebraic backbone of the rigidity theorems in the main text.

\begin{remark}[Free vs.\ projective stabilization]
	Although Schanuel's lemma and its generalizations hold fundamentally for
	projective modules over arbitrary rings~\cite[\S~5A]{Lam1999}, we restrict our
	exposition to finite‑rank free modules over a commutative ring $R$.
	Mathematically, this restriction is justified for the base ring
	$R=\mathbb Z_{p^m}[\mathbb Z^D]$ of translation‑invariant CSS codes.
	Reduction modulo the nilpotent ideal $pR$ gives the Laurent polynomial ring
	$\mathbb Z_p[\mathbb Z^D]$, over which every finitely generated projective
	module is free by the Laurent-polynomial extension of the Quillen--Suslin
	theorem~\cite{Swan1978}; nilpotent lifting gives the same conclusion over
	$R$. Physically, this
	allows us to classify codes using free stable \(R\)-equivalence, which
	corresponds exactly to the physical operation of adding unentangled product
	states.
\end{remark}

\begin{lemma}[Schanuel's lemma]
	\label{lem:base-schanuel}
	Let $M$ be an $R$-module and suppose we have two short exact sequences
	\begin{equation}
		0 \to K \xrightarrow{i} F \xrightarrow{\pi} M \to 0,\qquad
		0 \to K' \xrightarrow{i'} F' \xrightarrow{\pi'} M \to 0,
	\end{equation}
	where $F$ and $F'$ are free $R$-modules. Then there is an isomorphism
	\begin{equation}
		K \oplus F' \cong_R K' \oplus F .
	\end{equation}
\end{lemma}

\begin{proof}
	Form the pullback (fibre product) of the two surjections:
	\begin{equation}
		X \coloneqq \{\, (x,x') \in F \oplus F' \mid \pi(x) = \pi'(x') \,\}.
	\end{equation}
	The projection $p:X\to F$ is surjective because $\pi'$ is onto. Its kernel
	consists of pairs $(0,x')$ with $\pi'(x')=0$, so $\ker p \cong_R K'$. This
	gives $0 \to K' \to X \xrightarrow{p} F \to 0$, which splits because $F$
	is free. Hence $X \cong_R K' \oplus F$. Symmetrically, projecting onto $F'$
	yields $X \cong_R K \oplus F'$. Equating the two decompositions gives the
	result.
\end{proof}

For higher excitations we need the generalized Schanuel lemma, which asserts
that syzygies from different partial resolutions become freely stably
\(R\)-isomorphic.

\begin{lemma}[Generalized Schanuel's lemma for modules]
	\label{lem:generalized-schanuel-modules}
	Let $M$ be an $R$-module and suppose we have two partial free
	resolutions of length $n$:
	\begin{align}
		0 \to \Omega^n &\to F_{n-1} \to \cdots \to F_0 \to M \to 0,
		\label{eq:res1}\\
		0 \to \widetilde{\Omega}^n &\to G_{n-1} \to \cdots \to G_0 \to M \to 0,
		\label{eq:res2}
	\end{align}
	where all $F_i,G_i$ are finite‑rank free modules and
	$\Omega^n,\widetilde{\Omega}^n$ are the $n^{\text{th}}$ syzygies. Then
	$\Omega^n$ and $\widetilde{\Omega}^n$ are freely stably \(R\)-isomorphic:
	\begin{equation}
		\Omega^n \oplus G_{n-1} \oplus F_{n-2} \oplus \cdots
		\;\cong_R\;
		\widetilde{\Omega}^n \oplus F_{n-1} \oplus G_{n-2} \oplus \cdots .
	\end{equation}
\end{lemma}

\begin{proof}
	By induction on $n$. For $n=1$, Lemma~\ref{lem:base-schanuel} directly
	gives $\Omega^1 \oplus G_0 \cong_R \widetilde{\Omega}^1 \oplus F_0$, which
	is the required formula.

	Assume the statement for $n-1$. From the tails of~\eqref{eq:res1}
	and~\eqref{eq:res2} we obtain the short exact sequences
	\begin{align}
		0 \to \Omega^n &\to F_{n-1} \to \Omega^{n-1} \to 0, \label{eq:ses1}\\
		0 \to \widetilde{\Omega}^n &\to G_{n-1} \to \widetilde{\Omega}^{n-1}
		\to 0, \label{eq:ses2}
	\end{align}
	where $\Omega^{n-1} \coloneqq \operatorname{im}(F_{n-1}\to F_{n-2})$
	and $\widetilde{\Omega}^{n-1} \coloneqq \operatorname{im}(G_{n-1}\to G_{n-2})$.

	The truncated sequences
	\begin{equation}
		\begin{split}
			0 \to \Omega^{n-1} &\to F_{n-2} \to \cdots \to F_0 \to M \to 0,\\
			0 \to \widetilde{\Omega}^{n-1} &\to G_{n-2} \to \cdots \to G_0 \to M \to 0 .
		\end{split}
	\end{equation}
	are partial free resolutions of length $n-1$. By the induction hypothesis,
	there exist finite-rank free modules
	\begin{equation}
		P = G_{n-2} \oplus F_{n-3} \oplus \cdots,\quad
		Q = F_{n-2} \oplus G_{n-3} \oplus \cdots,
	\end{equation}
	such that $\Omega^{n-1} \oplus P \cong_R \widetilde{\Omega}^{n-1} \oplus Q$.

	Pad~\eqref{eq:ses1} and~\eqref{eq:ses2} with the trivial split extensions of
	$P$ and $Q$ to obtain
	\begin{align}
		0 \to \Omega^n &\to F_{n-1}\oplus P \to \Omega^{n-1}\oplus P \to 0,\\
		0 \to \widetilde{\Omega}^n &\to G_{n-1}\oplus Q
		\to \widetilde{\Omega}^{n-1}\oplus Q \to 0.
	\end{align}
	The rightmost modules are now \(R\)-isomorphic, so Lemma~\ref{lem:base-schanuel}
	yields
	\begin{equation}
		\Omega^n \oplus (G_{n-1}\oplus Q) \;\cong_R\;
		\widetilde{\Omega}^n \oplus (F_{n-1}\oplus P).
	\end{equation}
	Substituting $P$ and $Q$ recovers the alternating direct sum stated in
	the lemma.
\end{proof}

In the main text, Lemma~\ref{lem:generalized-schanuel-modules} guarantees that
the $(n-1)^{\text{st}}$ syzygies, corresponding to the cokernels of the leftmost presentation
matrices, are freely stably \(R\)-isomorphic. Proposition~\ref{prop:stable-module-matrix}
then translates this \(R\)-module isomorphism into the free stable \(R\)-equivalence of the
corresponding matrices, providing the algebraic bridge from macroscopic
excitation data back to the microscopic stabilizer maps.

\section{Grade bound, single-degree realizability, and companion-sector layers}
\label{app:profile-realization}

We work throughout over \(R=\mathbb Z_{p^m}[\mathbb Z^D]\), with \(p\)
prime and \(m\geq1\). We first establish the grade--dimension identity
in Eq.~\eqref{eq:grade-support-dimension}. We then prove
Lemma~\ref{lem:topological-grade-bound} by embedding the reduced Pauli
module into a finite free module, and recover the dimension bound in
Remark~\ref{rem:grade-geometry}. Next, we prove
Theorem~\ref{thm:single-degree-realizability} by verifying
Construction~\ref{constr:single-degree-realization}, including topological
exactness, the prescribed profile, and the finite free resolution of
\(\operatorname{coker}\varphi_X\). Finally, we prove
the companion-sector formula in
Proposition~\ref{prop:companion-sector-layers} for every realization,
without invoking Clifford rigidity.

\begin{proof}[Proof of Eq.~\eqref{eq:grade-support-dimension}]
	Let \(A=\mathbb Z[x_1^{\pm1},\ldots,x_D^{\pm1}]\). Since \(A\) is
	regular and \(p^m\) is a non-zero-divisor in \(A\),
	\(R=A/(p^m)\) is a \(D\)-dimensional Gorenstein ring and hence is
	Cohen--Macaulay~\cite[Sec.~21.3 and Cor.~21.19]{eisenbud1995commutative}.
	Thus the grade of a nonzero finitely generated
	\(R\)-module \(M\) is the minimum codimension of its support.
	Reduction modulo the nilpotent ideal \((p)\) identifies the prime
	spectrum of \(R\) with that of the Laurent polynomial ring over
	\(\mathbb F_p\), so
	\(\operatorname{ht}\mathfrak p+\dim(R/\mathfrak p)=D\) for every prime
	\(\mathfrak p\). Taking the minimum height and maximum quotient
	dimension over the support of \(M\) gives
	\(\operatorname{grade}_R(M)+\dim M=D\), proving
	Eq.~\eqref{eq:grade-support-dimension}.
\end{proof}

\begin{proof}[Proof of Lemma~\ref{lem:topological-grade-bound}]
	Topological exactness lets us shift the relevant Ext group up by one
	degree. By \(X/Z\) symmetry, take \(\sigma=X\) and set
	\(M_X\coloneqq\operatorname{coker}\varphi_X\).
	Conjugating the first identity in
	Eq.~\eqref{eq:topological-exactness-CSS} and using the canonical
	identifications for finite free modules gives
	\(\ker\varphi_Z^\vee=\operatorname{im}\overline{\varphi_X}\).
	Thus \(\varphi_Z^\vee\) embeds
	\(\overline{M_X}\cong_R\operatorname{im}\varphi_Z^\vee\)
	into the free module \(G_Z^\vee\).
	Identifying \(\overline{M_X}\) with this image and setting
	\(C\coloneqq G_Z^\vee/\overline{M_X}\), we obtain the short exact sequence
	\begin{equation}
		0\longrightarrow \overline{M_X}
		\longrightarrow G_Z^\vee\longrightarrow C
		\longrightarrow0.
		\label{eq:torsionless-embedding}
	\end{equation}
	Since the middle term is free, dimension shifting gives, for
	\(\ell\geq1\),
	\begin{equation}
		\operatorname{Ext}_R^\ell(\overline{M_X},R)
		\cong_R
		\operatorname{Ext}_R^{\ell+1}(C,R).
		\label{eq:grade-dimension-shift}
	\end{equation}

	Since \(R\) is Gorenstein, Ext groups over \(R_{\mathfrak p}\) with
	coefficients in \(R_{\mathfrak p}\) vanish in degrees greater than
	\(\operatorname{ht}\mathfrak p\)~\cite[Thm.~21.8 and Cor.~21.17]{eisenbud1995commutative}.
	At primes of height at most \(\ell\), the right-hand side of
	Eq.~\eqref{eq:grade-dimension-shift} therefore vanishes after
	localization. Thus
	\(\operatorname{Ext}_R^\ell(\overline{M_X},R)\) is supported only in
	codimension at least \(\ell+1\).

	By Eq.~\eqref{eq:grade-support-dimension} and the convention
	\(\operatorname{grade}_R(0)=\infty\), the support estimate gives
	\begin{equation}
		\operatorname{grade}_R
		\operatorname{Ext}_R^\ell(\overline{M_X},R)
		\geq \ell+1>\ell.
	\end{equation}
	Since \(E_X^\ell\cong_R\operatorname{Ext}_R^\ell(\overline{M_X},R)\), this
	proves Eq.~\eqref{eq:topological-grade-bound}.
	Together with Eq.~\eqref{eq:grade-support-dimension}, this yields the
	dimension bound~\eqref{eq:topological-dimension-bound} for nonzero layers;
	zero layers satisfy it by \(\dim0=-\infty\).
\end{proof}

\begin{proof}[Proof of Theorem~\ref{thm:single-degree-realizability} (Construction~\ref{constr:single-degree-realization})]
	Necessity follows from Lemma~\ref{lem:topological-grade-bound} at
	\(\ell=j\). Conversely, assume \(\operatorname{grade}_R(E)>j\) and use
	Construction~\ref{constr:single-degree-realization}. We check topological
	exactness and then compute the profile. The grade condition gives
	\(\operatorname{Ext}_R^i(E,R)=0\) for \(0\leq i\leq j\), so the
	\(R\)-dual complex is exact at
	\(G_0^\vee,\ldots,G_j^\vee\). Conjugation preserves this exactness in the
	dagger-dual complex. Since dagger duality is involutive on finite free
	modules, exactness of the original resolution and of its dagger dual
	at degree \(j\) gives
	\begin{equation}
		\begin{aligned}
			\operatorname{im}\varphi_Z
			&=\operatorname{im}\partial_{j+1}
			=\ker\partial_j
			=\ker\varphi_X^\dagger,\\
			\operatorname{im}\varphi_X
			&=\operatorname{im}\partial_j^\dagger
			=\ker\partial_{j+1}^\dagger
			=\ker\varphi_Z^\dagger.
		\end{aligned}
	\end{equation}
	These are precisely the two topological-exactness conditions, so
	\((\varphi_X,\varphi_Z)\) defines a topological code.

	To compute the \(X\)-sector profile, write
	\(M_X\coloneqq\operatorname{coker}\varphi_X\), so
	\(\overline{M_X}\cong_R\operatorname{coker}\partial_j^\vee\).
	Although the chosen resolution of \(E\) may be infinite, its dual is
	exact through degree \(j\), giving the finite free resolution
	\begin{equation}
		0\longrightarrow G_0^\vee\longrightarrow G_1^\vee
		\longrightarrow\cdots\longrightarrow G_j^\vee
		\longrightarrow \overline{M_X}\longrightarrow0.
		\label{eq:constructed-module-resolution}
	\end{equation}
	Dualizing this finite resolution recovers the segment
	\(G_j\to\cdots\to G_0\) of the original resolution of \(E\).
	Its cohomology therefore gives
	\begin{equation}
		\operatorname{Ext}_R^\ell(\overline{M_X},R)
		\cong_R
		\begin{cases}
			E, & \ell=j,\\
			0, & \ell>0,\ \ell\neq j.
		\end{cases}
	\end{equation}
	Thus the constructed code belongs to \(\mathrm{CSS}_j^X(E)\), proving the
	reverse implication in Eq.~\eqref{eq:realizability-grade}.

	Conjugating Eq.~\eqref{eq:constructed-module-resolution} gives a free
	resolution of \(M_X\) of length at most \(j\). If \(E\neq0\), then
	\(\operatorname{Ext}_R^j(\overline{M_X},R)\neq0\), so this length cannot
	be reduced. Since conjugation preserves projective dimension,
	\(M_X\) has projective dimension exactly \(j\), as claimed.
\end{proof}

The preceding proof constructs a realization. We now determine the
companion layers for an arbitrary code in \(\mathrm{CSS}_j^X(E)\).
Topological exactness relates the \(Z\)-sector module to the dual of the
\(X\)-sector module; biduality then expresses its layers in terms of \(E\).

\begin{proof}[Proof of Proposition~\ref{prop:companion-sector-layers}]
	Set
	\(M=\overline{\operatorname{coker}\varphi_X}\) and
	\(N=\operatorname{coker}\varphi_Z\), and write
	\((-)^\vee=\operatorname{Hom}_R(-,R)\).
	The quotient map \(q\colon P_Z^\vee=\overline{P_X}\twoheadrightarrow M\)
	dualizes to an inclusion whose image, by topological exactness, is
	\(\ker\varphi_X^\dagger=\operatorname{im}\varphi_Z\). Hence
	\begin{equation}
		0\longrightarrow M^\vee\xrightarrow{q^\vee}P_Z
		\longrightarrow N\longrightarrow0.
		\label{eq:companion-exact-sequence}
	\end{equation}
	The dual complex \(\mathbf R\!\operatorname{Hom}_R(M,R)\) has cohomology
	\(M^\vee\) in degree zero and, by the single-degree assumption,
	\(E\) in degree \(j\), with all other groups zero.
	Since \(R\) is Gorenstein of finite Krull dimension, dualizing this
	complex once more recovers \(M\) by derived
	biduality~\cite[Ex.~13.3.13 and Thm.~13.1.18]{Yekutieli2019DerivedCategories}.
	No finite-projective-dimension assumption is needed.
	Let \(\mathrm{ev}_M\colon M\to M^{\vee\vee}\) denote the natural
	evaluation map. Cohomological truncation separates the degree-zero
	part \(M^\vee\) from the positive-degree part \(E[-j]\).
	Dualizing the resulting exact triangle and taking cohomology gives
	\begin{equation}
		\operatorname{coker}(\mathrm{ev}_M)
		\cong_R\operatorname{Ext}_R^{j+1}(E,R),
	\end{equation}
	and, for \(k\geq1\),
	\begin{equation}
		\operatorname{Ext}_R^k(M^\vee,R)
		\cong_R\operatorname{Ext}_R^{j+k+1}(E,R).
	\end{equation}
	We now apply the long exact Ext sequence to
	Eq.~\eqref{eq:companion-exact-sequence}. Since \(P_Z\) is free, it gives,
	for \(\ell\geq2\),
	\begin{equation}
		\operatorname{Ext}_R^\ell(N,R)
		\cong_R\operatorname{Ext}_R^{\ell-1}(M^\vee,R)
		\cong_R\operatorname{Ext}_R^{j+\ell}(E,R).
	\end{equation}
	For \(\ell=1\), the same sequence identifies
	\(\operatorname{Ext}_R^1(N,R)\) with the cokernel of
	\(P_Z^\vee\to M^{\vee\vee}\). This map is \(\mathrm{ev}_M\circ q\).
	Since \(q\) is surjective, its composition with \(\mathrm{ev}_M\) has
	the same cokernel as \(\mathrm{ev}_M\), namely
	\(\operatorname{Ext}_R^{j+1}(E,R)\). The formula therefore holds for
	all \(\ell\geq1\).
	Finally, conjugation commutes with ordinary Ext, so
	\begin{equation}
		E_Z^\ell
		=\operatorname{Ext}_R^\ell(\overline N,R)
		\cong_R\overline{\operatorname{Ext}_R^\ell(N,R)}
		\cong_R\operatorname{Ext}_R^{j+\ell}(\overline E,R).
	\end{equation}
	This is Eq.~\eqref{eq:companion-sector-layers}, completing the proof.
\end{proof}

\section{Stable Clifford--FDLU separation for prime-power qudits}
\label{app:pm-obstruction}

The physical equivalence relation defining gapped topological phases can be characterized
by finite-depth local unitaries (FDLUs). For algebraic simplicity, the main
text restricts to translation-invariant finite-depth Clifford circuits. For
\(\mathbb Z_{p^m}\) qudits with \(m>1\), Clifford equivalence can be strictly
finer than FDLU equivalence. The following one-dimensional CSS example
demonstrates this distinction: it lies in the trivial FDLU phase
and has vanishing positive-degree \(X\)-sector Ext, yet cannot be disentangled
by any finite-depth Clifford circuit, even after adding or removing stabilizer
product-state ancillas and dropping translation invariance. Thus Clifford circuits are too restrictive to identify all topological states belonging to the same phase.

\begin{example}[An FDLU-trivial but Clifford-nontrivial \(\mathbb Z_{p^2}\) CSS chain]
	\label{prop:minimal-filtered-cluster}
	Let $p$ be prime, set
	\begin{equation}
		R=\mathbb Z_{p^2}[x^{\pm1}],
		\qquad
		t=x-1,
		\qquad
		\overline{t}=x^{-1}-1=-x^{-1}t,
	\end{equation}
	and consider
	\begin{equation}
		\varphi_X=
		\begin{pmatrix}
			p & \overline{t}\\
			0 & p
		\end{pmatrix},
		\qquad
		\varphi_Z=
		\begin{pmatrix}
			p & 0\\
			-t & p
		\end{pmatrix}.
	\end{equation}

	\textit{Exactness and Ext vanishing.}
	Direct calculation gives
	\begin{align}
		\ker\varphi_X^\dagger
		&=
		\bigl\{(pc,-tc+pd):c,d\in R\bigr\}
		=
		\operatorname{im}\varphi_Z,\\
		\ker\varphi_Z^\dagger
		&=
		\bigl\{(\overline{t}c+pd,pc):c,d\in R\bigr\}
		=
		\operatorname{im}\varphi_X.
	\end{align}
	Hence the stabilizer module is Lagrangian; the same calculation after
	passing to any finite periodic quotient shows that the model defines a
	stabilizer state.

	Set $M_X\coloneqq\operatorname{coker}\varphi_X$ and define
	\begin{equation}
		\alpha\coloneqq\overline{\varphi_X}
		=
		\begin{pmatrix}
			p & t\\
			0 & p
		\end{pmatrix},
		\qquad
		\beta\coloneqq\varphi_Z^\vee
		=
		\begin{pmatrix}
			p & -t\\
			0 & p
		\end{pmatrix}.
	\end{equation}
	The matrix \(\alpha\) presents \(\overline{M_X}\). The identities
	\begin{equation}
		\ker\alpha=\operatorname{im}\beta,
		\qquad
		\ker\beta=\operatorname{im}\alpha
	\end{equation}
	give the two-periodic free resolution
	\begin{equation}
		\cdots\longrightarrow R^2
		\xrightarrow{\alpha}R^2
		\xrightarrow{\beta}R^2
		\xrightarrow{\alpha}R^2
		\longrightarrow \overline{M_X}\longrightarrow0.
	\end{equation}
	The dual matrices likewise satisfy
	\begin{equation}
		\ker\alpha^\vee=\operatorname{im}\beta^\vee,
		\qquad
		\ker\beta^\vee=\operatorname{im}\alpha^\vee,
	\end{equation}
	so the dual complex is exact.
	Consequently,
	\begin{equation}
		\operatorname{Ext}_R^\ell(\overline{M_X},R)=0
		\qquad(\ell>0).
	\end{equation}
	This is consistent with the absence of long-range-entangled
	topological phases in one dimension: the model has no nontrivial bulk
	topological superselection sectors. The obstruction below therefore
	reflects the restriction to Clifford circuits rather than intrinsic
	topological order.

	\textit{No finite-length free resolution.}
	Let $S=R/(p)=\mathbb Z_p[x^{\pm1}]$. Reduction of the presentation
	modulo $p$ gives
	\begin{equation}
		\overline{M_X}\otimes_R S\cong_S S\oplus S/(t),
	\end{equation}
	which is not projective over $S$. Thus \(\overline{M_X}\) is not
	projective over $R$.
	In the two-periodic resolution above, the kernel at every second stage
	is isomorphic to \(\overline{M_X}\). If \(\overline{M_X}\) had finite
	projective dimension, a sufficiently high such kernel would be projective,
	forcing \(\overline{M_X}\) itself to be projective. Therefore
	\(\overline{M_X}\), and hence also \(M_X\), admits no finite-length free
	resolution.

	\textit{Clifford obstruction even under coarse graining and nonprimitive
	stabilization.}
	Let
	\begin{equation}
		P=R_X^2\oplus R_Z^2,
		\qquad
		\mathcal L=\operatorname{im}\varphi_X
		\oplus\operatorname{im}\varphi_Z
	\end{equation}
	be the Pauli module and its stabilizer submodule.

	Here stabilization allows a uniformly bounded number of decoupled
	$\mathbb Z_{p^2}$ qudits per site, each prepared in an arbitrary pure
	on-site stabilizer state; the ancillary types may vary from site to
	site. Besides primitive states stabilized by $\langle Z\rangle$ or
	$\langle X\rangle$, this includes the nonprimitive Lagrangian
	stabilizer
	\begin{equation}
		\langle X^p,Z^p\rangle,
		\qquad
		\lvert\chi\rangle
		=
		\frac{1}{\sqrt p}\sum_{a=0}^{p-1}\lvert pa\rangle.
	\end{equation}
	Although neither generator is primitive, their common \(+1\)
	eigenspace is one-dimensional and is spanned by $\lvert\chi\rangle$.
	Thus the obstruction below is stronger than the Clifford distinction
	between primitive and nonprimitive on-site stabilizer states: it
	persists even when both are freely available as stabilizing ancillas.
	More generally, the same obstruction persists under stabilization by
	on-site stabilizer states of uniformly bounded cyclic dimensions,
	provided Clifford circuits normalize the resulting fixed
	mixed-dimensional Pauli group.

	Modulo $p$,
	\begin{equation}
		(P/\mathcal L)\otimes_R S
		\cong_S
		S^2\oplus S/(\overline{t})\oplus S/(t)
		\cong_S
		S^2\oplus\bigl[S/(t)\bigr]^2.
	\end{equation}
	For a blocking of $b$ sites, set
	\begin{equation}
		y=x^b,
		\qquad
		T=\mathbb Z_p[y^{\pm1}].
	\end{equation}
	As \(S\) is free of rank \(b\) over \(T\) and
	\(S/(t)\cong_T T/(y-1)\), restriction of scalars gives
	\begin{equation}
		(P/\mathcal L)\otimes_R S
		\cong_T
		T^{2b}\oplus\bigl[T/(y-1)\bigr]^2.
	\end{equation}
	The torsion summands persist for every $b\geq1$. A
	translation-invariant Clifford circuit for the blocked lattice induces
	a $T$-linear symplectic automorphism and therefore preserves this
	quotient, whereas an on-site product stabilizer state gives a free
	$T$-module after reduction modulo $p$. Any periodic collection of the
	stabilizer ancillas described above likewise adds only free
	$T$-summands.
	Hence no Clifford circuit that is periodic after finite blocking can
	disentangle the state, even stably.

	\textit{Clifford nontriviality without translation invariance.}
	The preceding obstruction also excludes translation-breaking Clifford
	disentanglers. Suppose that a finite-depth, finite-range Clifford circuit
	disentangles the chain after adjoining product-state stabilizer ancillas
	whose number and cyclic dimensions per site are uniformly bounded. After
	padding unused ancilla slots, its local gates, ancilla states, and output
	product stabilizers take values in a finite alphabet, modulo phases.
	Circuit legality and disentangling impose the same finite-window
	constraints at every site.
	In any valid bi-infinite encoding, some boundary window must repeat;
	repeating the intervening segment therefore produces a periodic valid
	encoding. Hence the assumed circuit implies a Clifford disentangler
	that is translation invariant after finite blocking.

	On any sufficiently large periodic quotient compatible with this
	period, the image of the input stabilizer group is contained in the
	output stabilizer group. Both are Lagrangian on the same padded Hilbert
	space, so the inclusion is an equality. This contradicts the torsion
	obstruction above. Therefore no finite-depth, finite-range Clifford
	circuit can disentangle the state, even stably and without assuming
	translation symmetry of the circuit, ancillas, or output product state.

	\textit{Translation-invariant non-Clifford disentangling.}
	On each $\mathbb Z_{p^2}$ qudit, use the on-site factorization
	\begin{equation}
		D\lvert a+pb\rangle
		=
		\lvert a\rangle\otimes\lvert b\rangle,
		\qquad a,b\in\mathbb Z_p.
	\end{equation}
	The checks $X_1^p$ and $Z_2^p$ freeze one factor of each of the two
	qudits per cell. On the remaining factors,
	$(X_1,Z_1^p)$ and $(X_2^p,Z_2)$ act as ordinary
	$\mathbb Z_p$-qudit Pauli pairs, and the surviving stabilizer maps are
	\begin{equation}
		\psi_X=
		\begin{pmatrix}
			\overline{t}\\
			1
		\end{pmatrix},
		\qquad
		\psi_Z=
		\begin{pmatrix}
			1\\
			-t
		\end{pmatrix}.
	\end{equation}
	The effective $\mathbb Z_p$-qudit Clifford circuit
	\begin{equation}
		W_X=
		\begin{pmatrix}
			1 & -\overline{t}\\
			0 & 1
		\end{pmatrix},
		\qquad
		W_Z=(W_X^\dagger)^{-1}
		=
		\begin{pmatrix}
			1 & 0\\
			t & 1
		\end{pmatrix}
	\end{equation}
	sends $\psi_X$ and $\psi_Z$ to $(0,1)^\mathsf T$ and
	$(1,0)^\mathsf T$, respectively. Since
	$-\overline{t}=1-x^{-1}$, this circuit consists of two
	translation-invariant layers of CNOT gates on the effective
	$\mathbb Z_p$ qudits. Conjugating it by the on-site map $D$ gives a
	translation-invariant FDLU that disentangles the original
	$\mathbb Z_{p^2}$ chain. As shown above, this FDLU is not Clifford with
	respect to the original $\mathbb Z_{p^2}$ Pauli operators. No coarse
	graining is needed.
\end{example}

Thus the hypothesis of a finite-length free resolution in
Theorem~\ref{thm:topological-rigidity} identifies the sector in which the
present Clifford-based argument applies. Viewed as a $p^2$-torsion
$\mathbb Z[x^{\pm1}]$-system, this chain also lies in the invertible, or
trivial-charge, stabilizer sector studied by Geiko and
Shuklin~\cite{GeikoShuklin2025Invertible}. Their final one-dimensional
relative-$L$-theory quotient, which allows stabilization and condensation
beyond fixed-Pauli Clifford moves and further quotients lower-dimensional
subsystems, is trivial, consistent with the FDLU disentangler above. The
present example nevertheless retains a finer
obstruction under stable Clifford circuits that preserve the original
$\mathbb Z_{p^2}$ Pauli structure. The relation of their algebraic quotient
to finite-depth-circuit equivalence remains conjectural. For simplicity, this
paper restricts to Clifford connectivity and leaves a general FDLU treatment
over $\mathbb Z_{p^m}$ to future work.

\section{Uniform coarse splitting of qubit toric-code extensions}
\label{app:coarse-splitting}

\begin{proof}[Proof of Fact~\ref{fact:coarse-splitting}]
	Set \(\Lambda=2\mathbb Z^3\) and
	\(R_\Lambda=\mathbb Z_2[\Lambda]\). As an \(R_\Lambda\)-module,
	\begin{equation}
		R=\bigoplus_{\boldsymbol\epsilon\in\{0,1\}^3}
		R_\Lambda\mathbf x^{\boldsymbol\epsilon},
		\qquad
		\mathbf x^{\boldsymbol\epsilon}
		=x_1^{\epsilon_1}x_2^{\epsilon_2}x_3^{\epsilon_3}.
	\end{equation}
	Define the block-local \(R_\Lambda\)-linear map
	\(\mathsf T_j\) by
	\begin{equation}
		\mathsf T_j(r\mathbf x^{\boldsymbol\epsilon})
		=
		\begin{cases}
			r x_j^{-1}\mathbf x^{\boldsymbol\epsilon},&\epsilon_j=1,\\
			0,&\epsilon_j=0,
		\end{cases}
		\qquad r\in R_\Lambda.
	\end{equation}
	For \(f_\ell=x_\ell-1\), it satisfies
	\begin{equation}
		\mathsf T_jf_j+f_j\mathsf T_j=\operatorname{id},
		\qquad
		[\mathsf T_j,f_\ell]=0\quad(\ell\neq j).
	\end{equation}

	Use the cyclic basis
	\(e_{\hat1}=e_{23}\), \(e_{\hat2}=e_{31}\),
	\(e_{\hat3}=e_{12}\). For each cyclic triple
	\((i,j,k)=(1,2,3),(2,3,1),(3,1,2)\), define \(R_\Lambda\)-linear maps
	\(G_i\colon K_3\to K_2\),
	\(H_i\colon K_2\to K_1\), and
	\(Q_i\colon K_1\to K_0\) by their nonzero actions
	\begin{equation}
		\begin{aligned}
			G_i(r e_{123})
			&=\mathsf T_j(r)e_{\hat{k}},\\
			H_i(r e_{\hat{j}})
			&=\mathsf T_j(r)e_i,
			&
			H_i(r e_{\hat{i}})
			&=\mathsf T_j(r)e_j,\\
			Q_i(r e_k)&=\mathsf T_j(r),
		\end{aligned}
	\end{equation}
	where \(r\in R\). Direct substitution into
	Eq.~\eqref{eq:3d-differentials} gives
	\begin{equation}
		\partial_2G_i+H_i\partial_3=\iota_i,
		\qquad
		Q_i\partial_2+\partial_1H_i=\pi_i,
	\end{equation}
	where \(\iota_i(e_{123})=e_i\) and \(\pi_i(e_{\hat i})=1\), with all
	other basis elements sent to zero.
	For \(\mathbf a=(a_1,a_2,a_3)^{\mathsf T}\), set
	\(G=\sum_i a_iG_i\), \(H=\sum_i a_iH_i\), and
	\(Q=\sum_i a_iQ_i\). Then
	\begin{equation}
		\mathbf a=\partial_2G+H\partial_3,
		\qquad
		\mathbf a^{\mathsf T}=Q\partial_2+\partial_1H.
	\end{equation}
	Hence the upper-triangular changes of basis
	\begin{equation}
		V_2=\begin{pmatrix}1&G\\0&1\end{pmatrix},
		\qquad
		V_1=\begin{pmatrix}1&H\\0&1\end{pmatrix},
		\qquad
		V_0=\begin{pmatrix}1&Q\\0&1\end{pmatrix}
	\end{equation}
	satisfy
	\begin{equation}
		V_1\mathcal D_{2,\mathbf a}
		=\mathcal D_{2,\mathbf0}V_2,
		\qquad
		V_0\mathcal D_{1,\mathbf a}
		=\mathcal D_{1,\mathbf0}V_1.
	\end{equation}
	The maps \(V_0\) and \(V_2\) relabel stabilizer generators, while the
	cell-local physical shear \(V_1\) is implemented by a finite-depth
	\(\Lambda\)-translation-invariant CNOT circuit. Thus
	\(\mathcal C_{\mathbf a}\sim_{R_{\Lambda},\mathrm{Cl}}
	\mathcal C_{\mathbf0}\), without adding ancillas.
	Since blocking changes only the unit cell, the split model retains one
	degree-1 and one degree-2 toric-code factor, proving the intrinsic-order
	statement as well.
\end{proof}

\nocite{Ketkar2006FiniteFields,Heller1960LoopSpace}
\bibliography{bibliography/references}
\end{document}